\documentclass[onecolumn, lettersize,journal]{IEEEtran}
\usepackage[utf8]{inputenc}

\title{
Ising Model on Locally Tree-like Graphs: Uniqueness of Solutions to Cavity Equations}
\author{Qian Yu$^{*}$ and Yury Polyanskiy$^{\dagger}$\\
$^{*}$Department of Electrical and Computer Engineering, Princeton University 
\\
$^{\dagger}$Department of Electrical Engineering and Computer Science, Massachusetts Institute of Technology
}
\date{June 2020}

\usepackage{graphicx}
\usepackage{hyperref}
\usepackage{amsmath}
\usepackage{amssymb}
\usepackage{cite}

\usepackage{mathtools}

\usepackage{amsthm}
\usepackage{bbm}
\usepackage{stackengine} 
\newcommand\boxconv{\stackMath\mathbin{\stackinset{c}{0ex}{c}{0ex}{\ast}{\Box}}}

\def\esssup{\mathrm{esssup}}

\def\supp{\textup{supp}}

\usepackage{subcaption}

\theoremstyle{definition}
\newtheorem{defn}{Definition} 

\theoremstyle{plain}
\newtheorem{thm}[defn]{Theorem} 
\newtheorem*{thm*}{Theorem}

\theoremstyle{plain}
\newtheorem{prop}[defn]{Proposition}

\newtheorem{corollary}[defn]{Corollary}

\theoremstyle{remark}
\newtheorem{remark}{Remark} 

\def\EE{\mathbb{E}}
\def\PP{\mathbb{P}}

\def\calQ{\mathcal{Q}}
\def\calQs{\calQ_{\textup{s}}}
\def\mus{\mu_{\textup{s}}}

\def\mreals{\mathbb{R}}
\def\simiid{\stackrel{iid}{\sim}}
\def\Ber{\mathrm{Ber}}
\def\eqdef{\triangleq}

\newif\ifgeneral
\generaltrue

\usepackage{xcolor}

\begin{document}

\maketitle

\begin{abstract}
   In the study of Ising models on large locally tree-like graphs, in both rigorous and non-rigorous
methods one is often led to understanding the so-called belief propagation distributional
recursions and its fixed points. We prove that there is at most one non-trivial fixed point for Ising
models with zero or certain random external fields.
Previously this was only known for sufficiently
``low-temperature'' models. Our main innovation is in applying information-theoretic ideas of channel comparison
leading to a new metric (degradation index) between binary-input-symmetric (BMS) channels under which the Belief Propagation (BP)
operator is a strict contraction (albeit non-multiplicative). A key ingredient of our proof is a
strengthening of the classical stringy tree lemma of~\cite{EPKS2000}. 
   
   Our result simultaneously closes the following 6 conjectures in the literature:
1) independence of robust reconstruction accuracy to leaf noise in broadcasting on
trees~\cite{MNS16}; 
2) uselessness of global information for a labeled 2-community stochastic block model, or
2-SBM~\cite{kms16}; 
3) optimality of local algorithms for 2-SBM under noisy side information~\cite{mx15};
  4) uniqueness of BP fixed point in broadcasting on trees in the Gaussian (large degree)
  limit~\cite{mx15};
5) boundary irrelevance in broadcasting on trees~\cite{ACGP21};
6) characterization of entropy (and mutual information) of community labels given the graph in 2-SBM~\cite{ACGP21}.


   

\end{abstract}

\tableofcontents


\section{Main result and motivation}\label{sec:intro}

The central object of interest in this paper is a belief propagation
(BP) operator $\calQ_s$ that takes a symmetric distribution $\mu$ and produces another symmetric distribution
$\calQ_s \mu$. 
We 
call a probability distribution $\mu$ on $(-\infty,+\infty]$ \textit{symmetric} if 
 	\begin{equation}\label{eq:def_sym}
 	d\mu(r)= e^{r}
	d\mu(-r),  \qquad \iff \qquad \mu[E] = \int e^{-r} \mathbbm{1}\{-r \in E\} d\mu(r)
\end{equation} 
for every measurable subset $E\subseteq (-\infty,+\infty)$ (see \cite[Section. 15.2.2]{mezard2009information} or Section \ref{subsec:der} for motivation). A special distribution $\mu(\{0\})=1$
is denoted by a Dirac-delta $\delta_0$ and is called \emph{trivial}. The operator $\calQ_s$ 
depends on 
three parameters: a \textit{crossover probability} $\delta \in [0,1]$, a symmetric \textit{survey
distribution} $\mu_s$ and a \textit{(branching or) degree distribution} $P_d$ on $\mathbb{Z}_+$.
Given these, we define $\calQ_s\mu$ for any symmetric $\mu$ to be the 
probability law of random variable $R$
\begin{align}\label{eq:bp_survey}
     R \triangleq \sum_{u=1}^d Z_u F_{\delta}(\tilde R_u) + S, 
\end{align} 
  where $d\sim P_{d}$, $\tilde R_{u}\simiid \mu$, 
 $Z_u \simiid (-1)^{\Ber(\delta)}$, $S\sim \mu_s$ (all jointly independent) and 
\begin{align*}
    F_\delta(x)\triangleq\ln \frac{(1-\delta)e^x+\delta}{\delta e^x+1-\delta}=2 \tanh^{-1}\left( \left(1-2\delta\right)
\tanh{x\over2}\right).
\end{align*}
The special case of $S=0$ or, equivalently, $\mu_s=\delta_0 $  is referred to as \textit{BP
without survey} and in this case we denote the BP operator by $\calQ$ without the subscript $s$.

In this paper we consider the topic of convergence of iterations $\calQ_s^h \mu_0$ as
$h\to\infty$. Naturally, in this regard, we define distribution $\mu$ to be a \textit{BP fixed
point} if $\calQ_s \mu = \mu$. Note that in the case of no survey ($S=0$), and only in that case, there is a
\textit{trivial fixed point} $\mu = \delta_0$.

The main result of our work is the following.
\begin{thm}\label{thm:main} There exists at most one non-trivial symmetric BP fixed point $\mu^*$, unless we
are in the exceptional case of $S=0$ (no survey), $d = 1$ a.s. and $\delta \in \{0,1\}$ (in which
case $\mathcal{Q}_{\textup{s}}$ is an identity operator). In the
non-exceptional case, for all non-trivial symmetric $\mu$, the recursion $\calQ^{h}_{\textup{s}}\mu$ converges weakly to the same fixed point, to $\mu^*$ if it exists, or to the trivial $\delta_0$ otherwise.
\end{thm}
This result is contained in Theorems \ref{thm: uniq} and  \ref{thm: uniq_s} below. As alluded to in the abstract, Theorem~\ref{thm:main} resolves a number of long-standing questions in
the theory of Ising models on trees and locally tree-like graphs. In a nutshell, the main
innovation of our work is the discovery of a (rather strange) metric  between distributions
(equivalently, between BMS channels) under which a finite number of applications of
$\mathcal{Q}_{\textup{s}}$ is strictly contracting (see Definition \ref{def:deg_ind} below).

We note, however, that our result says little about the actual structure of the fixed point. From
prior work, though, we know that in
the case of $S=0$ (no survey) and fixed degree $d$ the $\delta_0$ is the unique
fixed point iff $(1-2\delta)^2 d \le 1$  (a Kesten-Stigum threshold~\cite{EPKS2000}). Above
criticality the non-trivial fixed point $\mu^*$ emerges and it is known to be approximately
Gaussian~\cite{9517800} in the sense that if $R_\delta \sim \mu^*$ then as $(1-2\delta)^2 d \to 1^+$:
$$ {R_{\delta}\over \sqrt{1-2\delta-d^{-{1\over2}}}} \stackrel{d}{\to}
\mathcal{N}(0,{\sigma^2\over 2}),\qquad \sigma^2 \eqdef {16d\sqrt{d}\over d-1}\,.$$

We next proceed to deriving the recursion~\eqref{eq:bp_survey} and explaining its connection to
Ising models, statistical physics and stochastic block model.

\subsection{Derivation of BP recursion}\label{subsec:der}

Consider an inference problem for the Ising model 
on infinite trees (see Fig.\ref{fig:main} for an illustration). We have a rooted tree channel that is generated recursively, where each vertex $v$ has an i.i.d. number of children sampled from a given degree distribution $P_{d}$.  Each vertex $v$ is associated with a binary random variable $X_v$. 
 The variable on the root, denoted by $X_0$, is $\Ber{(\frac{1}{2})}$. Then for any other vertex $v$,
 $X_v$ is identical to the variable on their parent node with probability $1-\delta$, conditioned on all other variables that are not their descendant, where $\delta$ is a given parameter in $[0,1]$.
 
 
 
 

 	\begin{figure}[htbp]
\centering
  \includegraphics[width=0.6\linewidth]{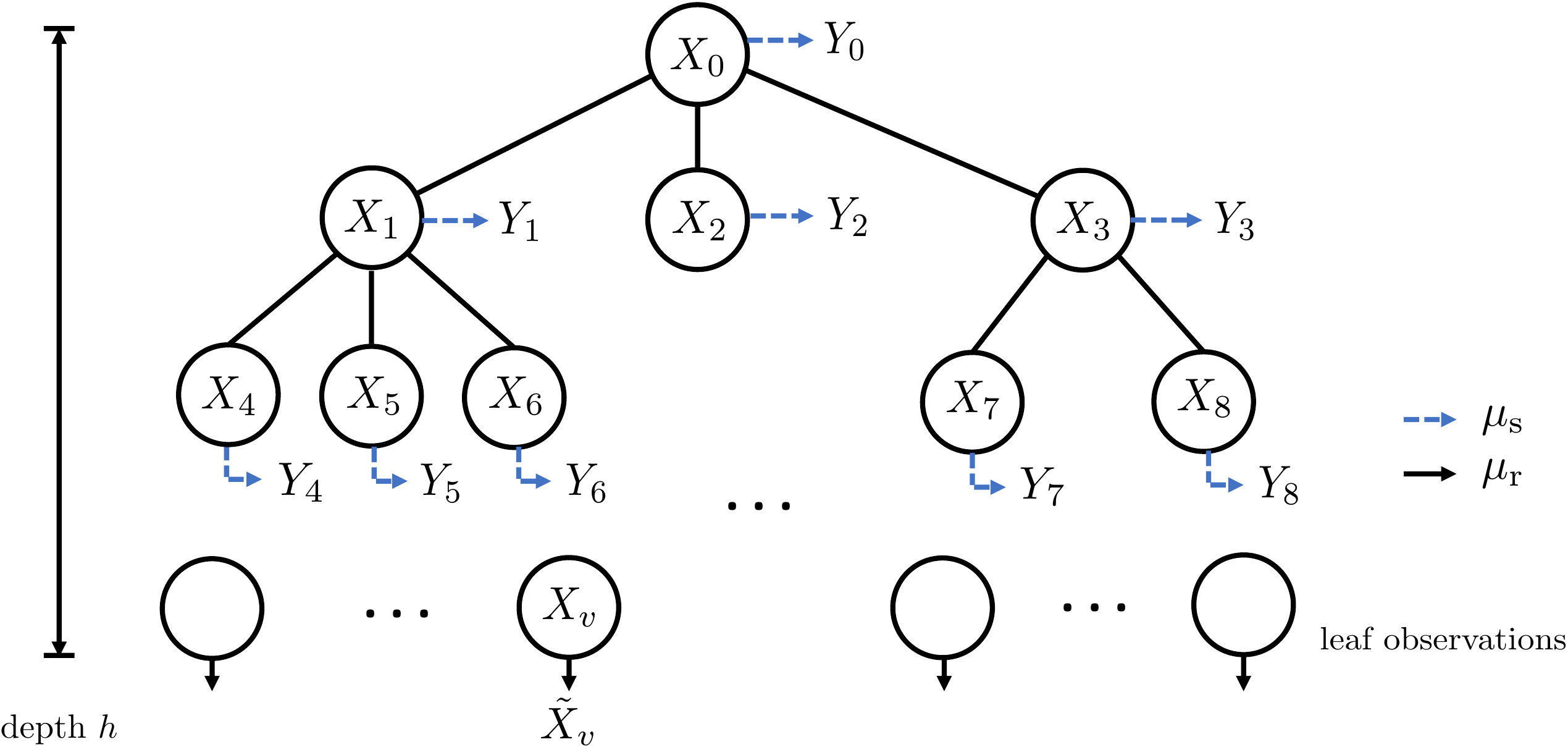}
\caption{Illustration of a tree channel with depth $h$.  
The solid lines represent binary symmetric channels  with error probability $\delta$. The
estimation of $X_0$ is based on leaf observations and the possible surveys. The corresponding
channels are represented by the solid arrows and the dashed arrows,  with log-likelihood ratio distributions given
by $\mu_{\textup{r}}$ and $\mu_{\textup{s}}$, respectively.  The leaf observations {consist
of }only the vertices at depth $h$, not including vertex $2$ in this example. 
}
\label{fig:main}
\end{figure}

 In a basic setting called broadcasting on trees (BOT), 
 we are interested in the process of estimating 
 $X_0$ given the random tree graph structure and the collection of variables on all vertices with depth $h$. 
 Formally, we aim to characterize the distribution of the log-likelihood ratio (LLR) for estimating $X_0$ conditioned on $X_0=0$. 
  Let $L_h$ denote the set of all vertices at depth $h$ from the root, $X_{L_h} = \{X_v: v \in
L_h\}$ be the vector of all values at depth $h$, and  $T_{(h)}$ denote the tree subgraph induced by all vertices in $\{L_0,L_1,...,L_h\}$. The set of all observed information consists of $T_{(h)}$ and $X_{L_h}$. 
  The LLR distribution for tree networks of depth $h$, denoted by $\mu_{(h)}$, is 
  given by    
  \begin{align*}
     \mu_{(h)}[{E}]=\PP\left[\ln {\PP[T_{(h)}, X_{L_h}|X_{0}=0]\over \PP[T_{(h)}, X_{L_h}|X_{0}=1]}\in{E}\bigg|X_0=0\right]
 \end{align*} for any measurable ${E}\subseteq (-\infty,+\infty]$.

     Knowing the LLR distribution, we can compute 
     fundamental quantities such as minimum error probabilities and mutual information. Particularly, 
     the recoverability of $X_0$ with non-trivial error probability (bounded away from $\frac{1}{2}$ for large $h$) is possible if and only if the LLR distributions converge to a non-trivial distribution.

   More generally, a variant of BOT is considered in 
 \cite{ACGP21}, where in addition to the \emph{leaf observations} 
 (i.e., variables at depth $h$), 
 each vertex above depth $h$ is also observed through identical BMS channels, and the received information is called the \emph{surveys}. This formulation can be further generalized to include the case where leaf observations are noisy, with corresponding observation channels being identical and BMS. The LLR distributions for such cases can be similarly defined. Formally, for any vertex $v$ with  depth less than $h$, there is a survey observation denoted by $Y_v$. For any vertex $v$ with depth $h$, there is a (possibly) noisy leaf observation denoted by $\tilde{X}_v$. All $Y_v$'s and $\tilde{X}_v$'s are jointly independent conditioned on the tree structure and the variables on the tree. The conditional distribution of each $Y_v$ and $\tilde{X}_v$ only depends on their corresponding variable $X_v$. They are specified  by either the transition distribution of the survey channels or the leaf observation channels, respectively. 
 Let $Y_{(h)}\triangleq\{Y_v: v\in L_0\cup L_1\cup...\cup L_{h-1} \}$ be the collection of all surveys and $\tilde{X}_{(h)}\triangleq\{\tilde{X}_v: v\in L_{h} \}$ be the collection of noisy leaf observations. The set of all observations is given by $\mathcal{O}_h\triangleq (T_{h}, Y_{(h)} ,\tilde{X}_{(h)})$. 
 
 Now that $\mathcal{O}_h$ can belong to continuous domains, 
 we define the LLR distribution $\mu_{(h)}$ to be the law of\footnote{Here the LLR variable $\ln\frac{f_{0,h}\left(\mathcal{O}_h\right)}{f_{1,h}\left(\mathcal{O}_h\right)}$ is constructed using the Radon–Nikodym derivative, which is well-defined for all measurable spaces and unique up to a set of zero measure.} 
   
  \begin{align*}
  \ln\frac{f_{0,h}\left(\mathcal{O}_h\right)}{f_{1,h}\left(\mathcal{O}_h\right)} \textup{ given } X_0=0\textup{, where } f_{v,h}\triangleq \frac{\textup{d}P_{\mathcal{O}_h|X_0=v}}{\textup{d}P_{\mathcal{O}_h}}.\nonumber
 \end{align*} 
  Following the same definition, we define the LLR distributions for all  binary-input channels by replacing $X_0$ with their input variables and  $\mathcal{O}_h$ with their outputs.


 

 %


    
   
    The BP operator arises naturally as the recursion rule for the LLR distributions of the tree channels described above. In the no survey setting, we have $\mu_{(h+1)}=\mathcal{Q}\mu_{(h)}$. When survey channels are present, 
  we have $\mu_{(h+1)}=\mathcal{Q}_{\textup{s}}\mu_{(h)}$, where $\mu_{\textup{s}}$ is set to be the LLR distribution of the corresponding survey channels. 
   
   The derivation of BP recursion relies on the fact that the observed information 
   can be partitioned into subsets of independent variables given  $X_0$ and the tree structure. Therefore, the overall LLR 
   can be written as the summation of individual LLRs from each component. Specifically, consider the tree channel of depth $h+1$. Let $d$ be the degree of the root vertex and let $1,...,d$ be the labels of the $d$ children. The individual terms consist of the LLRs from the subtrees rooted at each vertex $u\in[d]$, and the LLR from the survey when it exists.
   
   
   
   Due to the recursive structure of the tree construction, the subtree rooted at each vertex $u\in [d]$ resembles the same network that is defined for depth $h$. Consequently, if we consider  the LLR variable for estimating $X_u$ with the subtree information, which is a function of all relevant surveys, leaf observations, and the subtree structure, then by definition, the law of this variable is given by $\mu_{(h)}$ 
   conditioned on $X_u=0$. For convenience, 
   we denote this variable by $\hat{R}_u$. 
   Observe that each subtree is a BMS channel. We have $-\hat{R}_u\sim \mu_{(h)}$ conditioned on $X_u=1$. 
   Hence, by letting $Z_u\triangleq (-1)^{X_u}$ and  $\tilde{R}_u\triangleq Z_u \hat{R}_u$,  we have $\tilde{R}_u\sim \mu_{(h)}$ independent of $X_u$. 
      The LLR component for $X_0$ that corresponds to this subtree takes into account of the uncertainty of $X_u$, which is given by 
   $$\ln\frac{e^{\hat{R}_u}\mathbb{P}[X_u=0|X_0=0]+\mathbb{P}[X_u=1|X_0=0]}{e^{\hat{R}_u}\mathbb{P}[X_u=0|X_0=1]+\mathbb{P}[X_u=1|X_0=1]}=F_{\delta}(\hat{R}_u)=Z_uF_{\delta}(\tilde{R}_u).$$
 Hence, the overall LLR can be written as $\sum_{u=1}^d Z_uF_{\delta}(\tilde{R}_u) +S $, where 
   $S$ is the LLR variable for the survey at vertex $0$. 
   Recall that $\tilde{R}_u$, $Z_u$, and $S$ are jointly independent conditioned  on $X_0=0$, which is given by the tree channel construction. Further, we also have  $Z_u \sim (-1)^{\Ber(\delta)}$ and $S\sim \mu_s$ under the same condition.  Thus, we have recovered the BP operator specified by equation \eqref{eq:bp_survey}. Finally, the initial condition of the recursion is simply given by $\mu_{(0)}=\mu_{\textup{r}}$, where $\mu_{\textup{r}}$ is the LLR distribution of the leaf observation channels.
   
   
   The symmetry condition of $\mu_{\textup{s}}$,  $\mu_{\textup{r}}$, and $\mu_{(h)}$ arises from the BMS property of the 
   corresponding channels. 
   Generally, for any BMS channel, let $\mu$ denote its  LLR distribution and let $\mu^-$ be the distribution of the same LLR function generated with input $X_0=1$. By the BMS property, we must have $d\mu^-(r)=d\mu(-r)$ for any $r\in [-\infty,+\infty)$. Then the definition of LLR gives $d\mu_{(h)}(r)=e^{r}d\mu_{(h)}^-(r)$, which implies equation \eqref{eq:def_sym}. 
   


 

\subsection{Cavity method and previous work}

The operator $\calQ$
is also known as density evolution~\cite[Section 2.2]{coja2018information},
Bethe recursion~\cite[Definition 1.6]{dembo2013factor}. It arises from a so-called cavity
method~\cite{mezard1987spin}, which (non-rigorously, but often correctly) allows one to infer
important qualities and quantities of statistical physical systems based on the knowledge of the
fixed points of the $\calQ$. Correspondingly, the distributional identity $\calQ \mu = \mu$ is
known as the 1RSB cavity equation, with Parisi parameter set to $x=1$~\cite[Section 14.6,
(19.72)]{mezard2009information}. The particular version of the $\calQ$ corresponds to cavity
equations for the Ising model on a locally tree-like graph for ferromagnetic $\delta < {1\over2}$
or anti-ferromagnetic $\delta > {1\over 2}$ case. We mention that distributional recursions are
not necessary for understanding the former~\cite{mot09} (due to absence of frustration in the boundary condition) but
are necessary for the latter, see~\cite{coja2018information}. The version with survey (BP operator
$\calQ_s$) would correspond to certain random external fields, which are not independent across
sites.

Both the BOT and the BP fixed point
formulations have been widely studied in statistical physics \cite{cite-key, mot09}, evolutionary
biology \cite{10.1214/21-ECP423, Daskalakis2011}, and information theory \cite{EPKS2000,
gu2020broadcasting}. The 
condition for the existence of non-trivial symmetric BP fixed points has been determined exactly,
which can be described using the  branching number \cite{cite-key, EPKS2000, brc}. Equivalently, this resolved the problem of identifying the set of $(P_d, \delta)$ for which recovery is possible
in the BOT. The version of BOT with noisy leaf observations was introduced in~\cite{janson2004},
who also demonstrated that the regime of recovery (for general leaf observation channels under the Ising model case we are considering
here) is unchanged. 

A renewed interest in BOT was sparked by the groundbreaking works of~\cite{decelle2011asymptotic, mns15}, which
connected it to the stochastic block model with 2 communities (2-SBM). In 2-SBM the goal is to estimate a set of hidden labels by observing an associated random graph. The labels are defined on $n$ vertices, each being i.i.d. Ber$(\frac{1}{2})$. 
The graph is constructed by independently connecting any pair of vertices, with probability
$\frac{a}{n}$ for pairs with the same labels and with probability $\frac{b}{n}$ for the rest.
It turned out that in the regime of $n\to\infty$ the probability of error of estimating the label
given the graph can be lower bounded by the probability of recovery error in BOT with $P_d =
\mathrm{Poisson}({a+b\over 2})$ and $\delta = {b\over a+b}$. On the other hand,~\cite{MNS16} shows
that the 2-SBM error can be upper bounded by the recovery error in the robust BOT with the same
parameters. This was established by running BP from a good
initialization (for details, see \cite[Section 5]{MNS16}). Consequently, the uniqueness of BP fixed point implies that the upper and the
lower bounds coincide, showing that the performance of the optimal (but
exponential time) maximum likelihood estimator can be achieved in polynomial time. This result, however, was only shown
in~\cite{MNS16} under the condition of ``high SNR'' or low temperature. They made a conjecture
that the result (and BP uniqueness) should hold unconditionally. In~\cite{ACGP21} the range under which the
uniqueness holds was further enlarged. This paper resolves the conjecture in full.

The consequence of the discussion above is that  for the 2-SBM we can explicitly (modulo
computing the BP fixed point) evaluate the probability of error for recovering an individual
vertex label. A more global quantity is conditional entropy of \textit{all vertex labels} given
the graph. \cite{ACGP21} gives a formula for the latter, but only under the assumption of
\textit{boundary irrelevance (BI)} for the problem of BOT with survey. BI refers to the effect
of leaf observations becoming independent of the root variable when conditioned on the survey information. BI was conjectured to
be true in  \cite{kms16} for binary erasure survey channels, and in \cite{ACGP21} for general
symmetric survey channels. In our language, BI is
equivalent to uniqueness of BP fixed point for the $\calQ_s$. We thus provide a positive proof of these
conjectures in Section \ref{sec:bi_p}.

The setting of BOT with survey arose in the line of work~\cite{mx15,kms16} which investigated optimality of local
algorithms (of BP type) and made conjectures similar to the BI. Furthermore, using methods
of~\cite{MNS16} they were able to prove those conjectures in the regime of high-SNR. 
Our proof for all those conjectures closes the full spectra of the SNR and essentially follows 
from BI -- see Section \ref{sec:sbm_side}.

We also mention that uniqueness of BP fixed point has been investigated under a simplified formulation,
where the recursion is approximated using central limit theorem when the degrees are large -- see
Conjecture 2.6 in \cite{mx15}. Our proof techniques directly extend
to this limiting regime as well -- see Section \ref{sec:large_d}.

In conclusion, identifying the uniqueness of the BP fixed point for either $\calQ$ (no survey) or
$\calQ_s$ (with survey) has been a long-unsolved question, appearing in a web of interlinked
problems. Our resolution of the uniqueness closes all related conjectures as well
(Section~\ref{sec:imp}). Although partial resolutions were already presented in~\cite{MNS16,ACGP21}
the method here appears to be completely different and we do not believe that tightening of the
previous methods would be able  to close the full range of SNR -- this is briefly discussed
further in Section~\ref{sec:compare_mns}.

\subsection{Extension to non-symmetric distributions}

In the context of cavity equation, the BP recursion can  be defined and studied for asymmetric
distributions. 
To understand how the general class of distributions is defined, we need to recall that the BP operator is derived in a framework where the input $\mu$ can be viewed as the LLR distribution of some binary-input leaf observation channels conditioned on $X_0=0$. 
Any such distribution must satisfy the following condition. 
\begin{defn}\label{def:llrd}
We call a probability distribution $\mu$ on  $(-\infty,+\infty]$ an \emph{LLR distribution} if 
 \begin{align}\label{eq:intasym}
    \int e^{-r} d\mu(r)\leq 1. 
 \end{align}
We define the \emph{complement distribution} of $\mu$ to be any distribution $\mu^-$ on $[-\infty,+\infty)$ that satisfies 
 \begin{align}\label{eq:negative}
     d\mu (r) = e^{r}
	d\mu^-(r) \qquad \iff \qquad \mu^-[E] = \int e^{-r} \mathbbm{1}\{r \in E\} d\mu(r) 
 \end{align}
 for every measurable subset $E\subseteq (-\infty,+\infty)$.
\end{defn}
 In other words, 
$\mu$ must allow the existence of a $\mu^-$ that can be served as the law of the LLR variable generated by $X_0=1$.
\footnote{Conversely, 
      any LLR distribution can be mapped to a binary-input channel, similar to that any symmetric distribution is the LLR distribution of a BMS channel.} 
For general asymmetric LLR distributions, 
the BP operator that reflects the same process (i.e., BOT recursion specified by the same parameters but with asymmetric leaf observation channels) needs to be written in a slightly different form. It can be derived from the same steps in Section \ref{subsec:der}. Specifically, for any LLR distribution $\mu$ and any associated complement distribution $\mu^-$, we define $\mathcal{Q}_{\textup{s}} \mu$ to be the probability law of random variable $R$ 
 \begin{align}\label{eq:bp_basic_as}
    R \triangleq \sum_{u=1}^d F_{\delta}(\hat R_u) +S,\, 
\end{align} 
where $d\sim P_{d}$, $\hat R_{u}\simiid (1-\delta)\mu+\delta \mu^{-}$, and $S\sim \mu_s$, all jointly independent.
We call any LLR distribution $\mu$ a \emph{BP fixed point} if $\calQ_{\textup{s}}\mu=\mu$. 
It can be seen that for symmetric $\mu$ this definition coincides with the one we gave earlier in this section. 

Our main result implies the uniqueness of BP fixed points over general distributions  and  the unique convergence of BP recursion with  general initialization. 
The proof for the following result is presented in  Appendix \ref{app:asym}. 

\begin{corollary}[Asymmetric BP fixed points]\label{thm: uniq_as}
Fix any degree distribution $P_d$, parameter $\delta\in[0,1]$, and symmetric $\mu_{\textup{s}}$ that belongs to the non-exceptional case specified in Theorem \ref{thm:main}. There is at most one non-trivial BP fixed point $\mu^*$, and it is symmetric. For all non-trivial LLR distribution $\mu$, the recursion $\calQ^{h}_{\textup{s}}\mu$ converges weakly to the same symmetric fixed point, to $\mu^*$ if it exists, or to the trivial $\delta_0$ otherwise.  

\end{corollary}


\subsection{Organization of the paper}

The rest of the paper is organized as follows. In Section \ref{sec:pre}, we define some important tools and provide the proof ideas for our main theorem. Then in Section \ref{sec:pt_uniq}, we provide the proof details for the key intermediate steps. We illustrate in Section \ref{sec:imp}  how our results imply the solutions to the open conjectures mentioned earlier.     
Finally in Section~\ref{sec:gllr}, we extend our results to general tree structures, 
covering curious cases where the  
spin interactions can be stochastic (see the setting of i.i.d. weights in \cite{10.1214/21-ECP423}), periodic (see the first example in \cite{brc}), or nonisotropic (see the illustration in Fig. \ref{fig:threegraphs}).

\section{Proof ideas and outline}\label{sec:pre}

We first present the proof of the main result for the case of no survey $\mu_s = \delta_0$.
Namely, we show the following. 

\begin{thm}[Uniqueness without survey]\label{thm: uniq}
Fix any degree distribution $P_d$ and parameter $\delta\in[0,1]$ such that either $\PP[d=1] <1$ or
$\delta \in (0,1)$. There is at most one non-trivial symmetric BP fixed point $\mu^*$ for $\mathcal{Q}$. For all non-trivial symmetric $\mu$, the recursion $\calQ^{h}\mu$ converges weakly when $h\rightarrow\infty$ to the same fixed point, to $\mu^*$ if it exists, or to the trivial $\delta_0$ otherwise. 
\end{thm}

The proof of Theorem \ref{thm: uniq} builds upon ideas of channel comparison (a.k.a. comparison of experiments), which were previously used
in~\cite{EPKS2000} to show certain negative results, and more recently by~\cite{gu2020broadcasting} for the positive side. Here we extend this methodology in two ways:
a) strengthening the \textit{stringy tree lemma} from~\cite{EPKS2000}; and b) 
introducing of the concept of \emph{degradation index}. The latter allows us to define a potential function over symmetric distributions (and LLR distributions in general) that is only stabilized 
at a unique solution. 
For clarity, we illustrate the main concepts over symmetric distributions, which enables simplifications compared to their general forms. 
We first state the definition of degradation (see \cite[Section 15.2.3]{mezard2009information} for a reference). 
\begin{defn}
\label{def:degr}
For any two symmetric distributions $\mu_{Y}$, $\mu_{Z}$ defined on 
$(-\infty,+\infty]$, we say $\mu_Y$ is a \emph{degraded version} of $\mu_{Z}$, denoted by $\mu_Y\preceq
\mu_Z$, if one can define a joint distribution 
$\mu_{Y,Z}$ with $\mu_{Y}$, $\mu_{Z}$ as marginal distributions, such that $\mu_{Y|Z}$ is invariant under $(Y,Z)\rightarrow(-Y,-Z)$. 

\end{defn}
 
Intuitively, for any $\mu_Y\preceq
\mu_Z$, $\mu_Z$ can be viewed as the LLR distribution of a symmetric binary hypothesis testing problem and 
$\mu_Y$ can be viewed as a noisy version of $\mu_Z$ where the observation is corrupted by a symmetric noise channel $\mu_{Y|Z}$. A more detailed discussion on degradation can be found in Appendix \ref{app:deg_detail}. 

Our definition of the degradation index is based on the operator known as box-convolution
(notation coming from the LDPC codes, see \cite[page 181]{richardson2008modern}). Consider a pair
of BMS channels $A_1$ and $A_2$, with LLR distributions $\mu_i$, $i\in \{1,2\}$. Out of them we
can produce a new BMS channel $A$ as follows:
consider an input bit $X$ and generate an independent $X'$ as Bernoulli$(1/2)$, let $Y' =
A_2(X')$ be the noisy observation of $X'$ and let $Y = A_1(X \oplus X')$ be the noisy observation
of XOR of $X$ and $X'$. The $A$ channel is a channel from $X$ to the pair $(Y,Y')$ and we denote
its LLR by $\mu = \mu_1 \boxconv \mu_2$.  More formally, we have the following definition.

\begin{defn}[Box Convolution for Symmetric Distributions]\label{def:boxc}
For any $\phi\in[0,1]$, let $B_\phi$ denote the symmetric distribution defined on $\{-\ln\frac{1-\phi}{\phi},\ln\frac{1-\phi}{\phi}\}$. 
We define box convolution $\boxconv$ to be the  weakly continuous bilinear operator over the space of symmetric distributions satisfying the following condition  
$$ B_{\phi_1} \boxconv B_{\phi_2}\triangleq B_{\phi_1+\phi_2-2\phi_1\phi_2} \ \ \ \ \textup{for all } \phi_1, \phi_2\in\left[0,1\right].$$
\end{defn}
It is clear that box convolution is commutative. One can show the following alternative definition, which proves that box convolution is associative.
\begin{prop}\label{prop:box_alter}
Let $X\sim \mu$, $Y\sim \nu$ be independent random variables with symmetric distributions, then $\mu\boxconv\nu$ is 
identical to 
the law of 
\begin{align}\label{eq:prop6eq}
Z
\triangleq 2\tanh^{-1} \left(\tanh{X\over 2} \tanh {Y\over
2}\right).
\end{align}
\end{prop}
\begin{proof}
Note that the law specified by equation  \eqref{eq:prop6eq} is weakly continuous and bilinear. 
It fulfills all requirements in Definition \ref{def:boxc} (the special cases of $\mu=B_{\phi_1}$ and $\nu=B_{\phi_2}$ can be directly verified). 
On the other hand, the box convolution is uniquely determined through bilinear expansion where the input  distributions are expressed as   mixtures of
$\{B_{\phi}\}_{\phi\in[0,1]}$. For instance, any symmetric $\mu$ is a mixture 
under the law of  $\phi\sim \frac{1}{2}(1- \tanh\frac{|X|}{2})$. Hence, based on this uniqueness, any operator that satisfies Definition \ref{def:boxc} needs to be identical to the instance provided in the proposition. 
\end{proof}
\begin{remark}\label{remark:box_phy} We point out a convenient interpretation of the channel
corresponding to $B_\phi\boxconv \mu$. This channel corresponds to sequentially concatenating a binary symmetric channel (BSC)
with crossover probability $\phi$ and a channel with LLR $\mu$. Thus, compared to the $\mu$
channel the input bit first experiences a random $\phi$-flip. The general box convolution
$\mu_1\boxconv \mu$ corresponds to the same channel, except that the crossover probability $\phi$ is random (but known to the receiver) and $\log {1-\phi\over \phi} \sim \mu_1$. 
\end{remark}

The stringy tree lemma in \cite{EPKS2000} can be stated as follows, using box convolution.
\begin{thm}\label{thm:stringy}[Stringy Tree Lemma (STL)~\cite{EPKS2000}]
For any $\phi\in[0,1]$ and any symmetric distributions $\mu_1, \mu_2, ..., \mu_d$, we have 
$$ B_\phi \boxconv (\mu_1 * \mu_2 *\,...\, *\mu_d)\preceq (B_\phi \boxconv\mu_1) * (B_\phi \boxconv\mu_2)*\, ...\,* (B_\phi \boxconv\mu_d) .$$
\end{thm}
\begin{remark}
Recall the physical interpretation of box-convolution. 
Applying the theorem above repetitively compares any tree channel with a depth $1$ tree channel, each edge formed by concatenating all channels on a path in the original tree from its root to a corresponding leaf (hence the name stringy tree).  
\end{remark}


We establish the commutation relation between the BP operator and box convolution, but instead under a stronger notion of degradation (to be specified in Definition \ref{def:strict_deg}). Observe that the BP operator can be expressed using  elementary operations. 
%
\begin{prop}\label{prop:q_express}
For any symmetric $\mu$, we have $\calQ\mu= \EE_{P_d}[(B_\delta\boxconv \mu)^{*(d)}]$, where $(\cdot)^{*(d)}$ denotes self convolution by $d$ times.
\end{prop}
\begin{proof}
Recall the definition of BP operators. As implied by Proposition \ref{prop:box_alter}, for any independent  $Z_u\sim (-1)^{\textup{Ber}(\delta)}$ and  $\tilde{R}_u\sim \mu$, the law of $Z_uF_{\delta}(\tilde{R}_u)$ can be exactly written as $B_\delta\boxconv \mu$. Then, conditioned on any fixed $d$, the law for the summation of those intermediate variables equals the convolution of their distributions. Finally, the distribution $\calQ\mu$ is obtained by taking the mixture over the randomness in degree distribution.
\end{proof}
%
  The key result is summarized in Theorem \ref{lm:uniq_key} (proved in Section \ref{sec:pt_uniq}), which relies on the following definition. 
\begin{defn} [Strict Degradation]\label{def:strict_deg}
For any two distributions $\nu$ and $\mu$, we define $\nu \prec \mu$ if $\exists$ $\phi\in (0,\frac{1}{2}]$ such that $\nu\preceq B_\phi\boxconv\mu$.
\end{defn}
\begin{thm}\label{lm:uniq_key} 
For any $\delta\in [0,1]$ and degree distribution $P_d$ satisfying $\PP[d\leq 1]<1$, let $\calQ$ be the associated BP operator. 
Then for any $\phi\in(0,\frac{1}{2})$ and symmetric $\mu$, we have 
\begin{align}
    B_{\phi}\boxconv\mathcal{Q}^2\mu&\prec \mathcal{Q}^2(B_{\phi}\boxconv\mu) &&\textup{if } \PP[d>2]=0,\label{sdeg:key1}\\
   B_{\phi}\boxconv\mathcal{Q}\mu&\prec  \mathcal{Q}(B_{\phi}\boxconv\mu)&&\textup{otherwise}.\label{sdeg:key2}
\end{align}
\end{thm}

\begin{remark}
Note that the fixed point equation can be written as $\mathcal{Q}\mu=\mu$, Theorem \ref{lm:uniq_key} implies that $\mu$ and $B_{\phi}\boxconv\mu$ can not be both non-trivial symmetric fixed points. As we will show later in this work (see Remark \ref{rem:deg_met}),  Theorem \ref{lm:uniq_key} essentially states that either $\mathcal{Q}^2$ or $\mathcal{Q}$ will always reduce a distance function between distinct non-trivial symmetric distributions  measured based on the degradation index.  
\end{remark}

\begin{remark}
The stringy tree construction in \cite{EPKS2000} implies that $B_{\phi}\boxconv\mathcal{Q}\mu\preceq  \mathcal{Q}(B_{\phi}\boxconv\mu)$ for any degree distribution. Theorem \ref{lm:uniq_key} provides a strict version of this result for $\PP[d>2]>0$. Note that this condition can not be further relaxed, as one can show that inequality (\ref{sdeg:key2}) is not satisfied when $\PP[d>2]=0$. 
\end{remark}

\ifgeneral
\begin{remark}
Theorem \ref{lm:uniq_key} shows that our proof requires special treatment when $d\leq 2$. As shown in 
Section \ref{sec:gllr}, these two possible cases can be naturally unified when viewed under a generalized model, where the infinite tree is generated from arbitrary elements. A generalized version of inequality (\ref{sdeg:key1}) and (\ref{sdeg:key2}) is provided, and the number of applications of the BP operator for
strict 
inequality to hold depends on a requirement called polygon condition (see Definition \ref{def:polc} and Theorem \ref{thm:gen_t_tec}).


\end{remark}
\fi




We also have the following fact shown in Appendix \ref{app:pp_trans_deg_stdeg} (by checking that
degradation is transitive and box-convolution-preserving):
 \begin{prop}\label{prop:trans_deg_stdeg}
 If $\nu\preceq\tau\prec \mu$ or $\nu\prec\tau\preceq \mu$, then $\nu\prec\mu$.
 \end{prop}
 
Given these results, we are ready to prove the main theorem for the no survey case. 
\begin{proof}[Proof of Theorem \ref{thm: uniq}] 
Consider any two non-trivial symmetric fixed points $\mu,\nu$, we first prove that $\mu\preceq\nu$. 

\begin{defn}[Degradation Index] 
\label{def:deg_ind} 
For any two symmetric distributions $\mu$ and $\nu$, we define the \emph{degradation index} from $\nu$ to $\mu$ to be
\begin{align}\label{eq:degind_def}
\phi^*(\mu
,\nu)    \triangleq \inf\left\{{\phi}
\ |\  B_{\phi}\boxconv \mu\preceq \nu \right\}.
\end{align}
\end{defn}
We use the following Proposition, which is proved in Appendix \ref{app:pl_index_def}.
\begin{prop}\label{lemma:index_def}\label{prop:deg_closed}\label{prop:tri}
  Degradation index has the following properties.
  \begin{enumerate}
      \item We always have $\phi^*(\mu,\nu)<\frac{1}{2}$ for $\nu$ non-trivial.
      \item For any symmetric $\mu$ and $\nu$, we have $B_{\phi^*(\mu,\nu)}\boxconv\mu\preceq \nu$.
      \item For any $\phi\in(0, \frac{1}{2})$ and any symmetric 
$\mu$ and $\nu$ satisfying $B_{\phi}\boxconv\mu\prec \nu$, we have $\phi^*(\mu,\nu)<\phi$.
\item For any symmetric $\mu$, $\nu$, $\tau$, we have 
$1-2\phi^*(\mu,\tau)\geq (1-2\phi^*(\mu,\nu))(1-2\phi^*(\nu,\tau))$.
  \end{enumerate} 
\end{prop}

  
    We consider the fixed point condition $\calQ \mu=\mu$ and any $\phi\in(0,\frac{1}{2})$ satisfying   $B_{\phi}\boxconv\mu\preceq \nu$. From the first property in  Proposition \ref{lemma:index_def}, such $\phi$ exists. 
    Then from the second property, we can choose $\phi=\phi^*(\mu,\nu)$ unless $\phi^*(\mu,\nu)=0$. We focus on the non-trivial case where $\mathbb{P}[d\leq 1]<1$. Otherwise, any fixed point $\mu$ has to satisfy $\mu(\{0\})=1$, which makes them unique and trivial, unless $\mathbb{P}[d\leq 1]=1$, which falls into the exceptional case stated in the theorems.    
    Under this condition, Theorem \ref{lm:uniq_key} states that there is an integer $k$ for any degree distribution such that
    $$ B_{\phi}\boxconv\mu\prec \mathcal{Q}^k(B_{\phi}\boxconv\mu). $$
    Because $\mathcal{Q}$ describes BP, it preserves degradation.  Hence,
      $$ \mathcal{Q}^k(B_{\phi}\boxconv\mu)\preceq {Q}^k\nu=\nu. $$
    Recall the transitivity property stated in Proposition \ref{prop:trans_deg_stdeg}. 
    This implies $B_{\phi}\boxconv\mu\prec \nu$.

        However, this conclusion is mutually exclusive with $\phi=\phi^*(\mu,\nu)$ according to the third property in Proposition \ref{lemma:index_def}, which states that strict degradation provides strict upper bound on degradation index.   
          Thus, we must have $\phi^*(\mu,\nu)=0$, and $\mu\preceq\nu$ follows from the second property in Proposition \ref{prop:deg_closed}.
          
          By symmetry, we have $\nu\preceq\mu$ as well, 
	  which in turn implies $\nu = \mu$ since degradation satisfies antisymmetry (see the fourth property in Proposition \ref{prop:q_deg}). 
	  Therefore, there can be at most one non-trivial symmetric
	  BP fixed point.

	We next prove convergence of iterations $\mu_h = \calQ^h \mu$ to  $\mu^*$, which denotes
	from now on either the unique non-trivial fixed point (if it exists) or $\delta_0$
	(otherwise).
	First, notice
	that if $\mu_h \to \mu_\infty$ (weakly) then $\mu_\infty$ must be a fixed point of
	$\calQ$.  Indeed, we have $\mu_{h+1} = \calQ \mu_h$ and taking $h\to\infty$ we get the
	statement $\calQ \mu_\infty = \mu_\infty$ after applying (weak) continuity of $\calQ$. The
	weak continuity follows from the following argument. Let $\{R_{h,u}, u=1,\ldots\}$ be iid
	$\sim \mu_h$. From Skorokhod's representation we can assume $R_{h,u} \to R_u$ (almost
	surely) as $h\to \infty$. But then from equation~\eqref{eq:bp_survey} we see that
	conditioned on any value of degree $d$ we have
	$$ \sum_{u=1}^d Z_u F_\delta(R_{h,u}) \stackrel{a.s.}{\to} \sum_{u=1}^d Z_u
	F_\delta(R_{u})\,,$$
	implying that $\calQ \mu_h \to \calQ \mu_\infty$.

	Second, we notice that the sequence $\tilde \mu_h = \calQ^h B_0$ (i.e. BP iteration initialized from the
	measure corresponding to perfectly observed leaves: $B_0[+\infty] = 1$) is monotonically decreasing (since it
	corresponds to a channel from $X_0$ to $X_{L_h}$, see Fig.~\ref{fig:main}). Thus, it is
	convergent (Prop.~\ref{prop:lim_deg_com}) and even $\tilde \mu_h \to \mu^*$. Indeed, since $B_0 \succeq
	\mu^*$ we must have $\lim \tilde \mu_h \succeq \mu^*$, which implies that the limit cannot
	be $\delta_0$ unless $\mu^*=\delta_0$.

	Third, we notice that for any $\phi\in[0,1]$, the stringy tree lemma implies that the sequence $\hat \mu_h = \calQ^h
	\hat \mu $, where $\hat \mu = B_\phi\boxconv \mu^*$, is monotonically increasing and hence
	convergent by Prop.~\ref{prop:lim_deg_com}. Indeed, $\hat\mu_1 = \calQ \hat{\mu} =\calQ (B_\phi \boxconv \mu^*) \succeq B_\phi
	 \boxconv \calQ \mu^* = \hat\mu$. Thus applying $\calQ^{h-1}$ to both sides we obtain $\hat\mu_h \succeq
	\hat\mu_{h-1}$. Further, if $\hat \mu \neq \delta_0$ then we can see that the limit point satisfies $\lim_{h\to\infty} \hat
	\mu_h \succeq \hat{\mu}$, which can not be $\delta_0$. Hence, the convergence in this case is always to $\mu^*$. 

	Finally, we complete the proof. If $\mu^* = \delta_0$ (i.e. no non-trivial fixed points
	exist) then $\delta_0 \preceq \mu \preceq B_0$ and applying $\calQ^h$ we obtain that
	$\mu_h \to \delta_0$ by sandwiching (cf.~Prop.~\ref{prop:lim_deg_com}). If $\mu^* \neq
	\delta_0$ and $\mu \neq \delta_0$ then let $\phi < 1/2$ be such that $B_\phi \boxconv \mu^* \preceq
	\mu$ (Prop. \ref{prop:deg_closed}) and we have $\hat{\mu}\triangleq B_\phi \boxconv \mu^*\neq \delta_0$. In this case, $\hat \mu_h \preceq \mu_h \preceq \tilde \mu_h $ and again the
	sandwich property shows $\mu_h \to \mu^*$.
\end{proof}


\begin{remark}[Distance contraction]\label{rem:deg_met}

We have shown that a finite number of applications of the BP operator 
 satisfies the following contractivity condition 
 	$$ d(\calQ^k \mu, \calQ^k \nu) < d(\mu,\nu) \qquad \forall \mu,\nu \neq \delta_0\,,$$
	for a metric $d$ between non-trivial symmetric distributions, defined as 
\begin{align}\label{eq:deg_m}
    d(\mu,\nu)=|\ln(1-2\phi^*(\mu,\nu))+\ln(1-2\phi^*(\nu,\mu))|.
\end{align}
    That $d(\mu,\nu)=0$ implies $\mu=\nu$ is clear, while the triangle inequality follows from
    Proposition \ref{prop:tri} (fourth property). 
    Some properties of this metric are discussed in Appendix \ref{app:deg_met}. In particular, it
    is strictly stronger than weak convergence (Levy-Prokhorov) metric.

 A common way to prove convergence is to apply a well-known principle of Edelstein, which states
 that contractive self-map over a compact metric space defines recursions that converge to a
 unique fixed point~\cite[Remark
 3.2]{Edelstein62}. However, as 
 Remark 
\ref{remark:lack_comp} in Appendix \ref{app:deg_detail} shows, the space of symmetric distributions is not compact under the degradation metric,  even with a radius constraint. Thus, we proved the convergence of BP recursion directly.  


   We mention that classically, contraction methods have mostly been studied for linear recursions
   (i.e. affine combinations of independent variables), see  \cite{ROSLER1992195, Rosler99alimit, neininger2004general}. 
   Note that setting $S=0$, $Z_u=c$ (constant) and $F_\delta(r)= r$ in~\eqref{eq:bp_survey}
   reduces the search for BP fixed points to finding stable laws. The speed of convergence to
   stable laws (in particular in the central limit theorems) has been studied by constructing
   special 
   distances such as Zolotarev metrics~\cite{doi:10.1137/1121086}. 
 (For such metrics contraction properties can be proved by analyzing the norms of the coefficient
 matrices \cite{neininger2004general}.) Our work may be seen as identifying the appropriate
 counterpart (degradation distance) for the non-linear model of~\eqref{eq:bp_survey}.
\end{remark}


\subsection{Comparison to the methods of~\cite{MNS16} and~\cite{ACGP21}}\label{sec:compare_mns}

Since partial resolutions of the BP uniqueness were already done in~\cite{MNS16}, it is natural to
ask whether the method here is merely tightening of~\cite{MNS16}. Especially since
in~\cite{ACGP21} the range in which uniqueness holds was extended compared to~\cite{MNS16} by
precisely leveraging channel degradation. We want to argue, however, that the proof here is
fundamentally different and explain why the methods as
in~\cite{MNS16,ACGP21} are inherently tailored to ``high-SNR'' cases exclusively. 

Consider the case of regular trees (i.e., with fixed $d$) which are key building blocks in the proofs of \cite{MNS16,ACGP21}. The authors investigated contraction properties of potential functions that are either defined in the form of the $L_p$ distance, or can be lower bounded by them. In the regime of $d(1-2\delta)^2>1$, 
such potential functions are non-contractive when the LLR distributions are close to the trivial distribution. {For instance, taking 
 symmetric Gaussian distributions with vanishing second moments. The application of BP operators can be approximated 
 with $F_{\delta}(R)\approx (1-2\delta)R$, and the distributions are scaled by a factor of $d^{\frac{1}{2}}(1-2\delta)$. As a consequence, metrics defined in $L_p$ norms are also increased by the same factor 
 after each BP recursion.} 
Therefore, to exploit any contraction property, the core of the proofs in \cite{MNS16,ACGP21} is
to identify cases where the BP recursion is bounded away from the trivial distribution. 
However, when this condition is imposed, 
the high SNR requirement emerges  
 as all BP fixed points converge to $\delta_0$ when $d(1-2\delta)^2\rightarrow 1^+$.  
 

 

In this work, we employed a different approach by constructing a new measure of distance between
the LLR distributions. In contrast to previous works,  where degradation was mainly used to assist the analysis for existing potential functions, we incorporate degradation as a part of the construction. The  metric we developed is in some sense ``scale-invariant'' and allows us to treat the
``low--SNR'' cases and ``high--SNR'' cases simultaneously. 





\subsection{Extension to Broadcast with Survey}\label{sec:sur}

Now we present the proof ideas for non-trivial survey distributions. Formally, we prove the following theorem. 



\begin{thm}[Uniqueness with non-trivial survey]\label{thm: uniq_s}

Fix any degree distribution $P_d$, parameter $\delta\in[0,1]$, and non-trivial symmetric survey distribution $\mu_{\textup{s}}$. 
There is exactly one non-trivial symmetric BP fixed point $\mu^*$. For all symmetric $\mu$, the recursion $\mathcal{Q}_{\textup{s}}^h\mu$ converges weakly to $\mu^*$.  
\end{thm}

The proof of uniqueness relies on the following intermediate step, which is proved in Section \ref{sec:pt_key_survey}.

\begin{thm}\label{thm:key_survey}
For any $\delta\in [0,1]$, non-trivial symmetric $\mu_{\textup{s}}$, and degree distribution $P_d$ with $\PP[d=0]<1$, 
let $\calQ_{\textup{s}}$ be the associated BP operator. Then for any $\phi\in(0,\frac{1}{2})$ and symmetric $\mu$, we have  
\begin{align}
    B_{\phi}\boxconv\mathcal{Q}_{\textup{s}}^k\mu&\prec \mathcal{Q}_{\textup{s}}^k(B_{\phi}\boxconv\mu)
 &&\textup{if } \PP[d>1]=0,\label{sdeg:sur_key1}\\
   B_{\phi}\boxconv\mathcal{Q}_{\textup{s}}\mu&\prec  \mathcal{Q}_{\textup{s}}(B_{\phi}\boxconv\mu)&&\textup{otherwise},\label{sdeg:sur_key2}
\end{align}
for some $k\in\mathbb{N}_+$.
\end{thm}

Assuming the correctness of Theorem \ref{thm:key_survey},  the proof is obtained as follows.  

\begin{proof}[Proof of Theorem \ref{thm: uniq_s}] 
We focus on the non-trivial case where $\PP[d=0]<1$, otherwise, we have $\mathcal{Q}_{\textup{s}}\mu=\mu_{\textup{s}}$ and Theorem \ref{thm: uniq_s} clearly holds. 
The key observation here is that Theorem \ref{thm:key_survey} plays the same role as Theorem \ref{lm:uniq_key} in the no survey case. Hence, by following the same steps in the proof of Theorem \ref{thm: uniq}, 
one can show that 
there is at most one non-trivial symmetric BP fixed point. 
On the other hand, by the monotone convergence property of degradation, the BP recursion with noiseless leaf observation (i.e., $\mu_{\textup{r}}(\{+\infty\})=1$) converges weakly to a symmetric fixed point. Note that $\mu_{\textup{s}}$ is non-trivial. Any BP fixed point with respect to $\calQ_{\textup{s}}$ must also be non-trivial. 
This proves the existence of $\mu^*$. To summarize, there is exactly one unique non-trivial symmetric BP fixed point.  

The convergence of 
$\mu_h = \mathcal{Q}_{\textup{s}}^h \mu$ to  $\mu^*$ can be proved  by sandwiching between the BP recursions initialized by $B_{0}$ and $\delta_0$. The first sequence $\mathcal{Q}_{\textup{s}}^h B_0$ corresponds to the noiseless leaf observation case, which converges to $\mu^*$. The second sequence $\mathcal{Q}_{\textup{s}}^h \delta_0$ corresponds to the no leaf observation case (i.e., $\mu_{\textup{r}}(\{0\})=1$), which is monotonically increasing (since they each corresponds to a channel with an increasing set of surveys), and hence 
convergent to same the symmetric fixed point (due to uniqueness).  
Therefore, the sandwiching 
of Prop. \ref{prop:sand} can be applied with the comparison $\mathcal{Q}_{\textup{s}}^h\delta_0\preceq \mu_{h}\preceq \mathcal{Q}_{\textup{s}}^h B_0$, which is due to the natural coupling. 
\end{proof}

\begin{remark}
The existence of survey channels significantly affects the properties of BP fixed points. As we have shown, when the survey channels have a non-zero capacity, there is always one unique BP fixed point, and it is non-trivial. However, when survey channels are absent, either the trivial and non-trivial fixed points coexist, or only the trivial solution remains. This difference is also reflected in the statements of the uniqueness theorems.


\end{remark}










\section{Technical details}\label{sec:pt_uniq}

In this section, we prove the key intermediate steps, i.e., 
Theorem \ref{lm:uniq_key} and Theorem \ref{thm:key_survey}.
The gist of the proof is to 
characterize the commutation rules between the building blocks of the BP operators. 
   In particular, we use the commutativity of box convolution, which is equivalently $F_\phi \circ F_\delta = F_\delta \circ F_\phi$. Then we develop an exchange rule between convolution and box convolution. For convenience, we make the following definition.
\begin{defn}
\label{def:beta_s}
For any symmetric distribution $\mu$, we define its $\beta$\emph{-curve} as a function on domain $t \in\mathbb{R}$ given by the following equation.
\begin{align}
    \beta(t;\mu)\triangleq \mathbb{E}_{R\sim\mu}\left[\max\left\{\tanh\frac{|R|}{2},|t|\right\}\right].
\end{align}
    We also define
    \begin{align*}
        t_{\max}(\mu)&\triangleq \inf\{t\in[0,1]\ |\ \beta(t) = t\}. 
        \end{align*}
\end{defn}

The meaning of the $\beta$-curve is given by the next two results (one trivial, one classical). 

\begin{prop}\label{prop:rem:beta_err}
For any $t\in[0,1]$ and symmetric $\mu$, the function value $\frac{1}{2}(1-\beta(t;\mu))$    equals the minimum error probability 
for a binary hypothesis testing problem with LLR  distribution given by $\mu$ and prior given by $\Ber(\frac{1-t}{2})$.
\end{prop}

See Appendix \ref{app:pp_rem:beta_err} for a proof of the first result. The next result is a celebrated Blackwell–Sherman–Stein (BSS) theorem \cite{dbc,sso,csn,
blackwell1953equivalent} (an equivalent form is also stated in \cite[Theorem
4.74]{richardson2008modern}), which connects degradation and the $\beta$-curves.
\begin{thm}\label{thm:ftdeg}[Blackwell–Sherman–Stein]
For any pair of symmetric distributions $\nu$ and $\mu$, we have $\nu\preceq
\mu$ if and only if $\beta(t;\nu)\leq \beta(t;\mu)$ for all $t\in[0, 1]$.
\end{thm}

Our principal analytical tool is the following extension of the BSS result to the relation of strict degradation (see
Appendix \ref{app:pl_st_beta} for the proof).
\begin{prop}\label{lemma:strict_beta}
For any non-trivial
symmetric distributions $\mu$ and $\nu$, the following statements are equivalent.

\begin{enumerate}
    \item 
$\nu\prec\mu$.
\item 
$\beta(t;\nu)<\beta(t;\mu)$ for all $t\in[0, t_{\max}(\nu)]$. 
\item $\beta(t;\nu)<\beta(t;\mu)$ for all $t\in[0, t_{\max}(\mu))$ and $t_{\max}(\nu)<t_{\max}(\mu)$.
\end{enumerate}
\end{prop} 

Together with Prop.~\ref{prop:rem:beta_err} we can see that strict degradation between two
distributions can be checked by comparing
their associated Bayesian inference problems. The proof of Proposition \ref{lemma:strict_beta}
relies on a several elementary steps.  One such step is the fact that box convolution with any
$B_{\phi}$ can be viewed as a homothetic transformation:
\begin{prop}\label{prop:beta_box}
For any $\phi\in[0,\frac{1}{2})$, the $\beta$-curve of $B_\phi$ is given by
  \begin{align} \label{eq:beta_b}
        \beta(t; B_\phi)=\max\{|t|, 1-2\phi\}.
    \end{align} 
More generally, for any symmetric $\mu$, 
   \begin{align}
        \beta(t; B_\phi\boxconv\mu)&=(1-2\phi)\beta\left(\frac{t}{1-2\phi};\mu\right).\label{eq:beta_bconv}
    \end{align} 
\end{prop}
\begin{proof}
Equation (\ref{eq:beta_b}) directly follows from the definition of $\beta$-curves.
Then equation (\ref{eq:beta_bconv}) can be proved from the second definition of box convolution.  
\end{proof}

In order to prove the needed $\beta$-curve gaps 
for Theorem \ref{lm:uniq_key} and Theorem \ref{thm:key_survey},
we define the following refinement of strict degradation. 

\begin{defn}\label{def:semi_strict}
For any symmetric distributions $\mu$, $\nu$, and parameter $s$, we define $\nu\prec_s\mu$ if $\beta(t;\nu)<\beta(t;\mu)$ for all $t$ satisfying  
$$t\in\left(\tanh{\left(\frac{s}{2}\right)}
, 
t_{\max}(\nu)\right]$$
and $\beta(t;\nu)\leq \beta(t;\mu)$ for all other $t$. 
We also define
$$    r_{\max}(\mu)\triangleq \sup\{r\in[0,+\infty]\ |\ \mu([r,+\infty])>0\}=2\tanh^{-1}{t_{\max}(\mu)}.
    $$
\end{defn}

The following propositions are proved in 
Appendix \ref{app:pl_bounded_crit} and Appendix \ref{app:pl_uniq4b}, respectively.

\begin{prop}\label{prop:bounded_crit}
Let $\mu$, $\nu$, $\tau$ be symmetric distributions.

\begin{enumerate}
    \item 
We have $\nu\prec_s \mu$ if  $\beta(t;\nu)<\beta(t;\mu)$ for all $t\in \left(\tanh{\left(\frac{s}{2}\right)},
t_{\max}(\mu)\right)$, $\beta(t;\nu)\leq \beta(t;\mu)$ for all other $t$, and $t_{\max}(\nu)<t_{\max}(\mu)$; when $s<r_{\max}(\mu)$, the converse also holds true.
    \item If $\nu\preceq\tau\prec_s \mu$ or $\nu\prec_s\tau\preceq \mu$, then $\nu\prec_s\mu$.
    \item 
$\nu\prec \mu$ implies $\nu\prec_s\mu$ for any $s\geq - r_{\max}(\mu)$, the latter  implies $\nu\preceq\mu$. 
\item $\nu\prec_s\mu$ is always true when  $s\geq r_{\max}(\nu)$ and $\nu\preceq\mu$; if $\nu\prec_s\mu$ for any $s<0$, then $\nu\prec\mu$. 
\end{enumerate}
\end{prop}

\begin{prop}\label{lemma:uniq_4b}
For any $\phi\in(0,1)$ and $\delta_1,\delta_2\in[0,1]$, let 
$s_{\min}=|{{F_{\phi}(r_{\max}(B_{\delta_1}))-F_{\phi}(r_{\max}(B_{\delta_2}))}}|.$
Then we have 
\begin{align}\label{eq:simple_beta_rel}
    B_\phi\boxconv (B_{\delta_1}*B_{\delta_2})\prec_{s_{\min}} (B_\phi\boxconv B_{\delta_1})*(B_\phi\boxconv B_{\delta_2}).
\end{align}
\end{prop}
The $\beta$-curves for the distributions on both sides of inequality \eqref{eq:simple_beta_rel} are piecewise linear. 
Moreover, 
 the proof in Appendix \ref{app:pl_uniq4b} provides the exact condition for the inequality between $\beta$-curves to be strict. 
 By averaging over $\delta_1,\delta_2$ the $\beta$-curves in the above special case, one can prove the following corollary (see Appendix \ref{app:pc_uniq_4} for details).  
\begin{corollary}\label{coro:uniq_4}
For any $\phi\in(0,1)$ and any two symmetric distributions $\mu_1$, $\mu_2$, we have
\begin{align}\label{eq:gen_beta_rel}
    B_\phi\boxconv(\mu_1*\mu_2)\prec_{s} (B_\phi\boxconv\mu_1)*(B_\phi\boxconv\mu_2)
\end{align} 
for 
$$s= {|F_\phi(r_{\max}(\mu_1))-F_{\phi}(r_{\max}(\mu_2))|} .$$
\end{corollary}
Now we establish connections between $\beta$-curve and convolution. 
\begin{prop}\label{lemma:basic_conv}
   For any $\phi\in[0,\frac{1}{2}]$, $t\in[0,1)$,  and symmetric $\mu$, we have 
\begin{align*}
\beta(t; B_{\phi}*\mu)=\left(\frac{1+t-2t\phi}{2}\right)\beta\left(t_0;\mu\right)+\left(\frac{1-t+2t\phi}{2}\right)\beta\left(t_1;\mu\right),
\end{align*}
where
\begin{align*}
    t_0&\triangleq\frac{1+t-2\phi}{1+t-2t\phi}=\tanh{\left(\frac{r_{\max}(B_{\phi})}{2}+\tanh^{-1}(t)\right)},\\
        t_1&\triangleq\left|\frac{1-t-2\phi}{1-t+2t\phi}\right|=\tanh\left({\left|\frac{r_{\max}(B_{\phi})}{2}-\tanh^{-1}(t)\right|}\right).
\end{align*}
    \end{prop}
    Proposition \ref{lemma:basic_conv} is proved in Appendix \ref{app:pl_basic_conv}. By averaging
    the $\beta$-curves in Proposition \ref{lemma:basic_conv} over $\phi$, we have the following
    ``rule of convolution'', which is proved in Appendix \ref{app:pl_convdeg}.   
   
   \begin{defn}
   For any  symmetric $\mu$, define $\supp(\mu)\triangleq \{v\in\mathbb{R}\ |\ \mu([v-\epsilon,v+\epsilon])>0\ \forall \epsilon>0\}$. 
   \end{defn}

   \begin{prop}[Rule of Convolution] 
   \label{lemma:convdeg}
   For any symmetric $\mu$, $\nu$, $\tau$, and $\ell\in \supp(\tau)$, if  
  $\nu\prec_{s}\mu$, 
then 
\begin{align}\label{ineq:roc}
    \beta\left(\tanh{\frac{r}{2}};{\nu*\tau}\right)<\beta\left(\tanh{\frac{r}{2}}; {\mu*\tau} \right)
\end{align}
for any $r\in(s-\ell, r_{\max}(\mu)-\ell)$; if $\nu\prec\mu$, then inequality (\ref{ineq:roc}) holds for $r\in[-\ell, r_{\max}(\mu)-\ell)$.
    \end{prop}
   
   Notice that the above $\beta$-curve analysis can be used to show the expected convolution-preserving properties of degradation.  
   We extend them to  $\prec_s$ and
   $\prec$, stated as follows.
	\begin{prop}\label{prop:conv}
	For any symmetric $\mu$, $\nu$, and $\tau$, we have the following facts.
	\begin{enumerate}
	\item If $\nu \prec \mu$, 
	then $\tau
	\boxconv \nu \prec \tau \boxconv \mu$. 
	\item If $\nu \prec_{s} \mu$, 
	then $\tau \boxconv \nu \prec_{s_\tau}
	\tau \boxconv \mu,$ where $s_\tau\triangleq \ln\frac{e^{s}e^{r_{\max}(\tau)}+1}{e^{s}+e^{r_{\max}(\tau)}}$.
	\item If $\nu \prec \mu$, then unless $r_{\max}(\tau)=r_{\max}(\mu)=+\infty$, we have $\tau *\nu \prec_s \tau	*\mu$  with $s={r_{\max}(\tau)-r_{\max}(\mu)}$ ; if in addition, $r_{\max}(\tau) <  r_{\max}(\mu)$, 
	then $\tau *\nu \prec \tau
	*\mu$. 
	
	\end{enumerate}
	\end{prop}
\begin{proof}
The first two properties are based on the fact that box convolution with any symmetric $\tau$ can be viewed as a mixture of homothetic transformations in $\beta$-curve (see Proposition \ref{prop:beta_box}). Hence, any strict inequality is preserved, except when $\tau$ is trivial, where the proof is straightforward. 

The third property can be proved with Proposition \ref{lemma:convdeg}. Specifically, for the non-trivial case where $r_{\max}(\tau)<+\infty$, we can choose $\ell=\pm r_{\max}(\tau)$ and obtain that the $\beta$-curve inequalities are strict for $t\in [-r_{\max}(\tau), r_{\max}(\mu)-r_{\max}(\tau))\cup [r_{\max}(\tau), r_{\max}(\mu)+r_{\max}(\tau))$. Note that $\beta$-curves are even functions. We have the strict condition for any $t\in  (r_{\max}(\tau)-r_{\max}(\mu),$ $r_{\max}(\mu) +r_{\max}(\tau))$. Then we can apply 
Proposition \ref{prop:bounded_crit} to complete the proof. 
\end{proof}
   
The above results can be used to prove the following statements.
\begin{prop}\label{lemma:degfd}
For any $\phi\in(0,1)$ and any symmetric $\nu$, we have
\begin{align}
    B_\phi\boxconv(\nu^{*(2)})&\prec_{0} (B_\phi\boxconv\nu)^{*(2)},\label{ineq:fixedd2}\\ 
    B_\phi\boxconv(\nu^{*(d)})&\prec (B_\phi\boxconv\nu)^{*(d)} \ \ \ \ \ \ \ \ \ \ \ \ \ \ \ \ \ \ \ \ \ \ \ \  \ \ \ \ \  \forall\ d> 2.\label{ineq:fixedd3} 
\end{align} 
\end{prop}
\begin{proof}
Inequality (\ref{ineq:fixedd2}) directly follows from Corollary \ref{coro:uniq_4}. Inequality (\ref{ineq:fixedd3}) clearly holds when $\phi=\frac{1}{2}$ or $\nu$ is trivial. We prove inequality (\ref{ineq:fixedd3}) by induction for both  $B_\phi$ and $\nu$ are non-trivial. 
 
(a). Consider the base case where $d=3$. Let $r_{\nu}=r_{\max}(\nu)$. Because convolution preserves degradation, by stringy tree lemma, or  inequality (\ref{ineq:fixedd2}) and Proposition \ref{prop:bounded_crit},
we have
\begin{align}\label{ineq:degs1}
    (B_\phi\boxconv\nu)*(B_\phi\boxconv(\nu^{*(2)})) \preceq (B_\phi\boxconv\nu)^{*(3)}.
\end{align}  
Next, from Corollary \ref{coro:uniq_4}, we have
\begin{align}
    B_\phi\boxconv(\nu^{*(3)}) \prec_s(B_\phi\boxconv\nu)*(B_\phi\boxconv(\nu^{*(2)})),
\end{align} where 
$$s=F_{\phi}(r_{\max}(\nu^{*(2)}))-F_{\phi}(r_{\max}(\nu))=F_{\phi}(2r_{\nu})-F_{\phi}(r_{\nu}).$$
The above two steps form a chain of degradation. By transitivity stated in Proposition \ref{prop:bounded_crit},  this implies 
$$ B_\phi\boxconv(\nu^{*(3)}) \prec_s(B_\phi\boxconv\nu)^{*(3)}. $$

 To prove the needed statement, it suffices to show that any of these two steps has strict inequality in $\beta$-curves for $t\in[0,s]$.
We apply Rule of Convolution to inequality (\ref{ineq:degs1}) and let $\ell= F_{\phi}(r_\nu)=r_{\max}(B_\phi\boxconv\nu) \in\supp{(B_\phi\boxconv\nu)}$, the strict condition holds for $$t\in\left[0,\tanh\frac{r_{\max}((B_\phi\boxconv\nu)^{*(2)})-\ell}{2}\right)=\left[0,\tanh\frac{F_{\phi}(r_\nu)}{2}\right).$$
Because $F_{\phi}(r_\nu)>s$ for both $B_\phi$ and $\nu$ non-trivial, the strict degradation statement is implied by Proposition \ref{lemma:strict_beta}.

(b) Assume inequality (\ref{ineq:fixedd3}) is proved for some $d=d_0\geq 3$. By induction assumption, we have the following chain similar to the base case. 
$$ B_\phi\boxconv(\nu^{*(d_0+1)})\preceq (B_\phi\boxconv\nu)*B_\phi\boxconv(\nu^{*(d_0)}) \prec (B_\phi\boxconv\nu)^{*(d_0+1)}.$$
In particular, we apply the third property in Proposition~\ref{prop:conv} to show strict degradation in the second step.  
Therefore, the induction step follows from Proposition \ref{prop:trans_deg_stdeg}.

(c) To conclude, inequality (\ref{ineq:fixedd3}) is proved for any $d\geq 3$. 
\end{proof}
Proposition \ref{lemma:degfd} implies the following  degradation relationships between $B_\phi\boxconv(\mathcal{Q}\mu)$ and $\mathcal{Q}(B_\phi\boxconv\mu)$ when $d$ is deterministic.

\begin{corollary}\label{cor:fixeddbasic}
If $\PP[d=d_0]=1$ for some fixed $d_0$, then for any $\phi\in(0,\frac{1}{2})$ and symmetric $\mu$, we have
\begin{align}
   \label{eq:basic2} &B_\phi\boxconv(\mathcal{Q}\mu)\prec_0 \mathcal{Q}(B_\phi\boxconv\mu)& \textup{ if } d_0=2,\\
     &B_\phi\boxconv(\mathcal{Q}\mu)\prec \mathcal{Q}(B_\phi\boxconv\mu)& \textup{ if } d_0>2.\label{eq:2_cor_2}
\end{align}
\end{corollary}
\begin{proof}
Recall that $\mathcal{Q}$ can be expressed as in Proposition \ref{prop:q_express}. The results in the corollary can be exactly obtained by letting   $\nu=B_{\delta}\boxconv \mu$ in Proposition \ref{lemma:degfd} and applying the commutativity of box convolution. 
\end{proof}

\subsection{Proof of Theorem \ref{lm:uniq_key}}


\begin{proof}
For brevity, we focus on non-trivial cases where $\mu$ is non-trivial and $\delta\neq \frac{1}{2}$. 
We first consider the deterministic $d$ case and fill in the gap for $d=2$.  Note that $\calQ$ preserves degradation.  By Corollary \ref{cor:fixeddbasic} and Proposition \ref{prop:bounded_crit}, we have the following chain of degradation.
$$B_\phi\boxconv \calQ^2\mu \prec_0 \calQ (B_\phi\boxconv \calQ\mu) \preceq {\mathcal{Q}^2(B_\phi\boxconv\mu)}.$$
In particular, the first step is obtained by replacing $\mu$ with $\calQ\mu$ in Corollary \ref{cor:fixeddbasic}. 
The above chain implies strict inequality in $\beta$-curves for $t\in(0,t_{\max}(B_\phi\boxconv \calQ^2\mu)]$. By non-trivial condition,  from Proposition \ref{lemma:strict_beta}, it remains to prove strict inequality of $\beta$-curves at $t=0$. 

To that end, we zoom in on the second step and apply Rule of Convolution to the first inequality of the following chain.
\begin{align}
    \calQ (B_\phi\boxconv \calQ\mu)&=(B_\phi\boxconv B_{\delta}\boxconv\calQ\mu)^{*(2)}\\
    &\preceq (B_\phi\boxconv B_{\delta}\boxconv\calQ\mu)*( B_{\delta}\boxconv\calQ(B_\phi\boxconv\mu)) \label{inequality:1_2}\\ 
    &\preceq ( B_{\delta}\boxconv\calQ(B_\phi\boxconv\mu))^{*(2)}={\mathcal{Q}^2(B_\phi\boxconv\mu)}.
\end{align}
Note that Corollary \ref{cor:fixeddbasic} and Proposition \ref{prop:conv} implies $$B_\phi\boxconv B_{\delta}\boxconv\calQ\mu\prec_0  B_{\delta}\boxconv\calQ(B_\phi\boxconv\mu),$$ 
and the non-trivial condition implies that  the $r_{\max}$ functions for both sides are different.
Therefore, we can choose $\ell=r_{\max}(B_\phi\boxconv B_{\delta}\boxconv\calQ\mu)$ for the Rule of Convolution and apply it to  inequality (\ref{inequality:1_2}), which leads to the needed strict condition at $t=0$.  
 
Now we consider general degree distributions. 
First for $\PP[d>2]=0$, recall that our formulation assumes non-trivial cases where $\PP[d\leq 1]<1$. We have $d=2$ with non-zero probability. Then the $r_{\max}$ function of $B_\phi\boxconv(\mathcal{Q}^2\mu)$ is identical to that of  its $d=2$ component. Thus, by linearity, 
our earlier proof for the deterministic $d=2$ case implies strict inequality of $\beta$-curves for the full range $t\in[0,t_{\max}(B_\phi\boxconv(\mathcal{Q}^2\mu))]$, 
and the needed statement is implied.

On the other case, we have $\PP[d>2]>0$. If $d$ is upper bounded by some fixed integer almost surely, we can let $d_0$ be the largest possible $d$ for such degree distribution and apply the same linearity argument to inequality (\ref{eq:2_cor_2}) to prove the statement. Otherwise, $d$ is unbounded, and we have strict inequality on $\beta$-curves for any $t<t_{\max}(\mathcal{Q}(B_\phi\boxconv\mu))=1$. Note that  $t_{\max}(B_\phi\boxconv(\mathcal{Q}^2\mu)=1-2\phi<1$. The statement follows from Proposition \ref{prop:bounded_crit}.   
\end{proof}

\subsection{Proof of Theorem \ref{thm:key_survey}}\label{sec:pt_key_survey}



We start by formulating two useful results. 
\begin{prop}\label{prop:sur_basic} 
For any $\phi\in(0,1)$ and any symmetric distributions $\mu$, $\nu$, let $s_{\min}\triangleq F_{\phi}(r_{\max}(\mu))-r_{\max}(\nu)$. We have
\begin{align}\label{ineq:sur_basic_1}
    B_\phi\boxconv(\mu*\nu)\prec (B_\phi\boxconv\mu)*\nu 
\end{align} 
if $s_{\min}< 0$, and 
\begin{align}\label{ineq:sur_basic_2}
    B_\phi\boxconv(\mu*\nu)\prec_{s_{\min}} (B_\phi\boxconv\mu)*\nu
\end{align} 
otherwise.
\end{prop}
\begin{proof} 
Similar to the basic setting, our technique is to prove strict inequalities for beta-curves by forming a chain of degradation. 
First observe that $B_\phi\boxconv\nu\prec \nu$ for any $\phi\in(0,1)$.  By convolving $B_\phi\boxconv\mu$ on both sides and note that degradation is convolution-preserving, we have
\begin{align}\label{ineq:sur_proof_0}
 (B_\phi\boxconv\mu)*(B_\phi\boxconv\nu) \preceq (B_\phi\boxconv\mu)*\nu.
\end{align}
On the other hand, Corollary \ref{coro:uniq_4} implies following step, which completes the chain.
\begin{align}\label{ineq:sur_proof_1}
 B_\phi\boxconv(\mu*\nu)  \preceq (B_\phi\boxconv\mu)*(B_\phi\boxconv\nu) 
\end{align}

When $s_{\min}< 0$, we have $r_{\max}(B_\phi\boxconv\mu)=F_{\phi}(r_{\max}(\mu))<r_{\max}(\nu)$.
Hence, we can apply the third statement in Proposition \ref{prop:conv} to obtain a strict version of inequality (\ref{ineq:sur_proof_0}), i.e., $(B_\phi\boxconv\mu)*(B_\phi\boxconv\nu) \prec (B_\phi\boxconv\mu)*\nu$. Then, inequality (\ref{ineq:sur_basic_1}) follows from the transitivity statement in Proposition \ref{prop:trans_deg_stdeg}. 

More generally, the strict version of inequality (\ref{ineq:sur_proof_0}) can be written as $(B_\phi\boxconv\mu)*(B_\phi\boxconv\nu) \prec_{s_{\min}} (B_\phi\boxconv\mu)*\nu$ according to Proposition \ref{prop:conv}. Then, inequality \eqref{ineq:sur_basic_2} is proved by the second statement in Proposition \ref{prop:bounded_crit}.
\end{proof}

Next, consider the case of $d=1$, so that  $\calQ_{\textup{s}}\mu=(B_{\delta}\boxconv\mu)*\mu_{\textup{s}}$. 
By induction, one can derive the following result  (proof in Appendix \ref{app:pl_surimp}).
\begin{prop}\label{lemma:sur_imp} 
For $\phi\in(0,1)$ and $d=1$, let $r_{\textup{s}}\triangleq r_{\max}(\mu_{\textup{s}})$, $s_0\triangleq F_{\phi}(r_{\max}(\mu))$, then for any $k\in\mathbb{N}$, we have
\begin{align}
    B_\phi\boxconv(\calQ^k_{\textup{s}}\mu)\prec \calQs^{k}(B_\phi\boxconv\mu)\nonumber
\end{align} 
if $s_k\triangleq F_{\delta}(s_{k-1})-r_{\textup{s}}<0$, and 
\begin{align}
    B_\phi\boxconv(\calQ^k_{\textup{s}}\mu)\prec_{s_k} \calQs^{k}(B_\phi\boxconv\mu)\nonumber
\end{align} 
otherwise.
\end{prop}
One can show that there is a finite $k$ for $s_k<0$. For example, as a rough estimate, we have $s_0\leq F_{\phi}(+\infty)<+\infty$ and $s_k\leq s_{k-1}-r_\textup{s}$ for any positive $s_k$. Because $\mu_{\textup{s}}$ is non-trivial, we also have $r_{\textup{s}}>0$. Hence, we can find $k\leq 1+{s_0}/{r_{\textup{s}}}$ for strict degradation to hold, which gives the needed statement.

\begin{proof}[Proof of Theorem~\ref{thm:key_survey}]
First, consider the case $\PP[d>1]=0$. Recall that the theorem statement assumes $\PP[d=0]<1$, so that
we have $d=1$ with non-zero probability. Consequently, the $\beta$-curve analysis is dominated by
the $d=1$ component, and strict inequalities in the full range of $t$ can be obtained using Proposition \ref{lemma:sur_imp}. 

Formally, 
$\mathcal{Q}_{\textup{s}}$ can be written as a linear combination of two operators, each corresponds to the BP operator for a deterministic $d\in\{0,1\}$. For brevity, we denote them by $\mathcal{Q}_{\textup{s},0}$ and $\mathcal{Q}_{\textup{s},1}$. 
Then, $\mathcal{Q}_{\textup{s}}^k$ can be expanded into a linear combination of $2^k$ chains, and  each side of inequality  (\ref{sdeg:sur_key1}) can be  decomposed into the following $2^k$ terms.  
\begin{align*}
   \sum_{d_1,...,d_k\in\{0,1\}} 
   B_{\phi}\boxconv \mathcal{Q}_{\textup{s},d_1}\mathcal{Q}_{\textup{s},d_2}...\mathcal{Q}_{\textup{s},d_k}\mu \cdot \prod_{j=1}^{k} \PP[d=d_j]&
   \prec \sum_{d_1,...,d_k\in\{0,1\}} \mathcal{Q}_{\textup{s},d_1}\mathcal{Q}_{\textup{s},d_2}...\mathcal{Q}_{\textup{s},d_k}(B_{\phi}\boxconv\mu) \cdot \prod_{j=1}^{k} \PP[d=d_j]
\end{align*}
Each corresponding terms in the above inequality can be compared individually. Recall that $\calQ_{\textup{s}}\mu=(\calQ\mu)*\mu_{\textup{s}}.$ By Corollary \ref{cor:fixeddbasic} and Proposition \ref{prop:sur_basic}, 
the inequality $B_\phi\boxconv(\mathcal{Q}_{\textup{s}}\mu)\preceq \mathcal{Q}_{\textup{s}}(B_\phi\boxconv\mu)$ holds for any BP operator $\mathcal{Q}_{\textup{s}}$, which includes $\mathcal{Q}_{\textup{s},0}$ and $\mathcal{Q}_{\textup{s},1}$. By applying this inequality recursively, we have the following individual comparisons. 
\begin{align*}
   B_{\phi}\boxconv \mathcal{Q}_{\textup{s},d_1}\mathcal{Q}_{\textup{s},d_2}...\mathcal{Q}_{\textup{s},d_k}\mu &
   \preceq  \mathcal{Q}_{\textup{s},d_1}\mathcal{Q}_{\textup{s},d_2}...\mathcal{Q}_{\textup{s},d_k}(B_{\phi}\boxconv\mu) 
\end{align*}
Among all terms, the ones with all BP operators corresponds to $d=1$ achieves strict degradation due to Proposition \ref{lemma:sur_imp}. Note that these terms have the largest $t_{\max}$ values on both sides, because $t_{\max}(\mathcal{Q}_{\textup{s},d}\nu)$ is non-decreasing with respect to $d$ and $t_{\max}(\nu)$ for any symmetric $\nu$.
The interval $t\in [0, t_{\max}(B_\phi\boxconv(\mathcal{Q}_\textup{s}^k\mu))]$ must be contained within the range where the $\beta$-curve inequality between these two terms is strict. 
The $\PP[d=0]<1$ condition ensures that this  gap 
has non-zero weights in the overall $\beta$ functions. Then  
inequality (\ref{sdeg:sur_key1}) is proved by Proposition \ref{lemma:strict_beta}. 

~
It remains to consider general degree distributions with $\PP[d>1]>0$. 
 Note that in the basic setting, we have essentially proved that if 
 $\PP[d>1]>0$, then 
\begin{equation}\label{eq:si_7}
	B_{\phi}\boxconv\mathcal{Q}\mu\prec_0  \mathcal{Q}(B_{\phi}\boxconv\mu). 
\end{equation}
Specifically, the above inequality is directly implied by Theorem \ref{lm:uniq_key} if $\PP[d> 2]>0$. In the other case, we have that the $r_{\max}$ function of $B_{\phi}\boxconv\mathcal{Q}\mu$ is dominated by its $d=2$ component. Thus, inequality (\ref{eq:basic2}) implies non-zero gaps in $\beta$-curves for all $t\in(0,t_{\max}(B_{\phi}\boxconv\mathcal{Q}\mu)]$, which proves 
 inequality (\ref{eq:si_7}). 

Let 
$r_{\textup{s}} = r_{\max}(\mu_{\textup{s}})$, $r_Q = r_{\max}(\calQ \mu)$, $\tilde r_Q =
r_{\max}(\calQ (B_\phi \boxconv \mu))$.
We first apply Proposition \ref{prop:sur_basic} and then the Rule of Convolution 
to obtain
\begin{equation}\label{eq:si_6}
	B_{\phi}\boxconv\mathcal{Q}_{\textup{s}}\mu\preceq (B_{\phi}\boxconv\mathcal{Q}\mu)*\mu_{\textup{s}}\preceq  
\mathcal{Q}_{\textup{s}}(B_{\phi}\boxconv\mu).
\end{equation} 
Consider the first step of inequality~\eqref{eq:si_6}, 
the statement of Proposition \ref{prop:sur_basic} implies that the gap between the $\beta$-curves on both sides is strict for any $t=\tanh{|s|\over 2}$ with 
\begin{equation}\label{eq:si_8}
	F_\phi(r_Q) - r_{\textup{s}} < s < F_\phi(r_Q) + r_{\textup{s}}. 
\end{equation}
Then by the rule of convolution, the $\beta$-curve inequality for the second step is strict for $t=\tanh{|s|\over 2}$ with
\begin{equation}\label{eq:si_9}
	- r_{\textup{s}} < s < \tilde r_Q - r_{\textup{s}}.
\end{equation}
Note that inequality~\eqref{eq:si_7} implies that $\tilde r_Q > r_{\max}(B_\phi \boxconv \calQ \mu) =
F_\phi(r_Q)$. Thus,~\eqref{eq:si_8} and~\eqref{eq:si_9} cover all $0 \le s < F_\phi(r_Q) + r_{\textup{s}}$,
concluding the proof of 
	\begin{align}\label{ineq:fpd_final}
	    B_{\phi}\boxconv\mathcal{Q}_{\textup{s}}\mu\prec
	\mathcal{Q}_{\textup{s}}(B_{\phi}\boxconv\mu).
	\end{align} 
\end{proof}

\section{Implications}\label{sec:imp}

\subsection{Implications on Robust Reconstruction}\label{sec:rob}

A variant of BOT was formulated in \cite{janson2004}, called robust reconstruction, where all leaf observations are obtained through some identical noisy channels. The estimation problem 
is to infer the root variable given the tree structure and the noisy leaf observations. 
Robust reconstruction for the Ising model case was studied in \cite{MNS16}, which can be described as  
in the setting of Fig.~\ref{fig:main}, except that we select 
$\mu_{\textup{s}}$ to be the trivial distribution $\delta_0$. 
As derived earlier, the LLR distributions for this formulation at depth $h$ 
is exactly given by $
\calQ^h \mu_{\textup{r}}$, where $\calQ$ is the BP operator defined in Section \ref{sec:intro}. 

Theorem \ref{thm: uniq} (or Corollary \ref{thm: uniq_as} for the asymmetric case) 
 implies the following statement.


\begin{thm}\label{thm:robust_conv}
For any fixed $P_d$ and $\delta\in[0,1]$, the distributions in the following classes all exist and are identical, unless $d=1$ a.s. and $\delta\in\{0,1\}$.
  
    (a) The limiting LLR distribution for the basic setting.  
    
    (b) The limiting LLR distribution for robust construction with any non-trivial initialization.  
    
    (c) The dominant BP fixed point, i.e., a fixed point $\mu^*$ of  $\mathcal{Q}$ where any other fixed point $\mu$ satisfies 
    $\mu\preceq\mu^*$.
\end{thm}

\begin{remark}
In \cite{MNS16}, it was conjectured that the error probability for the maximum likelihood estimator is independent of the observation channels when $h\rightarrow\infty$, as long as their channel capacity is non-zero. 
Note that this error probability can be written as an expectation over the LLR distribution (see Proposition \ref{prop:rem:beta_err}).  The unique convergence stated in Theorem \ref{thm:robust_conv} provides a positive proof to this conjecture. More generally, the same guarantee holds for any quantity that can be written as the expectation of a bounded continuous function on $(-\infty,+\infty]$, such as mutual information and Bayesian estimation errors under different prior distributions. 
Theorem \ref{thm:robust_conv} also provides a proof of Proposition 1 in \cite{9517800}.
\end{remark}






\subsection{Boundary Irrelevance for Broadcast with Survey}\label{sec:bi_p}

We first present a definition of boundary irrelevance in terms of LLR distributions.
\begin{defn}
For any degree distribution $P_d$, $\delta\in[0,1]$, and symmetric non-trivial survey distribution $\mu_{\textup{s}}$,
let $\calQ_\textup{s}$ be the associated BP operator. We say boundary irrelevance (BI) is satisfied if both $ \mu_{(h)}\triangleq \calQ_\textup{s}^h B_0$ and $    \underline{\mu}_{(h)} \triangleq\calQ_\textup{s}^h B_{\frac{1}{2}}$ weakly converges to the same distribution on domain $(-\infty,+\infty]$ as $h\rightarrow \infty$.
\end{defn}

Note that $\mu_{(h)}$ represents the LLR distribution for estimation with full leaf information,
and $\underline{\mu}_{(h)}$ represents the corresponding LLR distribution with no leaf
information (see Fig.~\ref{fig:main}). The above definition essentially states that ignoring
leaf information will not affect estimation as $h\rightarrow\infty$, which is consistent with the
notation of BI defined in the literature. In particular, one can show that our definition is
equivalent to the version in \cite{ACGP21}, and is stronger than the error-probability guarantee
in \cite{kms16}. Therefore, we present the following Theorem, 
which is a direct consequence of Theorem \ref{thm: uniq_s}, which simultaneously resolves 
Conjecture 1 in \cite{ACGP21} and Conjecture 1 in \cite{kms16}.


\begin{thm}\label{thm:bi}
BI holds for any combination of $P_d$, $\delta\in[0,1]$, and symmetric non-trivial $\mu_{\textup{s}}$.
\end{thm}

\subsection{Uniqueness and Convergence in the Large $d$ Limit}\label{sec:large_d}

We consider a recursion process characterized by the following operator $\calQ_{\textup{L}}$: 
For any symmetric $\mu_{\textup{s}}$ and any distribution $P_{\overline{d}}$ on domain
$[0,+\infty]$, we set
 $$\calQ_{\textup{L}}\mu=\mathbb{E}_{\overline{d} \sim P_{\overline{d}}}\left[\mathcal{N}\left(\overline{d}\cdot V_\mu\right)\right]*\mu_{\textup{s}},$$
where $\mathcal{N}(s)\triangleq \mathcal{N}(\frac{s}{2},s)$ for any $s\in[0,+\infty]$,  $V_\mu\triangleq\EE_{R\sim\mu}\left[4{\tanh\left(\frac{R}{2}\right)}\right]$ for any symmetric $\mu$, and $\mathbb{E}_{\overline{d} \sim P_{\overline{d}}}\left[\cdot \right]$ represents a  mixture of distributions over the law of $\bar{d}$.  
This operator was considered in \cite{mx15}
 as a limit of $\calQ_{\textup{s}}$ (or $\calQ$ when $\mu_{\textup{s}}$ is trivial) for $\delta\rightarrow \frac{1}{2}$, 
 where the degree distribution $P_d$ is parameterized by $\delta$, and $\overline{d}\triangleq d(1-\delta)^2$ converges in distribution to $P_{\overline{d}}$ on domain $[0,+\infty]$.  
Similar to earlier sections, one can define the fixed point equation to be $\mu=\calQ_{\textup{L}}\mu$ and define BP recursion as $\mu_{(h+1)}=\calQ_{\textup{L}}\mu_{(h)}$. 
The operator $\mathcal{Q}_{\textup{L}}$ can also be defined for asymmetric distributions by setting $V_{\mu}=\mathbb{E}_{R\sim\frac{1}{2}(\mu+\mu^-)}[4\tanh^2\left(\frac{R}{2}\right)]$, where $\mu^-$ is the complement of $\mu$ (see Definition \ref{def:llrd}). 

To extend our earlier results to $\calQ_{\textup{L}}$, we prove its contractivity in terms of the degradation index. In particular, note that the contraction implied by the BP operator is non-multiplicative, a careful investigation is needed to show that strict inequalities in $\beta$-curves are maintained in the limit of large $d$.   
We present the results in the following theorem, and provide a proof in Appendix \ref{app:pl_ld}. 

\begin{thm}\label{thm:larged}
Consider the large $d$ regime defined by any $P_{\overline{d}}$ and any symmetric $\mu_{\textup{s}}$.
\begin{enumerate}
    \item There is at most one unique non-trivial BP fixed point, and it is symmetric. (Uniqueness of non-trivial BP fixed point)
    \item Non-trivial symmetric BP fixed point exists if and only if either  $\EE\left[\overline{d}\right]\in(1,+\infty]$ or $\mu_{\textup{s}}$ is non-trivial. (Existence of non-trivial BP fixed point)
    \item BP recursion with any non-trivial initialization converges to the  unique non-trivial fixed point if it exists, or to the trivial $\delta_0$ otherwise. (Independence of Convergence and Initialization)
    \item When $\mu_{\textup{s}}$ is non-trivial, the above convergence statement also applies for the trivial initialization. 
(Boundary Irrelevance)
\end{enumerate}
\end{thm}
\begin{remark}
The uniqueness statement in the above theorem resolves Conjecture 2.6 in \cite{mx15}, by applying the special case where $P_{\overline{d}}$ is a delta distribution 
and $\mu_{\textup{s}}=B_{\alpha}$. 
Generally, our formulation does not assume $d$ scales with $(1-2\delta)^{-2}$ in high probability. The sublinear and superlinear components of $P_d$ are naturally captured  by non-zero mass points in $P_{\overline{d}}$ at $0$ and $+\infty$.
\end{remark}

\subsection{Full Characterizations of Accuracy and Entropy in Stochastic Block Model}
\label{sec:sbm_} 


Consider a $2$-SBM problem with a set of $n$ vertices $\mathcal{V}=\{v_1,v_2,...,v_n\}$. Let $X_v$ denote the label on vertex $v$, and $X=(X_{v_1}, X_{v_2}, ..., X_{v_n})$ denote the collection of all labels. The entries of $X$ are i.i.d. Ber$(\frac{1}{2})$.  A random graph $G$ 
 is generated based on the labels according to the rules of SBM. Formally, we represent $G$ using its adjacency matrix, i.e., $G_{ij}=1$ if and only if $v_i$ and $v_j$ are adjacent. 
 Then all $\{G_{ij}\}_{i\leq j}$ are independent Bernoulli random variables, with 
 \begin{align*}\mathbb{P}[G_{ij}=1]=\begin{cases}
  \frac{a}{n} & \textup{if } X_{v_i}=X_{v_j}\textup{ and }i\neq j,
 \\\frac{b}{n} & \textup{if } X_{v_i}\neq X_{v_j},\\
 0 & \textup{otherwise}.
 \end{cases}
 \end{align*}
 The goal in this setting is to design algorithms that use  the random graph 
 to produce an estimate of $X$. 
 
  There are two main quantities of interests. For any estimator $\hat{X}$, its \emph{estimation accuracy}, denoted by $\textup{acc}_n(\hat{X})$, is defined as follows \cite{MNS16}.
\begin{align}\label{eqdef:acc}
    \textup{acc}_n(\hat{X})\triangleq \frac{1}{2}+\left|\frac{1}{n}\sum_{v\in\mathcal{V}}\left|X_v-\hat{X}_v\right|-\frac{1}{2}\right|. 
\end{align}
In particular, note that the conditioned graph distribution is invariant under a global bit flip of hidden labels. No algorithm can achieve a non-trivial estimation in the expected number of correctly estimated labels. 
The accuracy defined in equation (\ref{eqdef:acc}) captures the correlation between the partitions induced by the labels,  which removes the global bit-flip effect.


Note that $\textup{acc}_n(\hat{X})$ is random. The quantity $p_{\textup{G}}(a,b)$ was introduced in \cite{MNS16} to measure the performance of estimators, defined as the maximum accuracy that can be achieved by any estimator for large $n$ with non-zero probability. 
Formally, let $f$ denote the function that takes the observations and returns $\hat{X}$, $p_{\textup{G}}(a,b)$  is defined as follows. 
\begin{align*}
p_{\textup{G}}(a,b)=\adjustlimits\lim_{\epsilon\rightarrow 0} \limsup_{n\rightarrow \infty} \sup_{f} \sup \left\{{p}\ \Big|\  \mathbb{P}\left[\textup{acc}_n(\hat{X})\geq p\right]\geq \epsilon\right\}. 
\end{align*}
The problems of interests are to characterize $p_{\textup{G}}(a,b)$ and to prove whether it can be achieved by any algorithm with high-probability. Both were only resolved when $a$ and $b$ satisfy certain conditions. 
However, with the leaf-independence result  proved in Section \ref{sec:rob}, the proofs in \cite{mns15, MNS16, mns_proof} 
can be extended to all regimes.

The other quantity of interest is the so called \emph{SBM entropy}, denoted by $\mathcal{H}(a,b)$,
which is defined to be the limit of the normalized conditional entropy of all labels $X$ given the
graph $G$, as $n\rightarrow\infty$: 
$$ \mathcal{H}(a,b)\triangleq\lim_{n\rightarrow\infty}\frac{1}{n} H(X|G).$$
The SBM entropy also characterizes the normalized mutual information between the labels and the
graph defined as 
$$ \mathcal{I}(a,b) \triangleq\lim_{n\rightarrow\infty}\frac{1}{n} I(X;G)\,.$$
Similar to the accuracy metric, it was an open problem to characterize $\mathcal{H}(a,b)$ and $\mathcal{I}(a,b)$ for all parameter values using BP fixed points. It was pointed out in \cite{ACGP21} that a complete characterization can be obtained once the BI result stated in Section \ref{sec:bi_p} is proved.   

To summarize, we have the following theorem, which strengthens Theorem 2.9 in \cite{MNS16} and Theorem 1 in \cite{ACGP21}.  


\begin{thm}
For any $a$ and $b$, 
\begin{enumerate}
    \item we have
$p_{\textup{G}}(a,b)=p({\mu^*_{a,b}})$, where ${\mu^*_{a,b}}$ is the dominant BP fixed point (as specified in Theorem \ref{thm:robust_conv}) for the broadcast on tree problem  with $P_d=\textup{Pois}(\frac{a+b}{2})$ and $\delta=\frac{b}{a+b}$, and 
    $p({\mu^*_{a,b}})\triangleq {\mu^*_{a,b}}((0,+\infty])+\frac{1}{2}{\mu^*_{a,b}}(\{0\})$; 
\item there is a polynomial time algorithm that achieves  $p_{\textup{G}}(a,b)$ with high probability, i.e., with $\textup{acc}_n(\hat{X})$ converges in probability to $p_{\textup{G}}(a,b)$ as $n\rightarrow\infty$;
\item we have 
$\mathcal{H}(a,b)=\log 2-\mathcal{I}(a,b)=\int_{0}^1 \mathbb{E}_{R\sim\mu_{\epsilon,a,b}}[\log(2\cosh{\frac{R}{2}})-\frac{R}{2}\tanh(\frac{R}{2}) ]d\epsilon$, where $\mu_{\epsilon,a,b}$ is the unique BP fixed point for the broadcast with survey setting with $P_d=\textup{Pois}(\frac{a+b}{2})$, $\delta=\frac{b}{a+b}$, and BEC survey $\mu_{\textup{s}}=\epsilon B_{0}+(1-\epsilon) B_{\frac{1}{2}}$.  
\end{enumerate}   
\end{thm}

\subsection{Stochastic Block Model with Side Information}
\label{sec:sbm_side}   

Consider a variant of the $2$-SBM formulation, where the estimator has additional access to a noisy version of all hidden labels. Similar to the broadcast with survey setting, each label is observed through an independent symmetric channel, and we denote  
their LLR distribution by $\mu_{\textup{s}}$. 

In the presence of this side information, a different notion of accuracy was considered in the literature. In \cite{kms16,mx15}, the authors considered estimators that 
asymptotically maximizes 
the expected fraction of correctly estimated labels. Formally, we denote this function by  
 $p_n(\hat{X})$, which can be defined by the following equation. 
\begin{align}\label{eq:est_acc_noisy}
p_{n}(\hat{X})\triangleq \mathbb{E}\left[\frac{1}{n}\sum_{v\in\mathcal{V}}\left|X_v-\hat{X}_v\right|\right]=\frac{1}{n}\sum_{v\in\mathcal{V}}\PP\left[X_v=\hat{X}_v\right].
\end{align}

The estimation accuracy defined in equation (\ref{eq:est_acc_noisy}) can be maximized by applying  the ML estimator individually for each $X_v$. However, the ML estimator becomes computationally intractable when the graph is large as it relies on global information. Therefore, local algorithms have been studied, and they have been conjectured to be optimal 
\cite{kms16, mx15}. 
In particular, for any fixed parameter $t\in\mathbb{N}_+$, an algorithm is called $t$-local if it estimates each $X_v$ only using the information within the subgraph induced by vertices with a distance from $v$ less than $t$. 
Such local information resembles the distribution of local observation in the broadcast with survey setting  
as $n\rightarrow \infty$ for fixed $a$ and $b$, up to a graph isomorphism. Hence, one can estimate each $X_v$ using the same belief propagation 
whenever the graph is locally tree-like. 

Local BP is asymptotically optimal among local algorithms. 
We present the following theorem, which states that there are no gaps between the estimation accuracies of 
local and global algorithms.\footnote{In certain parts of \cite{mx15}, a generalized setting was considered, where $a$, $b$ are $n$-dependent. It is clear that the same generalization is not considered in their Conjecture 1, otherwise the stated limits may not converge. However, one can still prove a similar optimality result using the BI presented in Section \ref{sec:bi_p} and \ref{sec:large_d}, stated in terms of the absolute difference between estimation accuracies. More generally, this asymptotical optimality can hold whenever the local tree-like condition is satisfied for large $n$.  }
\begin{thm}\label{thm:2sbn_side}
For 2-SBM with any fixed $a$, $b$ and side information generated based on any non-trivial $\mu_{\textup{s}}$, we have
\begin{align}\label{eq:sbm_s_opt}
    \lim_{t\rightarrow\infty}\lim_{n\rightarrow\infty}p_n(\hat{X}_{\textup{BP}}^{(t)})= \lim_{n\rightarrow\infty}p^*_n,
\end{align}
where $\hat{X}_{\textup{BP}}^{(t)}$ is any estimator that runs local BP with parameter $\delta=\frac{b}{a+b}$, 
and $p^*_n$ is the optimal estimation accuracy over all estimators.
\end{thm}

\begin{remark}
Conjecture 1 in \cite{kms16} corresponds to taking $\mu_{\textup{s}}=pB_{0}+(1-p)B_{\frac{1}{2}}$ for $p>0$ (binary erasure channels). Conjecture 2.5 in \cite{mx15} corresponds to taking $\mu_{\textup{s}}=B_{\alpha}$ for $\alpha\in[0,\frac{1}{2})$ (binary symmetric  channels). Thus, Theorem \ref{thm:2sbn_side} 
closes both of them. 
\end{remark}

\begin{proof}[Proof of Theorem \ref{thm:2sbn_side}]
The proof can be established by first connecting SBM to the broadcast with survey setting, similar to the approach presented in \cite[Section 2.4]{mx15}.
Formally, one can show that Lemma 3.7 and Lemma 3.9 in \cite{mx15} holds for general symmetric $\mu_{\textup{s}}$, which 
bound the accuracies on both sides of equation \eqref{eq:sbm_s_opt} using limiting  distributions of BP recursion. I.e., for any fixed $t$, $a$, $b$, and symmetric $\mu_{\textup{s}}$, we have  
\begin{align*}
     \lim_{n\rightarrow\infty}p_n\left(\hat{X}_{\textup{BP}}^{(t)}\right)= p\left(\mu_{(t)}\right),\\
      \limsup_{n\rightarrow\infty}p_n^*\leq  p\left(\underline{\mu}_{(t)}\right),
\end{align*}
where $\mu_{(t)}$ and $\underline{\mu}_{(t)}$ are the LLR distributions in the BP recursions  initialized by $\mu_{(0)}=B_{0}$ and $\underline{\mu}_{(0)}=B_{\frac{1}{2}}$. 
As a consequence, the needed equality condition 
is implied if all limiting distributions are identical, which is essentially the BI condition stated in Section \ref{sec:bi_p}.  
In this work, we proved that the BI condition holds for all regimes (Theorem \ref{thm:bi}), which completes the proof of optimality of local BP algorithms.
\end{proof}

\ifgeneral
\section{Further Extensions}\label{sec:gllr}


Consider an infinite tree channel generated by the following 
class of structures. 
\begin{defn}[Element Tree]
An element tree $T$ is defined to be a tuple that consist of

(a) a finite rooted tree;

(b) a parameter $\delta_e$ for each edge $e$; 

(c) a subset of leaves $V_{\textup{G}}$  called the \emph{growing points};

(d) a symmetric distribution $\mu_{v}$ for each vertex $v\notin V_{\textup{G}}$. 
\end{defn}

The tree network is generated recursively based on a distribution of element trees, denoted by $P_{T}$. 
Initially, the tree has a single growing point at the root vertex. 
Then for each step $h\in\mathbb{N}_+$, all growing points are replaced by some random subtrees, that are i.i.d. with distribution $P_{T}$
.\footnote{For readers interested in models where the tree can grow from non-leaf vertices, any such growing point can be treated equivalently as a leaf vertex connected to the tree through an edge with $\delta_e=0$.} 
Given the tree structure, 
a single bit $X_0$ is broadcast downlink from the root, and 
each edge $e$ serves as a BSC channel with crossover probability $\delta_e$. We consider the process of estimating $X_0$ for the tree channel generated at each step $h$, where the estimator has  access to the tree structure, the variables on the boundary of the network (i.e., the growing points) and survey observations at all other vertices each through a channel characterized by the corresponding $\mu_{v}$. Generally, we also allow that the observations on the growing points are noisy, i.e.,   independently though identical binary-input channels.

Similar to earlier formulations, the LLR distribution of this problem can be determined through BP recursion. 
For any fixed element tree $T$, we denote its BP operator 
by $\calQ_T$, which is defined as follows. 
If $T$ is a single growing point, then $\calQ_T$ is the identity map. Otherwise, let $o$ denote the root vertex, $e_1,...,e_d$ denote its incident edges, 
and  $T_1,...,T_d$ denote the corresponding subtrees rooted at the $d$ children vertices, then $\calQ_T\mu$ is given by the law of 
\begin{align}
    R \triangleq R_{\textup{s}}+ \sum_{u=1}^d  F_{\delta_{e_u}}(\hat R_u)\,, 
\end{align} 
where $R_{\textup{s}}\sim \mu_o$, $\hat R_u\sim (1-\delta_{e_u}){\hat\mu_u}+\delta_{e_u}\hat\mu_u^-$ (all jointly independent), $\hat\mu_{u}\triangleq\calQ_{T_u}\mu$, and $\hat\mu_u^-$ is the complement (see Definition \ref{def:llrd}) of $\hat\mu_u$. For a general distribution $P_T$, we define $\calQ_{P_T}\mu\triangleq\mathbb{E}_{P_T}[\calQ_{T}\mu]$, and we say $\mu$ is a BP fixed point if $\mu=\calQ_{P_T}\mu$. 

\begin{prop}
Let $\mu_{(h)}$ be the LLR distribution of the tree model at step $h$, we have $\mu_{(h+1)}=\calQ_{P_T}\mu_{(h)}$. 
\end{prop}
\begin{proof}
The derivation follows the same steps in Section \ref{subsec:der}, as BP applies to any tree graphical models.   
\end{proof}

The above framework generalizes the settings in earlier sections, which allows us to cover several models of interest as special cases (see Fig. \ref{fig:threegraphs}).
We provide a full characterization for the uniqueness of BP fixed points in all cases.
\begin{figure}
     \centering
     \begin{subfigure}[b]{0.3\textwidth}
         \centering
         \includegraphics[width=\textwidth]{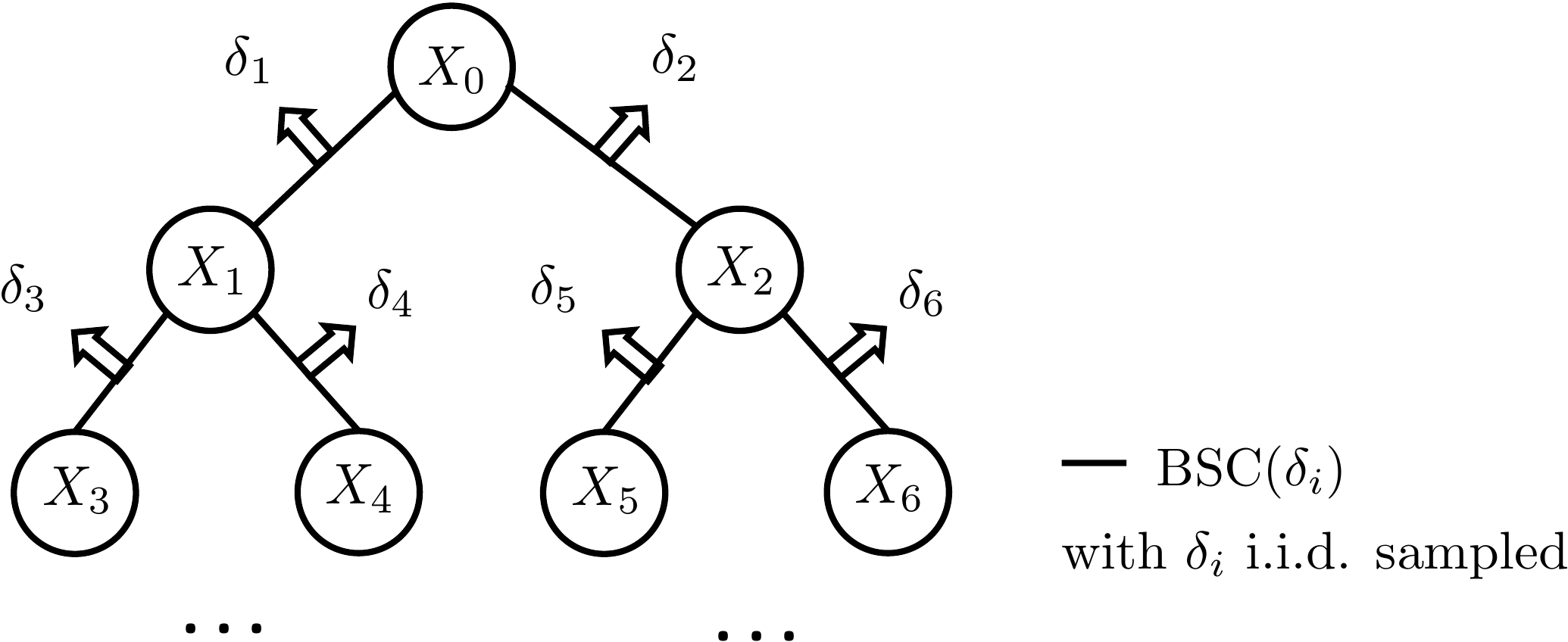}
         \caption{Crossover probabilities being  i.i.d. random and known to the estimator. \newline}
     \end{subfigure}
     \hfill
     \begin{subfigure}[b]{0.3\textwidth}
         \centering
         \includegraphics[width=\textwidth]{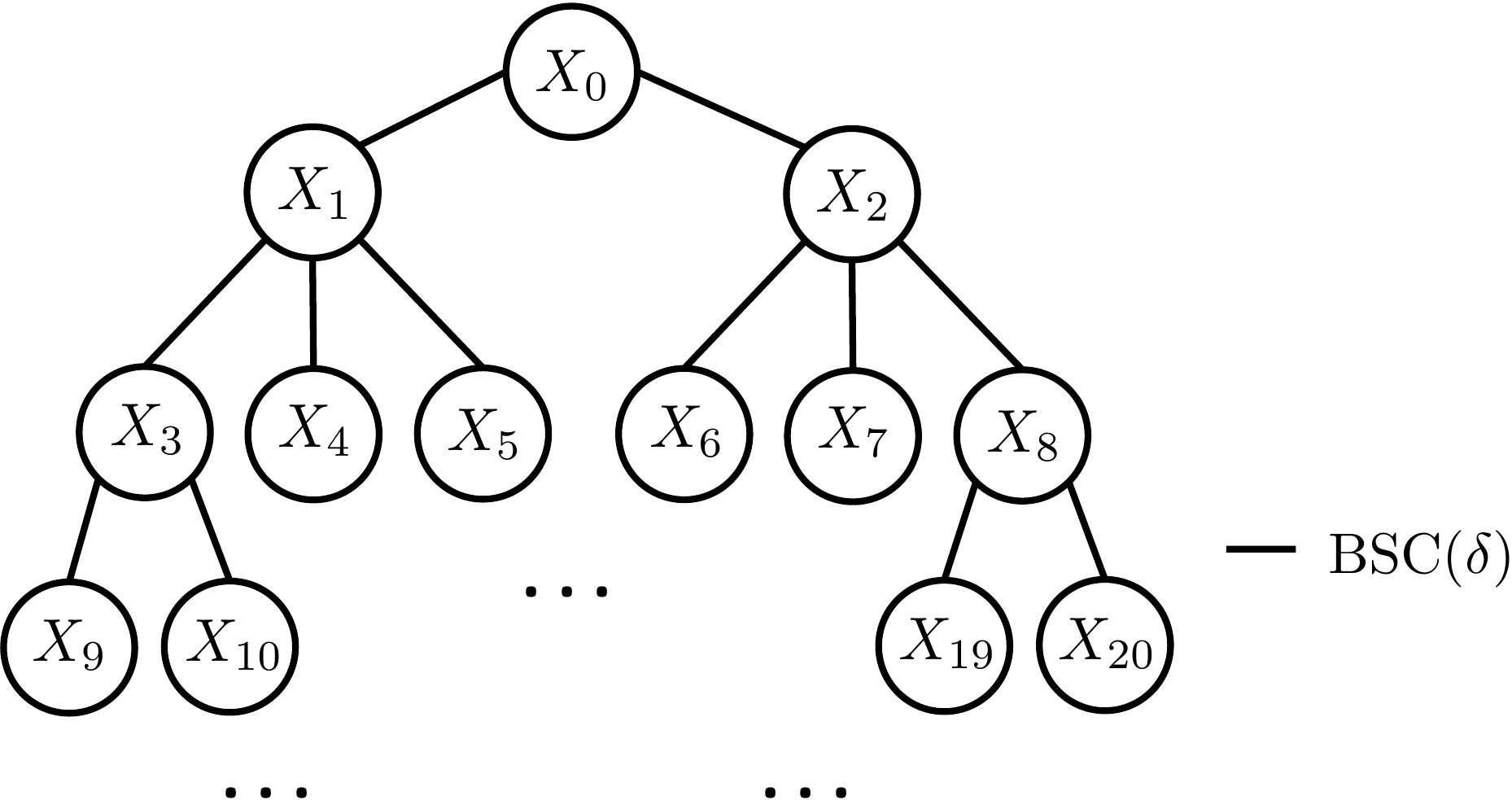}
         \caption{Periodic, i.e., the degree of each vertex depends periodically on their depths.\newline}
     \end{subfigure}
     \hfill
     \begin{subfigure}[b]{0.3\textwidth}
         \centering
         \includegraphics[width=\textwidth]{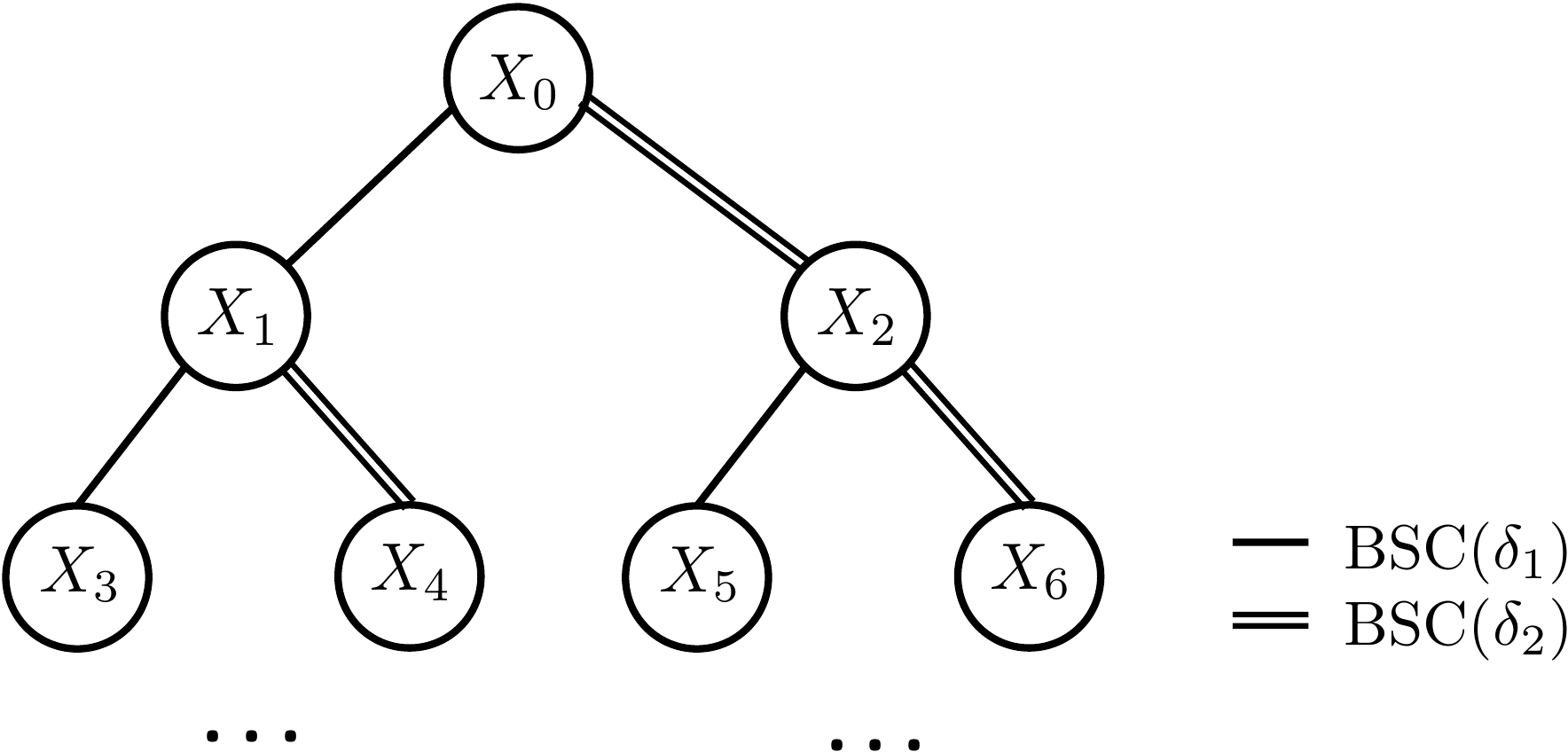}
         \caption{Nonisotropic, i.e., each vertex is connected through a fixed list of channels, each with  distinct crossover probabilities.}
     \end{subfigure}
        \caption{
An Illustration of three different classes of Ising models.   All presented examples can be generated from certain element-tree distributions. }
  \label{fig:threegraphs}
\end{figure} 
For ease of discussion, we assume 
non-zero capacity for all edges, i.e., $\delta_e\neq \frac{1}{2}$, otherwise the channel can be simplified by removing the corresponding subtrees. 
Among all tree distributions, there is one subclass where the recursion is easy to analyze, but the results need to be stated  separately. 
\begin{defn}\label{def:simple_path}
We say an element tree \emph{simple}, if the tree network reduces to a noiseless simple path, i.e.,
 there is exactly one growing point, all edges on the path
from the growing point to the root has $\delta_e\in\{0,1\}$, 
 and all  
survey distributions $\mu_v$
are trivial. We say an element tree distribution $P_T$ is \emph{trivial} if $T$ is simple w.p.1.
\end{defn}

This subclass of trivial tree distributions corresponds to the exceptional cases specified in Theorem \ref{thm:main} in the basic formulation. Generally, $\calQ_{P_T}$ is the identity operator for any trivial $P_T$. 
Given this definition, we state our main theorem for  non-trivial cases as follows.
\begin{thm}\label{thm:gen_t}
Fix any non-trivial distribution $P_T$. The operator $\calQ_{P_T}$ has at most one unique non-trivial fixed point $\mu^*$, and it is symmetric. For all non-trivial $\mu$, the recursion $\calQ^{h}_{P_T}\mu$ converges weakly to the same fixed point, to $\mu^*$ if it exists, and to the trivial $\delta_0$ otherwise. 
\end{thm}
\begin{remark}
Note that the trivial fixed point $\mu=\delta_0$ exists if and only if  
all survey distributions are trivial w.p.1. 
This corresponds to the no survey case in the basic formulation. 
 In all other cases, the operator $\calQ_{P_T}$ has exactly one unique non-trivial fixed point, and all BP recursions  converge to $\mu^*$, including ones initialized with the trivial distribution. 
\end{remark}
To state the key intermediate results, we introduce the following concepts. 
\begin{defn}\label{def:polc}
For any fixed tree channel $T$ and any $r>0$, we say a growing point $v$ is \emph{dominant}, if the sign of the LLR message returned by the BP algorithm 
remains consistent whenever the input message at $v$ is $r$ and 
the inputs at all other growing points are within $[-r,r]$.  
For a given parameter $r$, we say $T$ satisfies the \emph{polygon condition} if there are no dominant growing points, and 
we say $T$ is \emph{representative} if it also satisfies   
$$F_T(r)\triangleq r_{\max}\left(\calQ_TB_{\frac{1}{1+e^r}}\right)>r.$$
\end{defn}
\begin{thm}\label{thm:gen_t_tec}
For any distribution $P_T$, any non-trivial symmetric fixed point $\mu$, and any $\phi\in(0,\frac{1}{2})$, let $r= r_{\max}(B_{\phi}\boxconv\mu)$, we have
\begin{align}\label{ineq:gen_t_1}
    B_{\phi}\boxconv\mathcal{Q}_{P_T}\mu\prec  \mathcal{Q}_{P_T}(B_{\phi}\boxconv\mu)
\end{align}  
if $T$ is representative with non-zero probability. Moreover, for any $P_T$ that is non-trivial, there is an integer $h>0$ that the  tree network created by distribution $P_T$  in $h$ steps is representative with non-zero probability, implying 
\begin{align}\label{ineq:gen_t_h}
     B_{\phi}\boxconv\mathcal{Q}_{P_T}^h\mu\prec  \mathcal{Q}_{P_T}^h(B_{\phi}\boxconv\mu).
\end{align} 
\end{thm}
\begin{remark}
Intuitively, for any growing point $v$ to be dominant, it is equivalent to state that the estimation error probability 
 in a robust reconstruction setting can be minimized by solely measuring $X_v$. In particular, the estimation error 
 is Bayes with $\Ber({\frac{1}{2}})$ prior, and the observing channels  are 
BSC with crossover probability 
$\frac{1}{1+e^r}$. Theorem \ref{thm:gen_t_tec} states that whenever an element tree is not equivalent to a trivial network, its self-concatenation will forbid such optimal estimators as the tree grows large.
\end{remark}
\begin{remark}
 The settings in earlier sections can be interpreted as special cases 
 where the heights of the element trees are no greater than $1$  and all $\delta_e$ parameters are constant. 
 The contraction results 
 stated in our general formulation
 provide a natural interpretation on the need of special treatments for $d\leq 2$. When no survey channels are present, the polygon condition requires $F_{\delta_e}(r)$ for all edges to be the side lengths of a polygon. It requires at least $d\geq 3$ equilateral edges to form a polygon. Otherwise, self-concatenation or non-trivial survey distributions are needed to meet the condition. 
\end{remark}





We also state some useful properties.  which can be proved by induction over the structure of element trees.
\begin{prop}\label{prop:gen_t_ele}
The following statements hold for any distribution $P_T$. 
\begin{enumerate}
    \item If $\mu\preceq\nu$, then $\calQ_{P_T}\mu\preceq \calQ_{P_T} \nu$.
    \item If $\{\mu_n\}_{n\in\mathbb{N}}$ converges weakly to $\mu$, then $\calQ_{P_T}\mu_n$ converges weakly to $\calQ_{P_T}\mu$.
    \item For any symmetric $\mu$ and $\phi\in[0,1]$, we have  $  B_{\phi}\boxconv\mathcal{Q}_{P_T}\mu\preceq  \mathcal{Q}_{P_T}(B_{\phi}\boxconv\mu)$.
\end{enumerate}
\end{prop}
\begin{proof}
The first property follows from the fact that BP preserves degradation. The second property is due
to the continuity of convolution and box convolution for weak convergence. Formally, 
one can repeat the same relevant steps in the proof of Theorem \ref{thm: uniq} as follows. For any fixed rooted tree, 
we write $\calQ_{T}\mu_n$ as the law of a function that depends on a list of independent random variables with distributions given by $\mu_n$ and $\mu_v$'s.  The value of this function converges in probability as $n\rightarrow \infty$, due to the weak convergence of $\mu_n$, and this convergence is uniform with respect to all parameter values of $\delta_e$'s. Then the weak convergence of $\calQ_{P_T}\mu_n$ follows from the fact that finite rooted trees are countable. 
The third property follows from the stringy tree lemma and the commutativity of box convolution.   
\end{proof}

In the rest of this section, we prove the uniqueness theorem assuming the correctness of above  intermediate results. 
Then we provide a proof of Theorem \ref{thm:gen_t_tec}  
in Appendix \ref{app:gen_tec}.  

\begin{proof}[Proof of Theorem \ref{thm:gen_t}]
First, as explained in the proof of Theorem \ref{thm: uniq}, 
a  contraction statement in the form of inequality \eqref{ineq:gen_t_h} implies the uniqueness of non-trivial symmetric BP fixed point.  Then, as noted in the proof for asymmetric distributions (see 
proof of Corollary \ref{thm: uniq_as}
), the non-existence of asymmetric fixed point and the needed BP convergence can be proved by 
focusing on a certain subclass of the initial distributions.
In particular, let $\mu_{(h)}\triangleq \calQ^h_{P_T}B_0$ denote the recursion with a noiseless channel initialization. 
We only need to show that 
 when $\mu_{(h)}$ converges to the non-trivial $\mu^*$, the following recursion converges to the same limit for any non-trivial $\mu$. 
\begin{align} \underline{\mu}_{(h)}&\triangleq\calQ^h_{P_T}\underline{\mu}_{(0)},\nonumber\\
    \underline{\mu}_{(0)}&\triangleq B_{\phi^*(\mu^*,\mu)}\boxconv\mu^*.\nonumber\end{align}
This generalization is guaranteed by the fact that 
$\calQ_{P_T}$ preserves degradation. 

Recall that 
$\mu^*$ is a symmetric fixed point. We can apply the third property of Proposition \ref{prop:gen_t_ele} to show that \begin{align}
\underline{\mu}_{(0)}=B_{\phi^*(\mu^*,\mu)}\boxconv(\calQ_{P_T}\mu^*)\preceq \calQ_{P_T}(B_{\phi^*(\mu^*,\mu)}\boxconv\mu^*)=\underline{\mu}_{(1)}.    \nonumber
\end{align}
More generally, we have the following chain by applying the first property of Proposition \ref{prop:gen_t_ele}. \begin{align}\label{eq:ptchain}
    \underline{\mu}_{(h)}=\calQ^h_{P_T} \underline{\mu}_{(0)}\preceq \calQ^h_{P_T} \underline{\mu}_{(1)}=\underline{\mu}_{(h+1)}.
\end{align}
Therefore, the monotone convergence property of degradation can be applied and the recursion
$\underline{\mu}_{(h)}$ converges to a symmetric distribution. By the second property of Proposition \ref{prop:gen_t_ele}, this limit is a BP fixed point (see relevant steps in the proof of Theorem \ref{thm: uniq}), and we denote it by
$\underline{\mu}^*$. 
 
  Because the degradation index $\phi^*(\mu^*,\mu)$ always belongs to $[0,\frac{1}{2})$ for non-trivial $\mu$. We have that  $\underline{\mu}_{(0)}$ is non-trivial. Given the degradation chain stated by inequality \eqref{eq:ptchain}, 
 all $\underline{\mu}_{(h)}$ are bounded by $\underline{\mu}_{(0)}$, so is their limit. Hence, $\underline{\mu}^*$ is non-trivial as well. 
 Then by the uniqueness of non-trivial symmetric fixed point we have $\underline{\mu}^*=\mu^*$. 
\end{proof}








\fi
		\section*{Acknowledgement}
This work was supported in part by the MIT-IBM Watson AI Lab and by the National Science Foundation under Grant No CCF-2131115.
\appendices

\section{Properties of degradation, $\beta$-curves and weak convergence}\label{app:deg_detail}

When discussing symmetric distributions, here we consider the most general formulation where  distributions can have non-zero mass at $+\infty$, similar to \cite{richardson2004fixed, 910578}, we adopt the natural definition of weak convergence 
over domain $(-\infty,+\infty]$.

\begin{prop}\label{prop:weak_conv} The set of symmetric probability distributions on $(-\infty, \infty]$ and the set of convex, monotone functions
$\beta:[0,1]\to [0,1]$ such that $\beta(t) \ge t$ are in 1-1 correspondence given by
Def.~\ref{def:beta_s}.
Furthermore, the following are equivalent:
	\begin{enumerate}
	\item The sequence of symmetric distributions $\mu_n$ weakly converge to $\mu$ 
	\item $\beta(t; \mu_n) \to \beta(t; \mu)$ for all $t\in[0,1]$.
	\item $\beta(t; \mu_n) \to \beta(t; \mu)$ uniformly over $t\in[0,1]$.
	\end{enumerate} 
\end{prop}
\begin{proof}  First, notice that every symmetric $R\sim \mu$ is completely determined by the
distribution of $T=\tanh {|R|\over 2}$. Indeed, for any bounded  function $f$ that is continuous on $(-\infty,+\infty]$ we have 
$$  \EE_{R\sim \mu}[f(R)] = \EE_{R\sim \mu} \left[\frac{e^Rf(R)+f(-R)}{e^R+1} \right]\,.$$
The function under the latter expectation is even and thus can be written as $h(T)$ for some
continuous function $h:[0,1]\to \mreals$. Hence, we can equivalently study correspondence between
measures $\nu$ of $T$ on $[0,1]$ and $\beta(t;\mu) = \EE_{T\sim \nu} [T \vee t]$. A simple integration
by parts shows 
\begin{equation}\label{eq:bt_curve_id}
	\beta(t;\mu) = 1 - \int_{t}^1 dy \PP[T \le y] = t+\int_t^1 dy \PP[T>y]\,.
\end{equation}
From this identity it is clear how to recover the CDF of $T$ from the $\beta(t;\mu)$ by
differentiation. 

For $1\Rightarrow 2$ we only need to use the fact that $\tanh{|R|\over 2} \vee t$ is a bounded continuous
function of $R$. For $2\Rightarrow 3$ we simply notice that $\beta$-curves are always $1$-Lipschitz,
and thus the family $\beta(t; \mu_n), n\ge 1$ is equicontinuous (hence pointwise convergence and
uniform convergence coincide). For $3\Rightarrow 1$ we first notice that 
$\mu_n \to \mu$ (weakly) is equivalent to convergence of corresponding distributions of $\tanh
{|R|\over 2}$ (as discussed above). Denote this (to be shown) convergence by $\nu_n
\to \nu$ (weakly). To that end, notice that from~\eqref{eq:bt_curve_id} we have $\beta(0)=\EE[T]$
and 
\begin{align*} \int_0^1 t^{s-1} (\beta(t)-t) dt = {1\over s(s+1)} \EE[T^{s+1}] && \forall \, s>0.\end{align*}
Therefore, uniform convergence of $\beta$-curves implies convergence of moments of $T$ and thus
(since $T$ is supported on $[0,1]$) the weak convergence as well.
\end{proof}

\begin{prop}\label{prop:lim_deg_com}\label{prop:sand} \label{prop:order}\label{prop:convcom}\label{prop:bequi}\label{prop:q_deg}
Degradation has the following properties. 
 \begin{enumerate}
 \item  (Continuity) For any two sequences of  distributions $\{{\mu}_n\}_{n\in\mathbb{N}}$ and $\{{\nu}_n\}_{n\in\mathbb{N}}$ that weakly converge to ${\mu}$ and ${\nu}$ respectively and satisfy ${\mu}_{n} \preceq{\nu}_{n}$ for any $n$, we have ${\mu}\preceq{\nu}$.
 \item (Sandwich Theorem)
For any sequences of distributions $\{\underline{\mu}_n\}_{n\in\mathbb{N}}$, $\{{\mu}_n\}_{n\in\mathbb{N}}$, $\{\overline{\mu}_n\}_{n\in\mathbb{N}}$ satisfying $\underline{\mu}_{n}\preceq \mu_{n} \preceq\overline{\mu}_{n}$ for any $n$ and $ \underline{\mu}_{n},\overline{\mu}_{n}$ weakly converges to $ \mu^*$ for some $\mu^*$, we have ${\mu}_{n}$ weakly converges to $ \mu^*$. 
\item (Monotone Convergence) For any sequence of symmetric  distributions $\{{\mu}_n\}_{n\in\mathbb{N}}$, if either ${\mu}_{n} \preceq{\mu}_{n+1}$ or ${\mu}_{n+1} \preceq{\mu}_{n}$ holds for all $n$,  then the sequence converges weakly to a symmetric distribution. 
 \item (Poset Structure and Antisymmetry)
Degradation defines a partial order on the set of all symmetric distributions. Especially, $\mu\preceq \nu$ and $\nu\preceq \mu$ implies $\mu=\nu$. 
      \item (Convolution-preserving) 
For any symmetric distributions $\mu$, $\nu$, and $\tau$ satisfying $\mu\preceq\nu$, we have $\tau*\mu\preceq\tau*\nu$ and $\tau\boxconv\mu\preceq\tau\boxconv \nu$.
\item (Invertibility over BSC Operator)
For any $\phi\in[0,\frac{1}{2})$ and any 
symmetric distributions $\mu$ and $\nu$, 
the statements $\mu\preceq\nu$ and $B_{\phi}\boxconv\mu\preceq B_{\phi}\boxconv\nu $ are equivalent.
\item  (BP-preserving)
For any $\delta\in [0,1]$ and  degree distribution $P_d$, 
let $\calQ_{\textup{s}}$ be the operator defined in Section \ref{sec:intro}. Then for any two symmetric distributions $\mu\preceq\nu$, we have  $$\calQ_{\textup{s}}\mu\preceq \calQ_{\textup{s}}\nu. $$
  \end{enumerate} 
\end{prop}


\begin{remark}
 The concept of degradation
has been studied as early as in \cite{dbc,sso,csn, blackwell1953equivalent}, under a topic called comparison of experiments. It has also appeared later in \cite{shannon1958note, 1054727, richardson2008modern, 5075875, el2011network,
mastersthesis, 8362951} for 
communication channels. 
Under those contexts, degradation serves as a preorder that divides experiments or communication channels into equivalent classes. 
The relationship given by Definition \ref{def:degr} permits antisymmetry 
due to the fact that LLR is a minimal sufficient statistic. 
\end{remark}

\begin{proof}[Proof of Proposition \ref{prop:lim_deg_com}]
The first property follows from the sequential compactness of joint distributions. Consider any choice of joint distributions between $\mu_n$ and $\nu_n$ that satisfy Definition \ref{def:degr}. We can find a subsequence of these distributions that converges weakly as $n\rightarrow\infty$. Their limit provides a valid construction for degradation. Specifically, let $P_{Y,Z}$ be any such limit and by marginal convergence 
we assume that $Y\sim \mu$ and $Z\sim \nu$. 
The invariance of $P_{Y|Z}$ 
under $(Y,Z)\rightarrow (-Y,-Z)$ is preserved under weak convergence as it 
is equivalent to  $dP_{Y,Z}=e^{Z}dP_{-Y,-Z}$.

The second property is clear from Theorem~\ref{thm:ftdeg} and Prop.~\ref{prop:weak_conv}. Specifically, we only use the part of Theorem~\ref{thm:ftdeg} that states $\beta$-curve inequalities are implied by degradation. This directly follows from the fact that degradation implies inequalities in Bayes estimation errors, which lead to the needed inequalities in $\beta$-curves  (Proposition \ref{prop:rem:beta_err}).

The third follows from the identification with $\beta$-curves (see~\cite[Lemma
4.75]{richardson2008modern}). The fourth property is again via the $\beta$-curve inequalities stated in 
Theorem~\ref{thm:ftdeg}.

The rest of the properties state that degradation is  preserved under several elementary operations  and their compositions. 
They can be proved by considering the binary-input-channel equivalence of the related
distributions, and the needed Markov chain constructions naturally follow from the probabilistic
interpretations of these operations. In particular, box convolution can be viewed as channel
concatenation (see Remark \ref{remark:box_phy}), and convolution can be viewed as parallel
concatenation.
\end{proof}

\subsection{Convergence in Degradation Metric}\label{app:deg_met}


The next result explains that topology of the degradation metric is strictly finer than that of weak convergence.
\begin{prop}\label{prop:deg_met} Let 
 $\{\mu_n\}_{n\in\mathbb{N}}$ and $\mu$ be non-trivial symmetric distributions, then the following are equivalent
 \begin{enumerate}
     \item $d(\mu_n,\mu)\to 0$
     \item $\mu_n \to \mu$ (weakly) and 
    $t_{\max}(\mu_n) \to t_{\max}(\mu)$
 \end{enumerate}
\end{prop}
\begin{proof}
We first prove that convergence in degradation metric implies weak convergence. 
Recall 
Proposition~\ref{prop:weak_conv}. It suffices to show the pointwise convergence of $\beta$-curves.
The proof follows from the fact that both $\phi^*(\mu_n,\mu)$ and $\phi^*(\mu,\mu_n)$ converges to $0$ as $d(\mu_n,\mu)\rightarrow 0$, and each degradation index 
implies a uniform bound on  $\beta(t;\mu)-\beta(t;\mu_n)$. 

Consider any pair of symmetric distributions $\mu$ and $\nu$,  we derive a bound on $\beta(t;\mu)-\beta(t;\nu)$ using $\phi^*(\mu,\nu)$. 
Recall the second property of Proposition \ref{prop:deg_closed} states that
$$B_{\phi^*(\mu,\nu)}\boxconv\mu\preceq \nu.$$
By applying Theorem \ref{thm:ftdeg} and equation (\ref{eq:beta_bconv}),  we have
   \begin{align*} 
   \beta(t;\nu)\geq (1-2\phi^*(\mu,\nu))\beta\left(\frac{t}{1-2\phi^*(\mu,\nu)};\mu\right).
    \end{align*} 
    Note that $\beta$-curves are non-negative, non-decreasing function of $|t|$ that are upper bounded by $1$ for any $|t|\leq 1$. If $\phi^*(\mu,\nu)<\frac{1}{2}$, we have the following inequality.
       \begin{align*} 
   \beta(t;\nu)\geq (1-2\phi^*(\mu,\nu))\beta\left({t};\mu\right)\geq \beta\left({t};\mu\right)-2\phi^*(\mu,\nu)&&\textup{for }|t|\leq 1.
    \end{align*} 
    Because $\beta(t;\nu)=\beta(t;\nu)=|t|$ for  $|t|> 1$, we have obtained an bound of $\sup_t( \beta(t;\mu)-\beta(t;\nu))$, which converges to $0$ as $\phi^*(\mu,\nu)\rightarrow0$. 
    
    By symmetry, we also have 
    \begin{align*} 
   \beta(t;\mu)\geq \beta\left({t};\nu\right)-2\phi^*(\nu,\mu)
    \end{align*}
    for $\phi^*(\nu,\mu)<\frac{1}{2}$. Therefore, the convergence of degradation index function implies that
    $$\lim_{n\rightarrow\infty}\sup_t|\beta(t;\mu_n)-\beta(t;\mu)|\leq
    \lim_{n\rightarrow\infty}2(\phi^*(\mu,\mu_n)+\phi^*(\mu_n,\mu))=0,$$
which proves the weak convergence.

Similarly, we provide an analysis of $t_{\max}$ functions based on  the BSS Theorem \ref{thm:ftdeg} and equation (\ref{eq:beta_bconv}). For any symmetric $\mu$ and $\nu$,  Proposition \ref{prop:deg_closed} implies that  
\begin{align*}
    t_{\max}(\nu)\geq (1-2\phi^*(\mu,\nu))t_{\max}(\mu)\geq  t_{\max}(\mu)-2\phi^*(\mu,\nu).
\end{align*}
Therefore, by symmetry, we have
\begin{align*}
   \lim_{n\rightarrow\infty}|t_{\max}(\mu_n)-t_{\max}(\mu)|\leq   
    \lim_{n\rightarrow\infty}2(\phi^*(\mu,\mu_n)+\phi^*(\mu_n,\mu))=0.
\end{align*}

Now we prove the opposite direction. 
First, we show that unconditionally, the weak convergence of $\mu_{n}$ to $\mu$ implies 
\begin{align}\label{limit:aa_single}
\lim_{n\rightarrow\infty}\phi^*(\mu, \mu_n)=0.
\end{align}
 Recall that weak convergence implies $L_\infty$-convergence in $\beta$-curves (Proposition~\ref{prop:weak_conv}). We can 
  find a sequence of non-negative numbers $\epsilon_1, \epsilon_2,...$ converging to $0$ 
  such that  
\begin{align}
  \beta(t;\mu_n)\geq \beta\left({t};\mu\right)-\epsilon_n
\end{align}
for all $t\in[0,1]$ and $n\in\mathbb{N}_+$. 
The above lower bound can be realized by symmetric distributions. Formally, let $\beta_n(t)=\max\{ \beta\left({t};\mu\right)-\epsilon_n, t\}$. It is clear from Proposition~\ref{prop:weak_conv} that each  $\beta_n$ is a valid $\beta$-curve and we can find a symmetric $\nu_n$ such that $\beta\left({t};\nu_n\right)=\beta_n(t)$. Specifically, we have 
\begin{align}\label{eq:aa_nun}
    \nu_n[(-r,r)]=\begin{cases}\mu[(-r,r)]&\textup{if } r\in(0, r_{\textup{max}}(\nu_n)], \\ 1 &\textup{if } r\in(r_{\textup{max}}(\nu_n), +\infty).
\end{cases}
\end{align}


For $\mu$ being non-trivial, we have $t_{\max}(\mu)>0$.  
Therefore, we can choose $$\phi_n=\frac{1}{2}\left(1-\frac{t_{\max}(\nu_n)}{t_{\max}(\mu)}\right) 
$$
so that $r_{\max}(\nu_n)=F_{\phi_n}(r_{\max}(\mu))$ and  
\begin{align*}
\nu_n[(-F_{\phi_n}(r),F_{\phi_n}(r))]\leq \mu[(-r,r)]=
    B_{\phi_n}\boxconv\mu\,[(-F_{\phi_n}(r),F_{\phi_n}(r))]
\end{align*}
for all $r\in(0,+\infty)$. 
By the integration law in equation \eqref{eq:bt_curve_id} and the BSS Theorem, the above inequality implies 
$$B_{\phi_n}\boxconv\mu\preceq \nu_n\preceq \mu_n, $$
which proves that $\phi^*(\mu, \mu_n)\leq \phi_n$.  
From the continuity of $\beta$-curves, the convergence of $\epsilon_n$ to $0$ implies that $\lim_{n\rightarrow\infty} t_{\max}(\nu_n)=t_{\max}(\mu)$. Hence,
\begin{align*}
\lim_{n\rightarrow\infty}\phi^*(\mu, \mu_n)\leq \lim_{n\rightarrow\infty}\phi_n=0.
\end{align*}

On the other hand, we prove the following convergence with the additional assumption of $t_{\max}(\mu_n) \to t_{\max}(\mu)$. 
\begin{align}
\lim_{n\rightarrow\infty}\phi^*(\mu_n, \mu)=0.
\end{align}
By the same arguments, we pick any sequence of non-negative numbers $\epsilon_1, \epsilon_2,...$ converging to $0$ 
  such that  
\begin{align}\label{ineq:aa_r}
  \beta(t;\mu)\geq \beta\left({t};\mu_n\right)-\epsilon_n
\end{align}
holds for all $t\in[0,1]$ and $n\in\mathbb{N}_+$. Let $\tilde{\mu}_n$, $\tilde{\nu}_n$ be the symmetric distributions with $\beta\left({t};\tilde{\mu}_n\right)=\max\{ \beta\left({t};\mu_n\right), \beta\left({t};\mu\right)\}$ and 
$\beta\left({t};\tilde{\nu}_n\right)=\max\{ \beta\left({t};\tilde{\mu}_n\right)-\epsilon_n,t\}$. 
From BSS theorem we have ${\mu}_n\preceq \tilde{\mu}_n$ and $\tilde{\nu}_n\preceq\mu$.

Recall that $\mu$ is  non-trivial, which implies that $t_{\max}(\tilde{\mu}_n)\geq t_{\max}(\mu)>0$. 
We can choose  
$$\tilde{\phi}_n=\frac{1}{2}\left(1-\frac{t_{\max}(\tilde{\nu}_n)}{t_{\max}(\tilde{\mu}_n)}\right), 
$$
so the earlier proof steps imply that 
$B_{\tilde{\phi}_n}\boxconv\tilde{\mu}_n\preceq \tilde{\nu}_n$. Hence, we have the following chain $$B_{\tilde{\phi}_n}\boxconv{\mu}_n\preceq B_{\tilde{\phi}_n}\boxconv\tilde{\mu}_n\preceq \tilde{\nu}_n\preceq\mu,$$ 
which leads to the upper bound $\phi^*(\mu, \mu_n)\leq \tilde{\phi}_n$. 
Note that inequality \eqref{ineq:aa_r} and the definition of $\tilde{\nu}_n$ imply that
\begin{align*}
   \beta(t;\tilde{\nu}_n) \in[\beta(t;\mu)-\epsilon_n ,  \beta(t;\mu)].
\end{align*}
We have the convergence of $t_{\max}(\tilde{\nu}_n)$ to $t_{\max}(\mu)$ by the continuity of $\beta$-curves and the convergence of $\epsilon_n$. Note that $t_{\max}(\tilde{\mu}_n)=\max\{t_{\max}({\mu}_n), t_{\max}(\mu)\}$. We also have the convergence of $t_{\max}(\tilde{\mu}_n)$ to $t_{\max}(\mu)$.   Therefore, 
\begin{align*}
\lim_{n\rightarrow\infty}\phi^*(\mu_n, \mu)\leq \lim_{n\rightarrow\infty}\tilde{\phi}_n=0.
\end{align*}

To conclude, the convergence of both $\phi^*(\mu, \mu_n)$ and $\phi^*(\mu_n, \mu)$ to $0$ proves the convergence in degradation metric. 
\end{proof}

\begin{remark}\label{remark:lack_comp}
The convergence in degradation metric is not implied solely by the weak convergence  (or $L_\infty$-convergence in $\beta$-curves).
For example, the sequence $\mu_n= 2^{-n}\cdot  B_{\frac{1}{6}}+ \left(1-2^{-n}\right)\cdot B_{\frac{1}{3}}$ converges to $\mu=B_{\frac{1}{3}}$ both weakly and in $\beta$-curves, 
but $d(\mu_n, \mu)= \ln 2$ for any $n$.  
However, as we have shown in the earlier proof, a converse can be proved for a single-sided function $\phi(\mu, \mu_{n})$, which can be viewed as a potential function that 
is stabilized by the BP recursion. We present a slightly more general version of this result in Proposition \ref{lm:uniq_key2}. 

Besides, note that even the sequence $\mu_n$ has a bounded radius in $d$, 
it has no convergent subsequence under the same metric, which shows that bounded closed sets are
not compact in the $d$-metric. 
Further, all pairwise distances in this sequence are bounded away from zero as
$d(\mu_m,\mu_n)=\ln\frac{2}{1+2^{-|m-n|}}$. 
    Thus, while $d$-convergence is almost equivalent to weak convergence as per Prop. \ref{prop:deg_met}, the induced topologies are quite different.
\end{remark}
\begin{prop}\label{lm:uniq_key2} 
If a sequence of LLR distributions $\{\mu_n\}_{n\in\mathbb{N}}$ converges weakly to $\mu$, then for any symmetric $\nu$
\begin{align*}
\lim_{n\rightarrow\infty}\phi^*(\nu, \mu_n)=\phi^*(\nu, \mu).
\end{align*}
\end{prop}

 

\begin{proof}
The following lower bound follows from the first property in  Proposition \ref{prop:lim_deg_com} and the definition of degradation index,  
$$\liminf_{n\rightarrow\infty}\phi^*(\nu, \mu_n)\geq\phi^*(\nu, \mu). $$ 
Hence, it remains to show that
\begin{align}\label{ineq:uniq_l2_sup}
\limsup_{n\rightarrow\infty}\phi^*(\nu, \mu_n)\leq \phi^*(\nu, \mu).
\end{align} 
Recall that in the proof of Proposition \ref{prop:deg_met} we have proved equation \eqref{limit:aa_single}, which covers the special case of $\nu=\mu$. 
We bound the degradation index for general $\nu$ by applying the triangle inequality, i.e., the fourth property of Proposition \ref{prop:deg_closed}. 
\begin{align}
\phi^*(\nu,\mu_n)\leq \phi^*(\mu,\mu_n)+\phi^*(\nu,\mu)-2\phi^*(\mu,\mu_n)\phi^*(\nu,\mu).
\end{align}
Then the needed result is proved 
using equation \eqref{limit:aa_single}, particularly, 
\begin{align*}
\limsup_{n\rightarrow\infty}\phi^*(\nu,\mu_n)\leq \limsup_{n\rightarrow\infty}\phi^*(\mu,\mu_n)+\phi^*(\nu,\mu)= \phi^*(\nu,\mu).
\end{align*} 
\end{proof}


\section{Supplementary Proofs for Theorem \ref{thm:main}}

\subsection{Proof of Proposition \ref{prop:trans_deg_stdeg}}\label{app:pp_trans_deg_stdeg}


 \begin{proof}
 We first prove the proposition for condition  $\nu\preceq\tau\prec \mu$. Recall the definition of strict degradation. We can find $\phi \in(0,\frac{1}{2}]$ such that $\nu\preceq \tau \preceq B_{\phi}\boxconv\mu$. By applying the transitivity of degradation, we have $\nu\preceq  B_{\phi}\boxconv\mu$, which proves  $\nu\prec\mu$.
 
 On the other hand, if $\nu\prec\tau\preceq \mu$, we have $\phi \in(0,\frac{1}{2}]$ such that $\nu\preceq B_{\phi}\boxconv\tau$. Now we need the fifth property in Proposition \ref{prop:order} to obtain $B_{\phi}\boxconv\tau\preceq B_{\phi}\boxconv\mu$. Then the rest follows from the transitivity of degradation through the same arguments.  
 \end{proof}


\subsection{Proof of Proposition \ref{lemma:index_def}} \label{app:pl_index_def}
\begin{proof}

 The first property is proved by a probability-of-error argument.  Recall the definition of degradation index. We only need to show the existence of a $\phi\in[0,\frac{1}{2})$ that satisfies $B_{\phi} \boxconv\mu\preceq \nu$. We prove this fact by choosing $\phi=\EE_{R\sim\nu}[\frac{1-\textup{sgn}(R)}{2}]$. Because $\nu$ is non-trivial and symmetric, we have $\phi\in [0,\frac{1}{2})$. 

 Following the commutativity of box convolution and the physical interpretation, we have  $B_{\phi} \boxconv\mu \preceq B_{\phi}$. On the other hand, $B_{\phi}$ can be viewed as the LLR distribution of the 1-bit maximum likelihood estimator for the estimation problem characterized by $\nu$. Hence, there is a natural joint distribution that implies $B_{\phi}\preceq\nu$. Then the statement $B_{\phi} \boxconv\mu\preceq \nu$ follows from the transitivity of degradation.

 The second property states that the set defined in equation (\ref{eq:degind_def}) has a minimum. It 
 follows from the continuity of degradation under weak convergence, see the first property in Proposition \ref{prop:lim_deg_com}. 
 The third property can be proved using the sixth property in Proposition \ref{prop:bequi}.  
 The fourth property is due to the transitivity of degradation and associativity of box convolution. 
\end{proof}

\subsection{Proof of Proposition \ref{prop:rem:beta_err}}\label{app:pp_rem:beta_err}
\begin{proof}
Recall that the optimal Bayes estimation error is achieved by the maximum a posteriori (MAP) estimator, which compares the LLR with a fixed threshold  $\ln\frac{1-t}{1+t}=-2\, \tanh^{-1} t$ for the given prior distribution. Therefore, the achieved error can be written as
$$\frac{1+t}{2}\cdot \mu\left[\left(-\infty,-2\, \tanh^{-1}t\right]\right]+\frac{1-t}{2} \cdot \int e^{-r}   \mathbbm{1}\{r \in (-2\, \tanh^{-1}t,+\infty) \} d\mu(r), $$
where the second term is derived from the definition of LLR. For brevity, we denote this quantity by $P_{\textup{e}}(t)$.  

Note that for any bounded Borel function $f: (-\infty, +\infty] \rightarrow \mathbb{R}$, its expectation over $\mu$ can be written as follows using the  symmetry condition. 
\begin{align}
    \mathbb{E}_{R\sim\mu}\left[f(R)\right] &=\mathbb{E}_{R\sim\mu}\left[\frac{e^{\frac{R}{2}}}{e^{{\frac{R}{2}}}+e^{-{\frac{R}{2}}}}f(R)+\frac{e^{-\frac{R}{2}}}{e^{{\frac{R}{2}}}+e^{-{\frac{R}{2}}}}f(R)\right]\nonumber\\
    &=\mathbb{E}_{R\sim\mu}\left[\frac{e^{\frac{R}{2}}}{e^{{\frac{R}{2}}}+e^{-{\frac{R}{2}}}}f(R)+\frac{e^{-\frac{R}{2}}}{e^{{\frac{R}{2}}}+e^{-{\frac{R}{2}}}}f(-R)\right].\nonumber 
\end{align}
We specialize the above formula for $P_{\textup{e}}(t)$, which can be viewed as the expectation of the following function
\begin{align*}
    p(R)\triangleq \frac{1+t}{2} \mathbbm{1}\left\{t\leq -\tanh\frac{R}{2}\right\}+
\frac{1-t}{2} e^{-R} \mathbbm{1}\left\{t> -\tanh\frac{R}{2}\right\}.
\end{align*}
The resulting expectation is given by
\begin{align*}
  \mathbb{E}_{R\sim\mu}\left[\min\left\{ \frac{e^{-|R|}}{e^{{\frac{R}{2}}}+e^{-{\frac{R}{2}}}}, \frac{1-t}{2}\right\}\right]
\end{align*}
which equals $\frac{1}{2}(1-\beta(t;\mu))$.
\end{proof}

\subsection{Proof of Proposition \ref{lemma:strict_beta}}\label{app:pl_st_beta}



    \begin{proof} 
    
  {Proof of 1) $\Rightarrow$ 3).} 
  From $\nu\prec\mu$ and the non-trivial condition, we can choose  $\phi\in(0,\frac{1}{2})$ to satisfy $\nu \preceq B_{\phi}\boxconv\mu$. 
  By either applying Theorem \ref{thm:ftdeg} or Jensen's inequality over any joint distribution consistent with Definition \ref{def:degr}, we have $\beta(t;\nu)\leq \beta(t;B_{\phi}\boxconv\mu)$ for any $t\in \mathbb{R}$. Using equation (\ref{eq:beta_bconv}), we have
  \begin{align}
      ({1-2\phi})\beta(t; \mu)\geq {\beta\left((1-2\phi){t};\nu\right)}.\nonumber
  \end{align} 
    Therefore, following the definition of $t_{\max}$ and non-trivial condition, we have $t_{\max}(\nu)\leq (1-2\phi)t_{\max}(\mu)<t_{\max}(\mu)$. 
        On the other hand, note that any $\beta$-curve is $1$-Lipschitz. For $t\in[0, t_{\max}(\mu))$, we have
  \begin{align}
  {\beta\left((1-2\phi){t};\nu\right)}&\geq {\beta\left({t};\nu\right)}-2\phi t.\nonumber\\
  &>{\beta\left({t};\nu\right)}-2\phi \beta(t; \mu). \nonumber
  \end{align}
  Combining the above inequalities, 
  we obtain the stated gap condition.

  Proof of 3) $\Rightarrow$ 2).  If $t_{\max}(\nu)<t_{\max}(\mu)$, then $[0, t_{\max}(\nu)]$ is a subset of $[0, t_{\max}(\mu))$. Thus, the same gap condition applies.
  
  Proof of 2) $\Rightarrow$ 1).
  Recall that both $\mu$ and $\nu$ are non-trivial.
  The $\beta$-curves for both distributions are positive and continuous over $\mathbb{R}$. Therefore, the following quantity is well-defined
    $$\phi\triangleq
    \min_{t\in\left[0,t_{\max}(\nu)\right]} \frac{1}{2}\left(1-\frac{\beta(t;\nu)}{\beta(t;{\mu})}\right).$$
 From the gap condition, we have $\phi \in\left(0,\frac{1}{2}\right)$. It remains to show $\nu\preceq B_{\phi}\boxconv\mu$.

 
  
    One can verify that the $\beta$-curve is non-decreasing for any measure. 
    Hence, for any $t\in\left[0,t_{\max}(\nu)\right]$, the above definition implies  the following inequality.
    \begin{align}
        \beta(t;\nu)\leq(1-2\phi)\beta\left(t;\mu\right)\leq(1-2\phi)\beta\left(\frac{t}{1-2\phi};\mu\right).\nonumber
    \end{align}
    Using equation (\ref{eq:beta_bconv}), the above inequality can be written as
    \begin{align}\label{ineq:pl5_31}
        \beta(t;\nu)\leq\beta\left({t};B_{\phi}\boxconv\mu\right).
    \end{align}
        Because the $\beta$-curve is lower bounded by the identity function, recall the definition of $t_{\max}(\nu)$, 
    inequality (\ref{ineq:pl5_31}) holds for $t>t_{\max}(\nu)$ as well.
    Apply this conclusion to Theorem \ref{thm:ftdeg}, we have proved that $\nu \preceq B_{\phi}\boxconv\mu$, which implies strict degradation. 
    \end{proof}

\subsection{Proof of Proposition \ref{prop:bounded_crit}}\label{app:pl_bounded_crit}
\begin{proof}
Proof of 1). The forward statement is proved by noting that when $t_{\max}(\nu)<t_{\max}(\mu)$, the set $\left(\tanh{\left(\frac{s}{2}\right)}, 
t_{\max}(\nu)\right]$ is contained in the set $\left(\tanh{\left(\frac{s}{2}\right)},
t_{\max}(\mu)\right)$. For the converse statement, we have $\beta(t;\nu)=t<\beta(t;\mu)$ for all $t\in(t_{\max}(\nu),t_{\max}(\mu))$ from the definition of $t_{\max}$. So it remains to prove that  $t_{\max}(\nu)<t_{\max}(\mu)$. If $s\geq r_{\max}(\nu)$, this inequality is implied by $r_{\max}(\nu)\leq s< r_{\max}(\mu)$.  Otherwise, we have  $\beta(t_{\max}(\nu),\mu)>\beta(t_{\max}(\nu),\nu)=t$, 
and the inequality is implied due to the continuity of $\beta$-curves.

Proof of 2). The stated condition provides a chain of degradation, so we have $\beta(t;\nu)\leq \beta(t;\tau)\leq \beta(t;\mu)$ for all $t$. In either cases, the individual steps in the condition implies that the inequality is strict for $t\in\left(\tanh{\left(\frac{s}{2}\right)}
, 
t_{\max}(\nu)\right]$. 

Proof of 3). The first statement is a direct consequence of Proposition \ref{lemma:strict_beta} and the fact that $\beta$-curves are even functions. The second statement is implied by the BSS Theorem.

Proof of 4). The first statement directly follows from the definition. The second statement is implied by Proposition \ref{lemma:strict_beta}.  
\end{proof}

\subsection{Proof of Proposition \ref{lemma:uniq_4b}}\label{app:pl_uniq4b}
\begin{proof}
For convenience, we define 
\begin{align*}
\tau_{(\textup{T})}&=B_\phi\boxconv (B_{\delta_1}*B_{\delta_2}),\\{\tau}_{(\textup{S})}&=(B_\phi\boxconv B_{\delta_1})*(B_\phi\boxconv B_{\delta_2}).
\end{align*}
To prove Proposition \ref{lemma:uniq_4b}, we first derive closed-form expressions for the $\beta$-curves. 
Note that $\tau_{(\textup{T})}$ is a symmetric discrete distribution that can be written as a linear combination of $B_{\delta_{\textup{(T)}}^-}$ and $B_{\delta_{\textup{(T)}}^+}$ for some $0\leq \delta_{\textup{(T)}}^+\leq   \delta_{\textup{(T)}}^-\leq \frac{1}{2}$. From Proposition \ref{prop:beta_box}, the $\beta$-curve of $\tau_{(\textup{T})}$ must be a piecewise linear function with at most two corner points. Concretely, let\footnote{When $r_{\max}(B_{\delta_1})=r_{\max}(B_{\delta_2})=+\infty$, $\tau_{(\textup{T})}$ is simply $B_{\delta_{\textup{(T)}}^+}$, and we can let $t_{\textup{(T)}}^{-}$ 
take any value in $[0,t_{\textup{(T)}}^{+}]$.} 
\begin{align*}
      &t_{\textup{(T)}}^{-}\triangleq t_{\max}\left(B_{\delta_{\textup{(T)}}^-}\right)=(1-2\phi)\tanh\frac{|r_{\max}(B_{\delta_1})-r_{\max}(B_{\delta_2})|}{2},\\
         &t_{\textup{(T)}}^{+}\triangleq t_{\max}\left(B_{\delta_{\textup{(T)}}^+}\right)=(1-2\phi)\tanh\frac{r_{\max}(B_{\delta_1})+r_{\max}(B_{\delta_2})}{2},\\
      & \beta_{\textup{(T)},0}\triangleq\beta\left(0;\tau_{(\textup{T})}\right)=(1-2\phi)\max\{(1-2\delta_1),(1-2\delta_2)\}.
\end{align*} 
The value-function pair $(t,\beta(t;\tau_{(\textup{S})}))$ for $t\in[0,1]$ is on the lower convex envelope of the following finite set of points.
\begin{align}\Big\{&\left(0,\beta_{\textup{(T)},0}\right), 
\left(t_{\textup{(T)}}^{-},\beta_{\textup{(T)},0}\right), 
\left(t_{\textup{(T)}}^{+},t_{\textup{(T)}}^{+}\right),(1,1)\Big\}.\nonumber\end{align}
Similarly, let 
\begin{align*}
      &t_{\textup{(S)}}^{-}\triangleq t_{\max}\left(B_{\delta_{\textup{(S)}}^-}\right)=\tanh\frac{|F_\phi\left(r_{\max}(B_{\delta_1})\right)-F_\phi\left(r_{\max}(B_{\delta_1})\right)|}{2},\\
         &t_{\textup{(S)}}^{+}\triangleq t_{\max}\left(B_{\delta_{\textup{(S)}}^+}\right)=\tanh\frac{F_\phi\left(r_{\max}(B_{\delta_1})\right)+F_\phi\left(r_{\max}(B_{\delta_1})\right)}{2},\\
      & \beta_{\textup{(S)},0}\triangleq\beta\left(0;\tau_{(\textup{S})}\right)=(1-2\phi)\max\{(1-2\delta_1),(1-2\delta_2)\}.
\end{align*} 
Then function $\beta(t;\tau_{(\textup{S})})$ for $t\in[0,1]$ is given by the lower convex envelope of the following finite set. 
\begin{align}\Big\{&\left(0,\beta_{\textup{(S)},0}\right), 
\left(t_{\textup{(S)}}^{-},\beta_{\textup{(S)},0}\right), 
\left(t_{\textup{(S)}}^{+},t_{\textup{(S)}}^{+}\right),(1,1)\Big\}.\nonumber\end{align}

By some elementary calculus\footnote{In particular, the concavity of $F_{\phi}$ on $\mathbb{R}_{\geq0}$.}, one can prove the following equation and inequalities, which shows that $\beta(t;\tau_{(\textup{T})})\leq \beta(t;\tau_{(\textup{S})})$ for all $t\in\mathbb{R}$. 
$$t_{\textup{(S)}}^{-}\leq t_{\textup{(T)}}^{-},$$
$$t_{\textup{(S)}}^{+}\geq t_{\textup{(T)}}^{+}.$$
$$\beta_{\textup{(S)},0}=\beta_{\textup{(T)},0}.$$
Moreover, we always have $t_{\textup{(S)}}^{+}> t_{\textup{(T)}}^{+}$ except when any of $\phi, \delta_1, \delta_2$ is $\frac{1}{2}$. In all cases, this leads to $\beta(t;\tau_{(\textup{T})})< \beta(t;\tau_{(\textup{S})})$ for any $t\in (t_{\textup{(S)}}^{-}, t_{\textup{(T)}}^{+}]$. Note that $t_{\textup{(S)}}^{-}=\tanh{\frac{s_{\min}}{2}}$ and $t_{\textup{(T)}}^{+}=t_{\max}(\tau_{\textup{(T)}})$. We have $\tau_{\textup{(T)}}\prec_{s_{\min}}\tau_{\textup{(S)}}$. 
\end{proof}

\subsection{Proof of Corollary \ref{coro:uniq_4}}\label{app:pc_uniq_4}
\begin{proof}
Let $R_1\sim \mu_1$, $R_2\sim\mu_2$ be independent variables. Then each  $\mu_j$ for $j\in\{1,2\}$ can be expressed as a mixture of distributions $B_{\delta_j}$ with $\delta_j= \frac{1}{1+e^{R_j}}$. From  the bilinearity of convolution, the distributions on both sides of inequality   \eqref{eq:gen_beta_rel} can be expressed as mixtures of the corresponding terms in inequality \eqref{eq:simple_beta_rel}, i.e., $$B_\phi\boxconv(\mu_1*\mu_2)=\mathbb{E}_{R_1,R_2}[B_\phi\boxconv (B_{\delta_1}*B_{\delta_2})],$$
 $$(B_\phi\boxconv\mu_1)*(B_\phi\boxconv\mu_2)=\mathbb{E}_{R_1,R_2}[(B_\phi\boxconv
 B_{\delta_1})*(B_\phi\boxconv B_{\delta_2})].$$  The   $\beta$-curves can be compared
 using the (obvious) property: $\beta$-curve of a mixture is given by a mixture of the
 $\beta$-curves. Thus, 
the degradation relation stated in  Proposition \ref{lemma:uniq_4b} implies that $B_\phi\boxconv(\mu_1*\mu_2)\preceq (B_\phi\boxconv\mu_1)*(B_\phi\boxconv\mu_2)$.

It remains to verify the strict condition for the inequality on  $\beta$-curves. 
For brevity, we focus on the non-trivial cases where $\phi\neq\frac{1}{2}$ and both $\mu_1$ and $\mu_2$ are non-trivial. By the strict concavity of $F_{\phi}$ on $\mathbb{R}_{\geq 0}$,  and the definition of $r_{\max}$,
we have 
\begin{align*}
    r_{\max}(B_\phi\boxconv(\mu_1*\mu_2))&=F_{\phi}(r_{\max}(\mu_1)+r_{\max}(\mu_2))\\&<F_{\phi}(r_{\max}(\mu_1))+F_{\phi}(r_{\max}(\mu_2))\\&=r_{\max}((B_\phi\boxconv\mu_1)*(B_\phi\boxconv\mu_2)),
\end{align*}
which implies that $  t_{\max}(B_\phi\boxconv(\mu_1*\mu_2))<t_{\max}((B_\phi\boxconv\mu_1)*(B_\phi\boxconv\mu_2))$.
Therefore, it suffices to compare the $\beta$-curves at any fixed $t\in \left(\tanh\left(\frac{s}{2}\right), t_{\max}( (B_\phi\boxconv\mu_1)*(B_\phi\boxconv\mu_2))\right)$. By the definition of $r_{\max}$ function, there is a non-zero probability that $R_1$, $R_2$ are sufficiently close to $r_{\max}(\mu_1)$ and $r_{\max}(\mu_2)$, respectively, such that $$t\in \left(\tanh\left(\frac{|{{F_{\phi}(r_{\max}(B_{\delta_1}))-F_{\phi}(r_{\max}(B_{\delta_2}))}}|}{2}\right), t_{\max}( (B_\phi\boxconv B_{\delta_1})*(B_\phi\boxconv B_{\delta_2}))\right).$$ According to Proposition \ref{lemma:uniq_4b} and the first statement in Proposition
\ref{prop:bounded_crit}, this condition  implies strict inequality in $\beta$-curves at point $t$ between $B_\phi\boxconv (B_{\delta_1}*B_{\delta_2})$ and $(B_\phi\boxconv B_{\delta_1})*(B_\phi\boxconv B_{\delta_2})$. Hence, their integration also contributes to a strict inequality. 
\end{proof}

\subsection{Proof of Proposition \ref{lemma:basic_conv}}\label{app:pl_basic_conv}
 \begin{proof}
    Recall that $\beta$-curves are characterized by the error probabilities of MAP estimators (see Proposition \ref{prop:rem:beta_err}). We construct an estimation problem with prior $X\sim \Ber(\frac{1-t}{2})$, and independent observations $Y,Z$ such that $Y$ is measured through a symmetric channel characterized by $\mu$ and $Z$ is measured through a BSC with crossover probability $\phi$. 
    Let $\hat{X}$ be the MAP estimator. Because the LLR distribution of this experiment is characterized by $B_{\phi}*\mu$, we have
    $$\mathbb{P}[\hat{X}= X]= \frac{1+\beta(t; B_{\phi}*\mu)}{2}.$$
    
    Conditioned on $Z=0$ or $Z=1$, the inference problem reduces to  estimating $X$ given $Y$ with different priors. Note that the MAP estimator remains the same. We have
    $$\mathbb{P}[\hat{X}= X|Z=0]= \frac{1+\beta(t_0; B_{\phi}*\mu)}{2},$$
    $$\mathbb{P}[\hat{X}= X|Z=1]= \frac{1+\beta(t_1; B_{\phi}*\mu)}{2}.$$
    Combining above results, we obtain the following equation, which is identical to the needed statement.
    $$\beta(t; B_{\phi}*\mu)=\PP[Z=0] \beta(t_0; B_{\phi}*\mu)+\PP[Z=1]\beta(t_1; B_{\phi}*\mu).$$
    \end{proof}

\subsection{Proof of Proposition \ref{lemma:convdeg}}\label{app:pl_convdeg}

    \begin{proof}
    
    Consider any fixed $r\in(s-\ell,r_{\max}(\mu)-\ell)$. Let $t=\tanh\frac{|r|}{2}$, we have $t\in[0,1)$. Then we can apply Proposition \ref{lemma:basic_conv} to evaluate $\beta(t; \tau*\mu)-\beta(t; \tau*\nu)$ by writing $\tau$ as an integration of $B_{\phi}$ distributions.    
        In particular, let $Z\sim\tau$ and $\phi=\frac{1}{e^{|Z|}+1}$. In general, we have
    \begin{align}\label{eq:pl4_int}
\beta(t; \tau*\mu)&=\mathbb{E}\left[\left(\frac{1+t-2t\phi}{2}\right)\beta\left(t_0;\mu\right)+\left(\frac{1-t+2t\phi}{2}\right)\beta\left(t_1;\mu\right)\right]
\end{align}
for any symmetric $\mu$, where $$    t_0=\tanh{\left(\frac{|Z|+|r|}{2}\right)},$$

$$
    t_1=\tanh\left({\left|\frac{|Z|-|r|}{2}\right|}\right).
$$

When $\nu\prec_s \mu$, we always have $\beta\left(t_0;\nu\right)\leq \beta\left(t_0;\mu\right)$ and $\beta\left(t_1;\nu\right)\leq \beta\left(t_1;\mu\right)$. To prove $\beta(t; \tau*\nu)<\beta(t; \tau*\mu)$ using equation (\ref{eq:pl4_int}), it remains to show that $\beta\left(t_0;\nu\right)< \beta\left(t_0;\mu\right)$ or $\beta\left(t_1;\nu\right)< \beta\left(t_1;\mu\right)$ with non-zero probability. 
Let $t_Z=\tanh\left(\left|\frac{r+Z}{2}\right|\right)$.
Because $r+\ell\in(s,r_{\max}(\mu))$ and $\ell\in\supp(\tau)$, we have $\beta\left(t_Z;\mu\right)> \beta\left(t_Z;\nu\right)$ for $Z$ in a neighbourhood of $\ell$, which holds with non-zero probability. 
Note that $t_Z\in\{t_0,t_1\}$. The needed inequality is proved.

For the second statement, we can assume $\mu$ is non-trivial,  otherwise the stated interval is an empty set and nothing needs to be proved.
Note that in this case, $\nu\prec \mu$ 
is equivalent to $\nu\prec_s \mu$ for any $s\in(-r_{\max}(\mu),0)$. Therefore, following exactly the same steps, we have the needed strict inequality holds for $r\in(s-\ell,r_{\max}(\mu)-\ell)$ for any such $s$. Hence, inequality (\ref{ineq:roc}) holds if $r$ belongs to $r\in(-r_{\max}(\mu)-\ell,r_{\max}(\mu)-\ell)$, which contains $[-\ell,r_{\max}(\mu)-\ell)$ as a subset.
    \end{proof}

\subsection{Proof of Proposition \ref{lemma:sur_imp}}\label{app:pl_surimp}

\begin{proof}
We apply induction over $k\in\mathbb{N}$. 
 The base case $k=1$ follows directly from Proposition~\ref{prop:sur_basic}. Indeed, when $s_1<0$, it implies 
$$ B_\phi \boxconv \calQs \mu = B_\phi \boxconv ((B_\delta \boxconv  \mu) * \mus) \prec
(B_\phi
\boxconv B_\delta \boxconv  \mu) * \mus = \calQs (B_\phi \boxconv \mu)$$
as $s_{\min} = F_\phi(r_{\max}(B_\delta \boxconv \mu)) - r_{\textup{s}} = F_\phi(F_\delta(r_{\max}(\mu))) - r_{\textup{s}}
= s_1<0$. Here we used the fact that $r_{\max}(B_\tau  \boxconv \nu) =
F_\tau(r_{\max}(\nu))$ and 
\begin{equation}
	F_\phi \circ F_\delta = F_\delta \circ F_\phi\,,\nonumber
\end{equation}
due to commutativity of $B_\phi \boxconv B_\delta$. Similarly, when $s_1\geq 0$, the statement of Proposition~\ref{prop:sur_basic} implies $B_\phi \boxconv \calQs \mu \prec_{s_{1}}
 \calQs (B_\phi \boxconv \mu)$.

Assume  Proposition \ref{lemma:sur_imp} holds for some $k\in\mathbb{N}$. We prove that it holds for $k+1$. 
As everywhere before, our method is to start from a non-strict degradation chain given by
\begin{align}\label{eq:si_1}
	B_\phi \boxconv (\calQs^{k+1}\mu) \preceq  \calQs (B_\phi \boxconv \calQs^k \mu) 
				\preceq \calQs (\calQs^k(B_\phi \boxconv \mu)) 
\end{align}				
and keep track of areas where the comparison of $\beta$-curves is strict.  
Denote for convenience $\mu_k = \calQs^k \mu$, $\nu_k = \calQs^k (B_\phi \boxconv \mu)$ and $r_k =
r_{\max}(\mu_k)$. We will repeatedly use the fact that $r_{\max}(B_\tau \boxconv \nu) =
F_\tau(r_{\max}(\nu))$, so that for example $r_{k+1} = F_\delta(r_k) + r_{\textup{s}}$.

Observe the first step of inequality~\eqref{eq:si_1}, which
essentially states that
$$ B_\phi \boxconv ((B_\delta \boxconv \mu_k) * \mus) \preceq (B_\delta \boxconv B_\phi \boxconv
\mu_k) * \mus\,. $$
This is simply an instance of applying Proposition~\ref{prop:sur_basic}. Thus, we have that the 
comparison of $\beta$-curves is strict for 
\begin{equation}\label{eq:si_4}
	F_\phi(F_{\delta}(r_k)) - r_{\textup{s}} < s < F_\phi(F_\delta(r_k)) + r_{\textup{s}}\,.
\end{equation}
Therefore, we can assume $r_{\textup{s}}$ 
is finite for the rest of the proof because otherwise we have already established strict comparison for the entire $s\in\mathbb{R}$, which  includes $[0,r_{\max}(B_\phi \boxconv \mu_{k+1})]$ as a subset. 

Now we analyze the second step in~\eqref{eq:si_1}. From the induction hypothesis we 
know that
\begin{equation}\label{eq:si_0}
	B_\phi\boxconv\mu_{k}\prec_{s_{k}} \nu_{k}\,
\end{equation}
for $s_k\geq 0$, and $B_\phi\boxconv\mu_{k}\prec \nu_{k}$ otherwise. 
Applying box convolution with $B_\delta$ to both sides of these inequalities and by equation (\ref{eq:beta_bconv}), we obtain 
$$ B_\delta \boxconv B_\phi\boxconv\mu_{k}\prec_{s_{k}'} B_\delta \boxconv \nu_k\,,$$
for $s'_k \triangleq F_\delta(s_k)$. 
Next, by convolving with $\mus$ on both sides, and then 
Proposition~\ref{lemma:convdeg} with $\ell = r_{\textup{s}}$, we get that 
$$ (B_\delta \boxconv B_\phi\boxconv\mu_{k}) * \mus \preceq (B_\delta \boxconv \nu_k) * \mus =
\nu_{k+1} $$
with inequality for $\beta$-curves strict for all $t=\tanh{|s|\over 2}$ with
\begin{equation}\label{eq:si_3}
	s_{k+1} = s_{k}' - r_{\textup{s}} < s < F_\delta(r_{\max}(\nu_k)) - r_{\textup{s}}\,.
\end{equation}

From induction hypothesis~\eqref{eq:si_0} we have $r_{\max}(\nu_k) > r_{\max}(B_\phi
\boxconv \mu_k) = F_\phi(r_k)$. Therefore, because of the strictness of the inequality we have
that~\eqref{eq:si_3} and~\eqref{eq:si_4} together imply comparison of first and last $\beta$-curves
in~\eqref{eq:si_1} for all
$t=\tanh(|s|/2)$ with
\begin{equation}\label{eq:si_5}
	s_{k+1} < s < F_\phi(F_\delta(r_k)) + r_{\textup{s}}\,.
\end{equation}
Finally, notice that $F_\phi(x+y) < F_\phi(x) + y$ for $x\ge 0$, $y> 0$, and 
$\phi\in(0,1)$.
Thus the right-hand side of~\eqref{eq:si_5} is strictly bigger than $r_{\max}(B_\phi \boxconv
\mu_{k+1}) = F_\phi(r_{k+1}) = F_\phi(F_\delta(r_k) + r_{\textup{s}})$. In all, we have established
strict comparison of $\beta$-curves for 
$$ s_{k+1} < s \le r_{\max}(B_\phi \boxconv \mu_{k+1})\,.$$
\end{proof}

\section{Proof of Corollary \ref{thm: uniq_as}}\label{app:asym}


\begin{prop}\label{prop:llrd_basic}
The following statements are true.
\begin{enumerate}
\item All symmetric distributions are LLR distributions. 
\item Any LLR distribution has a unique complement.
    \item Any distribution $\mu$ on $(-\infty,+\infty]$ is an LLR distribution if and only if it has a complement. 
\end{enumerate}  
\end{prop}
\begin{proof}

For the first statement, consider any symmetric $\mu$. The LHS of inequality \eqref{eq:intasym} equals $\mu[(-\infty,+\infty)]$, which is no greater than $1$.
For the second statement, 
the CDF of any complement distribution is completely determined by equation \eqref{eq:negative}, which proves the  uniqueness. 
Then, inequality \eqref{eq:intasym} ensures that the CDF specified by equation \eqref{eq:negative} is always bounded within $[0,1]$, which proves the existence.
For third statement, the LHS of inequality \eqref{eq:intasym} equals $\mu^{-}[(-\infty,+\infty)]$, which is no greater than $1$.

\end{proof}
The notion of degradation can be generalized to LLR distributions. One can verify that the following definition is consistent with Definition \ref{def:degr} on symmetric distributions. 
\begin{defn}
\label{def:degr_asym}
For any two LLR distributions $\mu_{Y}$, $\mu_{Z}$, we say $\mu_Y$ is a \emph{degraded version} of $\mu_{Z}$, denoted by $\mu_Y\preceq
\mu_Z$, if one can define joint distributions 
$\mu_{Y,Z}$ and $\mu_{Y^-,Z^-}$, 
with $\mu_{Y}$, $\mu_{Z}$, 
and their complements being the marginal distributions, 
such that $\mu_{Y|Z}$ and $\mu_{Y^-|Z^-}$ are identical. 
\end{defn}

Moreover, 
we have the following generalization.
\begin{prop}\label{prop:order_LLR}
Degradation defines a partial order on the set of all LLR distributions, which satisfies continuity, sandwich theorem, and has monotone convergence.
\end{prop} 
\begin{proof}
Our proof relies on the following generalization of $\beta$-curves. 
\begin{defn}
\label{def:beta_as}
For any LLR distribution $\mu$, we define its $\beta$\emph{-curve} as a function on domain $t \in\mathbb{R}$ given by the following equation.
\begin{align}
    \beta(t;\mu)\triangleq \mathbb{E}_{R\sim\frac{1}{2}(\mu+\mu^-)}\left[\left|\tanh\frac{R}{2}-t\right|\right],
\end{align}
where $\mu^-$ is the complement distribution.
\end{defn}
Given this definition, one can prove that LLR distributions and  their $\beta$-curves are in one-to-one correspondence. This is because every symmetric $\mu$ is completely determined by the distribution of $T=\tanh\frac{R}{2}$ for $R\sim \frac{1}{2}(\mu+\mu^-)$, and their CDF  can be obtained by differentiating the $\beta$-curve based on the following equation. 
\begin{align*}
	\beta(t;\mu) =
	1-t + \int_{-1}^t dy \PP[T \le y]- \int_{t}^1 dy \PP[T \le y]\, . 
\end{align*}
The BSS theorem can be generalized for asymmetric distributions for the above $\beta$-curve definition. In particular, we only require 
$\beta(t;\nu)\leq \beta(t;\mu)
$ 
for any $t\in\mathbb{R}$  and $\nu\preceq\mu$. 
This can be proved by the fact that Proposition \ref{prop:rem:beta_err} holds for asymmetric distributions, then the comparison of $\beta$-curves is implied by the comparison of corresponding hypothesis testing problems. 

Combining the above results, for any $\mu\preceq\nu\preceq\mu$ the $\beta$-curves of $\mu$ and $\nu$ are identical. Therefore, we have $\mu=\nu$, which proves antisymmetry.
The rest of the properties are obtained by the corresponding steps in Proof of Proposition \ref{prop:order},  
except that for continuity, the  $(Y,Z)\rightarrow(-Y,-Z)$ transformation is replaced by the comparison to the joint distribution generated by complements.
\end{proof}


\begin{proof}[Proof of Corollary \ref{thm: uniq_as}] 

The corollary is proved by first showing the unique convergence, i.e., the BP recursion converges to the same symmetric fixed point for all non-trivial initializations. 
Once it is proved, the non-existence of asymmetric fixed point follows naturally from the fact that any fixed point is the limiting distribution of the BP recursion with itself as the initialization.


 Consider any fixed non-trivial initialization $\mu$. We denote its recursion by $\tilde{\mu}_{(h)}\triangleq \calQ_{\textup{s}}^h\mu$. As a reference, we also consider the BP recursion in the setting with perfect leaf observation, 
 which is given by $\mu_{(h)}= \calQ^h_{\textup{s}}B_0$.  Note that $\tilde{\mu}_{(h)}$ can be interpreted as the LLR distributions with noisy leaf observations. 
 By the natural coupling, we have $B_{\frac{1}{2}}\preceq \tilde{\mu}_{(h)}\preceq\mu_{(h)}$. 
 Recall that the sandwich theorem of degradation holds for general LLR distributions. This implies the needed convergence if $\mu_{(h)}$ converges to the  trivial distribution. 



For the other case,  we have that $\mu_{(h)}$ converges to a non-trivial symmetric fixed point $\mu^*$ by the statement of Theorem \ref{thm:main}.
When the limit of $\mu_{(h)}$ is non-trivial, we need to replace  $B_{\frac{1}{2}}$ 
with degradation bounds that also converge to $\mu^*$. 
To that end, we extend the notion of degradation index. 
Because $\mu^*$ is symmetric, we can define $\phi^*(\mu^*,\tilde{\mu}_{(0)})$ exactly the same way as in equation (\ref{eq:degind_def}), and we let $\underline{\mu}_{(h)}$ be the recursion initialized by $$ \underline{\mu}_{(0)}\triangleq B_{\phi^*(\mu^*,\tilde{\mu}_{(0)})}\boxconv\mu^*.$$   
The first and the second property 
in Proposition \ref{prop:deg_closed} hold 
under this extension, implying that $\underline{\mu}_{(0)}$ is non-trivial and $\underline{\mu}_{(0)}\preceq\mu$. 
The proof for the first property can be generalized 
by choosing  
$\phi$ in the proof of in Proposition \ref{prop:deg_closed} as follows. 
$$\phi=\inf_{P_{Z|Y}} \max\left\{\mathbb{P}_{Y\sim \nu}[Z=0],\mathbb{P}_{Y\sim \nu^-}[Z=1] \right\},$$
where $\nu^-$ is the complement distribution of $\nu$.
From elementary statistics, it is known that $B_\phi$ is the LLR distribution of a thresholding quantized version of the estimation problem characterized by $\nu$, where the quantizer is chosen to achieve the minimax error probability. Therefore, we have $B_{\phi}\preceq \nu$, and the rest of the proof follows the same steps. 

Because degradation is still preserved under BP recursion for general distributions. We have $$\underline{\mu}_{(h)}=\calQ^h\underline{\mu}_{(0)}\preceq\calQ^h{\mu}=\tilde{\mu}_{(h)}. $$
Note that $\underline{\mu}_{(0)}$ is  symmetric. The uniqueness theorem states that the 
the constructed lower bounds converge to $\mu^*$. 
Then the unique convergence of 
$\calQ^h\mu$ 
follows from the sandwich theorem. 
\end{proof}

\section{Proof of Theorem \ref{thm:larged}} \label{app:pl_ld}
Note that the definition of degradation index remain unchanged. Theorem \ref{thm:larged} can be proved using the same steps upon the following proposition.\footnote{Except for the existence of non-trivial fixed points, which can be proved by analyzing the evolution of $V_{\mu_{h}}$ with noiseless initialization.}
\begin{prop} \label{prop:larged} For any symmetric $\mu$ and $\nu$,
\begin{enumerate}
    \item if $\nu\preceq\mu$, then $\calQ_{\textup{L}} \nu\preceq\calQ_{\textup{L}} \mu$; 
     \item we have $B_\phi\boxconv \calQ_{\textup{L}}\mu\prec \calQ_{\textup{L}}(B_\phi \boxconv\mu)$ for any $\phi\in(0,\frac{1}{2})$;
      \item  if $\nu$ 
is nontrivial, then 
$\phi^*(\calQ_{\textup{L}}\mu,\calQ_{\textup{L}}\nu)<\phi^*(\mu,\nu)$ or $\mu\preceq\nu$.
\end{enumerate}
\end{prop}

\begin{proof}
The first statement in Proposition \ref{prop:larged} is proved by viewing $\calQ_{\textup{L}}$ as a limit of $\calQ$ or $\calQ_\textup{s}$. The third statement can be proved using the second statement and Proposition \ref{prop:deg_closed}. 
Therefore, we focus on the proof for the second statement, 
which is obtained by deriving and comparing related $\beta$-curve functions.
For brevity, we shall ignore cases where $\mu$ is trivial or $P_{\overline{d}}$ is a delta distribution at $0$, where the statement is obviously true. Hence, the LHS of the needed inequality is non-trivial, and it suffices to examine the $\beta$-curves on both sides due to Proposition \ref{lemma:strict_beta}.

We first focus on the case where $\mu_{\textup{s}}$ is trivial, which implies that any $\calQ_{\textup{L}}\mu$ is a mixture of Gaussians, and their $\beta$-curves can be written using integrals of elementary functions. We have the following statement, which is derived from the fact that $\beta$-curves are linear functions of the errors in Bayes estimation. The error function for each fixed prior distribution is the minimum of estimation error over all threshold decoders, and the optimal threshold is the LLR that is consistent with the prior.  

\begin{prop}
For any $t\in[0,1]$ and $s\in[0,+\infty]$, we have
$$\beta(t;\mathcal{N}(s))=\max_{r\in[0,+\infty]}\frac{1}{2}\left((1-t)\textup{erf}\left(\frac{\sqrt{s}}{2\sqrt{2}}-r\right)+(1+t)\textup{erf}\left(\frac{\sqrt{s}}{2\sqrt{2}}+r\right)\right),$$
where $\textup{erf}(z)\triangleq {\frac{2}{\sqrt{\pi}}}\int_{0}^{z}e^{-x^2}dx$, and a maximizer is given by $r=r^*(s,t)\triangleq\sqrt{\frac{2}{s}}\, \textup{arctanh} (t)$. Moreover, for any fixed $t\in[0,1]$,  
$$\frac{d}{d \sqrt{s}}\beta(t;\mathcal{N}(s))=\sqrt{\frac{1-t^2}{2\pi}}e^{-r^*(s,t)^2-\frac{s}{8}}.$$
\end{prop}
Due to convexity of $\textup{arctanh}$, we have $r^*(\theta^2s,\theta t)\leq r^*(s,t)$ for any $\theta\in(0,1)$ and $t\in[0,1]$. Therefore, the following proposition is implied by the Lagrange mean value theorem.
\begin{prop}
The following inequality holds for any $\theta\in(0,1)$, $t\in[0,\theta]$, and $s\in(0,+\infty]$.
$$\beta( t;\mathcal{N}(\theta^2s))> \theta \beta(t/\theta;\mathcal{N}(s)).$$ 
\end{prop}

From linearity of  $\beta$-curves, 
 for $\mu_{\textup{s}}$ being the trivial distribution, we have
 $$\beta(t;\calQ_{\textup{L}}\mu)=\mathbb{E}_{P_{\overline{d}}}\left[\beta(t;\mathcal{N}\left(\overline{d}\cdot V_\mu\right)\right].$$
For brevity, let $\theta\triangleq(1-2\phi)^2$. Note that $V_{B_\phi\boxconv \mu}=\theta^2 V_{\mu}$ and recall Proposition \ref{prop:beta_box}. For any $t\in[0,\theta]$ we have
\begin{align*}
    \beta(t;B_{\phi}\boxconv\calQ_{\textup{L}}\mu)&=\theta\mathbb{E}_{P_{\overline{d}}}\left[\beta(t/\theta;\mathcal{N}\left(\overline{d}\cdot V_\mu\right)\right],\\
    \beta(t;\calQ_{\textup{L}}(B_\phi \boxconv\mu))&=\mathbb{E}_{P_{\overline{d}}}\left[\beta(t;\mathcal{N}\left(\overline{d}\cdot \theta^2V_\mu\right)\right].
\end{align*}
Recall that we only need to provide a prove for $\PP{[\overline{d}>0]>0}$ and non-trivial $\mu$. We have $V_\mu>0$ and $t_{\max}({B_{\phi}\boxconv\calQ_{\textup{L}}\mu})=\theta$. By an integration argument, the above results implies that 
$$\beta(t;\calQ_{\textup{L}}(B_\phi \boxconv\mu))>\beta(t;B_{\phi}\boxconv\calQ_{\textup{L}}\mu)$$ for any $t\in [0,t_{\max}({B_{\phi}\boxconv\calQ_{\textup{L}}\mu})]$, which proves the needed statement. 

~

Now we take non-trivial $\mu_{\textup{s}}$ into account. Let $\mathcal{Q}_{\textup{L},\delta_0}\mu \triangleq \mathbb{E}_{\overline{d} \sim P_{\overline{d}}}\left[\mathcal{N}\left(\overline{d}\cdot V_\mu\right)\right]$, we have that $\mathcal{Q}_{\textup{L}}\mu=\mathcal{Q}_{\textup{L},\delta_0}\mu*\mu_{\textup{s}}$. The earlier proof steps have essentially shown that $$B_\phi\boxconv \calQ_{\textup{L},\delta_0}\mu\prec \calQ_{\textup{L},\delta_0}(B_\phi \boxconv\mu)$$
for any $\phi\in(0,\frac{1}{2})$. Compared to the case with fixed $P_d$, the above inequality exactly leads to a form that corresponds to inequality \eqref{eq:si_7}. Therefore, the proof can be completed using the same steps from inequality \eqref{eq:si_6} to inequality \eqref{ineq:fpd_final}.
\end{proof}



\section{Proof of 
Theorem \ref{thm:gen_t_tec}
}\label{app:gen_tec}

In this appendix, we prove the intermediate results for the uniqueness theorem stated in Section \ref{sec:gllr}. For brevity, we shall assume $\delta_e<\frac{1}{2}$ throughout this section, as the more general cases produce the same set of BP operators.
We start by investigating the case of deterministic $T$, where a generalized version of inequality \eqref{ineq:gen_t_1} can be proved based on the following concepts.
\begin{defn} 
For any element tree $T$ and parameter $r>0$, we define the \emph{polygon number} of any growing point $v$ to be essential infimum of the LLR message returned by the BP algorithm, constrained on the input message at $v$ being $r$ and at all other growing points being $-r$.   We define the polygon number of the pair $(T,r)$ 
to be maximum polygon number over all growing points, and denote it by $p(T,r)$. 
\end{defn}

\begin{prop}\label{pp:gen59}
For any $r>0$, an element tree $T$ satisfies the polygon condition if and only if $p(T,r)<0$. 
\end{prop}
\begin{proof}
Recall Definition \ref{def:polc}. 
By the monotonicity of BP, any growing point is domiant if and only if their polygon number non-negative. Thus, the polygon coditon holds if and only if all of them are negative, which is equivalent to having  $p(T,r)<0$.   
\end{proof}

\begin{prop}\label{prop:gen_ele_ft}
 Let $\mathcal{Q}_{T}$ denote the BP operator for any deterministic element tree $T$. Then for any non-trivial symmetric $\mu$ and any $\phi\in(0,\frac{1}{2})$, we have
\begin{align}\label{ineq:genl_sineq}
    B_{\phi}\boxconv\mathcal{Q}_{T}\mu\prec_{p}  \mathcal{Q}_{T}(B_{\phi}\boxconv\mu),
\end{align}  
where $p=p(T,r_{\max}(B_{\phi}\boxconv\mu))$.
\end{prop}

\begin{proof}
We prove the proposition by induction over the size of the tree. For the base cases where the element tree only has one single vertex, there are two possibilities. If the tree has a growing point, we have $p=r_{\max}(B_{\phi}\boxconv\mu)$ and $\calQ_T$ being the identity function. In this case, inequality \eqref{ineq:genl_sineq} reduces to $\nu\prec_{r_{\max}(\nu)}  \nu$ with $\nu=B_{\phi}\boxconv\mu$, which directly follows from the definition of $\prec_p$. 
Otherwise, let $\mu_o$ denote the survey distribution on the only vertex. We have $p=-r_{\max}(\mu_{o})$ and $\calQ_T$ being the constant operator that returns $\mu_{o}$. In this case, inequality \eqref{ineq:genl_sineq} reduces to $ B_{\phi}\boxconv\mu_o\prec_{-r_{\max}(\mu_{o})} \mu_o,$  which can be proved using the rule of box convolution on $\beta$-curves (Proposition \ref{prop:beta_box}). 

Now 
we prove the induction step by considering three possible cases. Again, let $o$ denote the root vertex. We first prove the inequality when $\mu_o$ is trivial and the root has at most one child.  
Note that the base cases have already covered one-vertex trees. 
The root has exactly one incident edge in this regime. We denote this edge by $e$, and the sub element tree rooted at the unique child by $\tilde{T}$. 
The induction assumption implies
$$  B_{\phi}\boxconv\mathcal{Q}_{\tilde{T}}\mu\prec_{\tilde{p}}  \mathcal{Q}_{\tilde{T}}(B_{\phi}\boxconv\mu),$$
where $\tilde{p}=p(\tilde{T},r_{\max}(B_{\phi}\boxconv\mu))$. 
It is clear that $p=F_{\delta_e}(\tilde{p})$ and $\calQ_T\nu=B_{\delta_e}\boxconv \calQ_{\tilde{T}}\nu$ for any symmetric $\nu$. Therefore, the needed inequality in this case can be obtained by box convolving both sides of the above inequality with $B_{\delta_e}$, and the analysis follows from the commutativity of box convolution and the second property of Proposition \ref{prop:conv}.

For the second regime, we assume that $\mu_o$ is trivial, but the root has at least two children. 
We create a new element tree for each incident edge of the root,  by removing all other incident edges and their subtrees, and we denote them by $\tilde{T}_1,...,\tilde{T}_d$. 
Further, let $j^*$ be any index that minimizes $r_{\max}(\calQ_{T_j}
\mu)$. We create another element tree, denoted by $\hat{T}_{j^*}$, 
to be the sub element tree of $T$ that excludes everything from $\tilde{T}_{j^*}$, but not the root vertex. It is clear that   $\calQ_T\mu=(\calQ_{\tilde{T}_{j^*}}\mu)*(\calQ_{\hat{T}_{j^*}}\mu)$. 
Hence, the following is implied by Corollary \ref{coro:uniq_4}. 
\begin{align}\label{ineq:treeg_1}
    B_{\phi}\boxconv\mathcal{Q}_{T}\mu\prec_{p^*} (B_{\phi}\boxconv\mathcal{Q}_{\tilde{T}_{j^*}}\mu)*(B_{\phi}\boxconv\mathcal{Q}_{\hat{T}_{j^*}}\mu),
\end{align} 
where $p^*=|F_{\phi}(r_{\max}(\mathcal{Q}_{\hat{T}_{j^*}}\mu))-F_{\phi}(r_{\max}(\mathcal{Q}_{\tilde{T}_{j^*}}\mu))|$. Recall that 
both $\tilde{T}_{j^*}$ and $\hat{T}_{j^*}$ are not identical to the full tree in this regime. The induction assumptions can be applied to each,  and 
we have the following inequalities due to that convolution preserves degradation (see Proposition \ref{prop:convcom}). 
\begin{align}
   (B_{\phi}\boxconv\mathcal{Q}_{\tilde{T}_{j^*}}\mu)*(B_{\phi}\boxconv\mathcal{Q}_{\hat{T}_{j^*}}\mu)
&\preceq   (\mathcal{Q}_{\tilde{T}_{j^*}}(B_{\phi}\boxconv\mu))*(B_{\phi}\boxconv\mathcal{Q}_{\hat{T}_{j^*}}\mu)\label{ineq:gent_2}\\&\preceq(\mathcal{Q}_{\tilde{T}_{j^*}}(B_{\phi}\boxconv\mu))*(\mathcal{Q}_{\hat{T}_{j^*}}(B_{\phi}\boxconv\mu))\label{ineq:gent_3}\\
&=\mathcal{Q}_{T}(B_{\phi}\boxconv\mu).\nonumber
\end{align}  
Note that the above steps and inequality \eqref{ineq:treeg_1} exactly complete a full chain of degradation. Similar to the basic settings, 
we can prove the  overall 
strict inequality in 
$\beta$-curves 
by keeping track of the strict conditions for each individual step. 
We first look at 
inequality \eqref{ineq:gent_3}. By induction assumptions, we have 
\begin{align*}
    B_{\phi}\boxconv\mathcal{Q}_{\hat{T}_{j^*}}\mu&\prec_{\hat{p}}  \mathcal{Q}_{\hat{T}_{j^*}}(B_{\phi}\boxconv\mu),
\end{align*}
where 
$\hat{p}=p(\hat{T}_{j^*},r_{\max}(B_{\phi}\boxconv\mu))$. Therefore, by choosing $\ell=
r_{\max}(\mathcal{Q}_{\tilde{T}_{j^*}}(B_{\phi}\boxconv\mu))$ for the rule of convolution, inequality \eqref{ineq:gent_3} provides strict conditions for
\begin{align*}
    s\in(\hat{p}-r_{\max}(\mathcal{Q}_{\tilde{T}_{j^*}}(B_{\phi}\boxconv\mu)), r_{\max}(B_{\phi}\boxconv\mathcal{Q}_{\hat{T}_{j^*}}\mu)-r_{\max}(\mathcal{Q}_{\tilde{T}_{j^*}}(B_{\phi}\boxconv\mu))].
\end{align*} 
The lower boundary of the above set is no greater than $p$, because it is exactly the maximum polygon number over growing points that are within $\hat{T}_{j^*}$.
As mentioned earlier, inequality \eqref{ineq:treeg_1} already provides non-zero gaps for $s\in (p^*, r_{\max}( B_{\phi}\boxconv\mathcal{Q}_{T}\mu)]$. Given our choice of $j^*$, we also have $p^*= r_{\max}(B_{\phi}\boxconv\mathcal{Q}_{\hat{T}_{j^*}}\mu)-r_{\max}(B_{\phi}\boxconv\mathcal{Q}_{\tilde{T}_{j^*}}\mu)$.  Therefore, when $\tilde{T}_{j^*}$ is simple, we have $$\mathcal{Q}_{\tilde{T}_{j^*}}(B_{\phi}\boxconv\mu)=B_{\phi}\boxconv\mu=B_{\phi}\boxconv\mathcal{Q}_{\tilde{T}_{j^*}}\mu,$$ and we have already established strict inequality for $s\in (p, r_{\max}( B_{\phi}\boxconv\mathcal{Q}_{T}\mu)]$. 

When $\tilde{T}_{j^*}$ is not simple, we also need to investigate inequality \eqref{ineq:gent_2}. Similarly, by applying rule of convolution to the induction assumption, with $\ell=r_{\max}(B_{\phi}\boxconv\mathcal{Q}_{\hat{T}_{j^*}}\mu)$, and using the fact that $\beta$-curves are even functions for symmetric distributions, we obtain strict inequalities for 
\begin{align*}
    s\in(r_{\max}(B_{\phi}\boxconv\mathcal{Q}_{\hat{T}_{j^*}}\mu)-r_{\max}(\mathcal{Q}_{\tilde{T}_{j^*}}(B_{\phi}\boxconv\mu)), r_{\max}(B_{\phi}\boxconv\mathcal{Q}_{\hat{T}_{j^*}}\mu)-p(\tilde{T}_{j^*},r_{\max}(B_{\phi}\boxconv\mu))).
\end{align*} 
To proceed, we need to prove a simple fact: for any  $\tilde{T}_{j^*}$ being non-simple, we have \begin{align}\label{ineq_gent_simple}
    p(\tilde{T}_{j^*},r_{\max}(B_{\phi}\boxconv\mu))<r_{\max}(B_{\phi}\boxconv\mathcal{Q}_{\tilde{T}_{j^*}}\mu).
\end{align}  Let $v$ be a growing point in $\tilde{T}_{j^*}$ that maximizes the polygon number for $r=r_{\max}(B_{\phi}\boxconv\mu)$, and let $\delta_1,...,\delta_k$ denote the parameters of the edges on the path from the root to $v$. We define $p_0=r_{\max}(B_{\phi}\boxconv B_{\delta_1}\boxconv...\boxconv B_{\delta_k}\boxconv\mu)$. We have 
$p_0> p(\tilde{T}_{j^*},r_{\max}(B_{\phi}\boxconv\mu))$, because $p_0$ is exactly the output LLR message of the BP algorithm with the input at $v$ being $r$ and all other inputs setting to $0$ (including the surveys). Here we used the commutativity of box convolution. On the other hand, note that $B_{\delta_1}\boxconv...\boxconv B_{\delta_k}\boxconv\mu\preceq \mathcal{Q}_{\tilde{T}_{j^*}}\mu$, because the first quantity corresponds to a network with strictly less observation points. Hence, $p_0< r_{\max}(B_{\phi}\boxconv\mathcal{Q}_{\tilde{T}_{j^*}}\mu)$, and inequality (\ref{ineq_gent_simple}) is proved. This  implies that the upper boundary of the set obtained from inequality \eqref{ineq:gent_2}  is greater than $p^*$. Thus, the union of all three intervals covers the entire $(p, r_{\max}( B_{\phi}\boxconv\mathcal{Q}_{T}\mu)]$.


Having proved the induction step for any element tree with trivial $\mu_o$, we consider the third regime where $\mu_o$ can be any symmetric distribution. Let $\tilde{T}$ be the element tree obtained by removing the survey channel at the root from $T$. We have $\calQ_{T}\nu=\mu_o*\calQ_{\tilde{T}}\nu$ for any symmetric $\nu$, and we have already proved that $\calQ_{\tilde{T}}$.
  \begin{align}
      B_{\phi}\boxconv\mathcal{Q}_{\tilde{T}}\mu\prec_{\tilde{p}}  \mathcal{Q}_{\tilde{T}}(B_{\phi}\boxconv\mu),
\end{align}  
where $\tilde{p}=p(\tilde{T},r_{\max}(B_{\phi}\boxconv\mu))$. The above inequality leads to the following chain 
\begin{align*}
    B_{\phi}\boxconv\mathcal{Q}_{\tilde{T}}\mu&=B_{\phi}\boxconv(\mu_o*\mathcal{Q}_{\tilde{T}}\mu)\preceq
    \mu_o*(B_{\phi}\boxconv\mathcal{Q}_{\tilde{T}}\mu)\preceq \mu_o*\mathcal{Q}_{\tilde{T}}(B_{\phi}\boxconv\mu)= \mathcal{Q}_{{T}}(B_{\phi}\boxconv\mu).
\end{align*} 
 The first inequality is a direct application of Proposition \ref{prop:sur_basic}, and it states strict inequality in $\beta$-curves for $$s\in (F_{\phi}(r_{\max}(\mathcal{Q}_{\tilde{T}}\mu))-r_{\max}(\mu_o),r_{\max}( B_{\phi}\boxconv\mathcal{Q}_{\tilde{T}}\mu)].$$
 Apply rule of convolution  to the second inequality with $\ell=r_{\max}(\mu_o)$, we obtain strict conditions for
 $$s\in(\tilde{p}-r_{\max}(\mu_o),F_{\phi}(r_{\max}(\mathcal{Q}_{\tilde{T}}\mu)) -\ell ].$$
 By the definition of polygon number, the lower boundary of the above interval is exactly $p$. Hence, we have proved that the overall degradation implied by the above chain has strict inequalities in $\beta$-curves for all $s\in(p,r_{\max}( B_{\phi}\boxconv\mathcal{Q}_{\tilde{T}}\mu)]$.
\end{proof}
We present another important fact. 
\begin{prop}\label{prop:gent_ess}
For any distribution $P_T$, non-trivial symmetric point $\mu$, and any $\phi\in(0,\frac{1}{2})$, let $r=r_{\max}(B_{\phi}\boxconv\mu)$, we have
\begin{align}\label{ineq:gent_techm1}
    \esssup\, F_T(r)\geq r.
\end{align}
Moreover, the above inequality is strict if and only if $P_T$ is non-trivial.
\end{prop}
\begin{proof}
The stated inequality is a direct consequence from the third statement in Proposition \ref{prop:gen_t_ele}, by taking $r_{\max}$ on both sides 
and using the fact that $\mu$ is a BP fixed point. So, our proof will be focused on the strict condition. 
Specifically, it is clear that the equality condition holds for trivial $P_T$, because $F_T$ is the identity function w.p.1 under such condition. We only need to prove the converse, showing that the equality condition implies the trivialness of $P_T$.

For convenience, we define 
$$\alpha(s;\nu)\triangleq \beta\left(\tanh\frac{s}{2};\nu\right)-\tanh\frac{|s|}{2}$$ 
for symmetric $\nu$.
Our proof relies on the following property, which is proved in Appendix \ref{app:gent_ptecht}. 

\begin{prop}\label{prop:gent_ptt}
For any element tree $T$, symmetric $\nu$, and $s\in(0,r_{\max}(\nu))$, we have 
 \begin{align}\label{ineq:gent_teccot}
\alpha\left({F_T(s)};\calQ_T\nu\right)\leq \tanh\frac{F_T(s)}{2} \cdot
\alpha\left({s};\nu\right)/\tanh\frac{s}{2} \, ,
\end{align}
where $\calQ_T$ denotes the BP operator for  the fixed  $T$. The equality condition holds if and only if $T$ 
has no more than one growing point and all survey distributions are trivial. 
\end{prop} 
Assume the equality condition of inequality \eqref{ineq:gent_techm1}.  We have
$F_T(r)\leq r$ w.p.1.  
Note that 
$\alpha(s,\nu)$ 
is non-increasing for $s\geq 0$. This condition implies  
$$\alpha\left({r};\calQ_T\mu\right)\leq \alpha\left({F_T(r)};\calQ_T\mu\right).$$ 
Recall that $\mu$ is non-trivial and $\phi\in(0,\frac{1}{2})$. We have $r=F_{\phi}(r_{\max}(\mu))\in(0,r_{\max}(\mu))$. 
 Hence, one can apply Proposition \ref{prop:gent_ptt} to the above inequality and obtain \begin{align}\label{ineq:702}
    \alpha\left({r};\calQ_T\mu\right)\leq\tanh\frac{F_T(r)}{2} \cdot
\alpha\left({r};\mu\right)/\tanh\frac{r}{2}\leq \alpha\left({r};\mu\right). 
\end{align}
Because $\mu$ is a BP fixed point, two sides of the above inequalities are identical in expectation. Therefore, 
we need all equality conditions to hold in probability.  
First, $r\in(0,r_{\max}(\nu))$ implies $\alpha\left({r};\mu\right)>0$. We must have $F_T(s)=s$ w.p.1 for the second step in inequality \eqref{ineq:702}. Then, combined with the equality condition stated in   Proposition \ref{prop:gent_ptt}, the tree $T$ must be simple w.p.1, implying that $P_T$ is trivial.  
\end{proof}

Now we use 
these results to prove Theorem \ref{thm:gen_t_tec}. 
\begin{proof}
The first statement relies on the linearity of $\beta$-curves and Proposition \ref{prop:gen_ele_ft}. As proved earlier in this appendix (recall Proposition  \ref{pp:gen59} and Proposition \ref{prop:gen_ele_ft}), any element tree $T$ that satisfies the polygon condition with parameter $r=r_{\max}(B_{\phi}\boxconv\mu)$ implies 
$$    B_{\phi}\boxconv\mathcal{Q}_{T}\mu\prec_p  \mathcal{Q}_{T}(B_{\phi}\boxconv\mu)$$
for some $p=p(T,r)<0$. From the fourth property of Proposition \ref{prop:bounded_crit}, we essentially have
$$    B_{\phi}\boxconv\mathcal{Q}_{T}\mu\prec  \mathcal{Q}_{T}(B_{\phi}\boxconv\mu).$$
Given this relation, any $T$ that is representative implies strict inequalities in $\beta$-curves 
for all $s\in[0,r]$. By linearity, the same non-zero gap condition holds for two sides of inequality (\ref{ineq:gen_t_1}). Hence, the needed strict degradation 
is implied by the fact that $r_{\max}(B_{\phi}\boxconv\mathcal{Q}_{P_T}\mu)$ is exactly $r$, as $\mu$ is a BP fixed point.    

~

To prove that random self-concatenation leads to representative trees with non-zero probability, it suffices to find a subset of element trees with non-zero measure on $P_T$, such that all self-concatenations within this subset is representative for some fixed $h$. 
Recall that $P_T$ is not simple. According to Proposition \ref{prop:gent_ess}, we can find $\tilde{r}>r$ such that  $\mathbb{P}[F_T(r)>\tilde{r}]>0$. Hence, let $\mathcal{T}$ be the subset of element trees 
satisfying $F_T(r)>\tilde{r}$, the set $\mathcal{T}$ has non-zero measure.

Now we show that there is a fixed $h$ such that all self-concatenations within $\mathcal{T}$ with $h$ steps lead to element trees satisfying the polygon condition. 
Assume the contrary, by Definition \ref{def:polc},  for each sufficiently large $h$ there is a concatenation $T_{(h)}$ of trees from $\mathcal{T}$ and a growing point $v$ at the leaf such that the polygon number of $v$ is non-negative. 
 We utilize the following proposition, proved in Appendix \ref{app:last}, which essentially provides an upper bound of polygon numbers using $F_T(r)$.
\begin{prop}\label{prop:last}
Consider any element tree with $k$ growing points and a list of non-negative numbers $r_0, r_1,...,r_k$. Let $s$ denote the essential supremum of the output LLR message for the BP algorithm with $k$ inputs at the growing points given by $r_1,r_2,...,r_k$, and let $p$ denote the essential infimum 
with $k$ inputs given by   $r_0,-r_2,-r_3,...,-r_k$. If $p\geq 0$, then $p+s\leq r_0+r_1$.  
\end{prop}

Let $L$ denote the path from the root to $v$, and let $T_1,...,T_h$ denote the element trees that overlaps with $L$ in the sequential order that the root belongs to $T_1$ and $v$ belongs to $T_h$. Let $o_{j-1}$ denote the root of $T_j$, and let $o_h$ denote $v$. We consider the BP algorithm on $T_{(h)}$ where the input is $r$ for $v$ and $-r$ for all other growing points. Let $p_j$ denote essential infimum of the message sent from vertex $o_j$. Because all inputs except for the one at $v$ are negative and $p_0$ is non-negative in the end, 
the essential infimum for all messages sent from vertices on $L$ must be non-negative, i.e., $p_0,...,p_h\geq 0$. 

Now we use Proposition \ref{prop:last} to derive a recursive relation between these quantities. Consider each $T_j$, let $p_j$ be the input at its growing point on $L$ (i.e., $o_j$). For each other growing point, let their input be the essential infimum of their  message sent in the above process. Then the essential infimum of the message sent from $o_{j-1}$ is exactly $p_{j-1}$. 
We view the input at $o_{j}$ to be $r_0$, and the inputs at all other growing points to be $-r_2,-r_3,...,-r_k$. Then the variable $p$ in Proposition \ref{prop:last} is identical to $p_{j-1}$.
Consider the other set of input configurations with $r_1=r$. Because of the monotonicity of BP algorithms and $F_{T}(r)>r$ for all element trees in $\mathcal{T}$, all $r_2,...,r_k$ are lower bounded by $r$. Again, due to monotonicity, we have $s\geq \tilde{r}$. Therefore, Proposition \ref{prop:last} states that $p_{j-1}\leq p_j-(\tilde{r}-r)$.

Because $p_h=r$ is finite, and $\tilde{r}-r$ is a fixed positive number given $\mathcal{T}$. There can be at most finitely many $h$ for the recursion to holds with $p_0\geq 0$, which leads to a contradiction.   
\end{proof}

 

\subsection{Proof of Proposition \ref{prop:gent_ptt}}\label{app:gent_ptecht}

Note that any BP operator is a composition of convolution and box convolution. We prove the inequality by keeping track of all the composition steps. We first state the 
specialized statements for each building block in the following proposition. 
\begin{prop}\label{prop_gent_basic_a}
Consider any symmetric $\nu$ and  $s>0$.  
\begin{enumerate}
    \item For any $\delta\in[0,1]$, we have 
        \begin{align}\label{ineq:gent_teccot_box}
    \alpha(F_{\delta}(s); B_{\delta}\boxconv\nu )=  \left(\tanh\frac{F_{\delta}(s)}{2}\right) \cdot \frac{\alpha(s; \nu)}{ \tanh\frac{s}{2}}.
    \end{align}
\item For any symmetric $\tau$ and $r\geq 0$, 
we have
   \begin{align}\label{ineq:gent_teccot_sur}
    \alpha(s+r;\nu*\tau)\leq \left(\tanh\frac{s+r}{2}\right)\cdot \frac{\alpha(s; \nu)+\alpha(r;\tau)}{\tanh\frac{s}{2}+\tanh\frac{r}{2}},
\end{align}
where the equality condition holds if and only if either  $r=0$ and $\tau$ is trivial, or $\alpha(s; \nu)=\alpha(r;\tau)=0$. 
\end{enumerate}
\end{prop}
\begin{proof}
The first property follows from the rule of box convolution on $\beta$-curves (Proposition \ref{prop:beta_box}). To prove the second property, consider independent variables $R_{\nu}\sim\nu$ and $R_{\tau}\sim\tau$. We have 
$$   \alpha(s+r;\nu*\tau)=\mathbb{E}\left[\max\left\{\tanh\frac{|R_{\nu}+R_{\tau}|}{2}-\tanh\frac{s+r}{2},0\right\}\right].$$
By symmetry condition, the above quantity can be written as
$$   \alpha(s+r;\nu*\tau)=\mathbb{E}[f(|R_{\nu}|,|R_{\tau}|)],$$
where
\begin{align*}
    f(|R_{\nu}|,|R_{\tau}|)\triangleq\begin{cases}
    \tanh\max\left\{\frac{|R_{\nu}|}{2},\frac{|R_{\tau}|}{2}\right\}-\tanh\frac{s+r}{2} & \textup{if } s+r< ||R_{\nu}|-|R_{\tau}||\\
    \frac{1}{2}\left(\tanh\frac{|R_{\nu}|+|R_{\tau}|}{2}-\tanh\frac{s+r}{2}\right)\left({1+\tanh\frac{|R_{\nu}|}{2}\tanh\frac{|R_{\tau}|}{2}}\right) & \textup{if } ||R_{\nu}|-|R_{\tau}||\leq s+r< |R_{\nu}|+|R_{\tau}|\\
    0& \textup{otherwise}
    \end{cases}.
\end{align*}
Note that the RHS of inequality \eqref{ineq:gent_teccot_sur} can be written similarly as the expectation of the following function 
\begin{align*}
    g(|R_{\nu}|,|R_{\tau}|)\triangleq\left(\tanh\frac{s+r}{2}\right)\cdot\left( \frac{\tanh\max\left\{\frac{|R_{\nu}|}{2},\frac{s}{2}\right\}+\tanh\max\left\{\frac{|R_{\tau}|}{2},\frac{r}{2}\right\}}{\tanh\frac{s}{2}+\tanh\frac{r}{2}}-1 \right).
\end{align*}
By elementary calculus, we have 
\begin{align*}
    f(|R_{\nu}|,|R_{\tau}|)\leq  g(|R_{\nu}|,|R_{\tau}|).
\end{align*} 
Then by taking the expectation, we have proved inequality \eqref{ineq:gent_teccot_sur}.

When equality holds, we need $f(|R_{\nu}|,|R_{\tau}|)=  g(|R_{\nu}|,|R_{\tau}|)$ w.p.1. 
We prove that either one of the equality conditions  stated in Proposition \ref{prop_gent_basic_a} is needed for this statement. 
When $s\in(0,r_{\max}(\nu))$, 
we have $\mathbb{P}[|R_{\nu}|>s]>0$. Then for any fixed $|R_{\nu}|>s$, the equality $f(|R_{\nu}|,|R_{\tau}|)=  g(|R_{\nu}|,|R_{\tau}|)$ is equivalent to $r=0$ and trivial $\tau$.
Consequently, the first condition is needed in this regime. 

For the other case, we have $s\geq r_{\max}(\nu)$, which implies $\alpha(s;\nu)=0$. By symmetry, one can prove that inequality \eqref{ineq:gent_teccot_sur} is strict for $r<r_{\max}(\tau)$, given that we assumed $s>0$. Hence, in this case, equality also requires $r\geq r_{\max}(\tau)$, which implies $\alpha(r;\tau)=0$. 
To summarize, the second equality condition is needed in this regime.
\end{proof}

\begin{proof}[Proof of Proposition \ref{prop:gent_ptt}]
If $T$ has no growing points, then  $\calQ_T$ returns a fixed distribution and $F_T(s)=r_{\max}(\calQ_T\nu)$. Hence, the LHS of inequality (\ref{ineq:gent_teccot}) is zero. Then, because $s<r_{\max}(\nu)$, we have $\alpha{(s;\nu)}$ being positive, and the inequality is proved. The equality condition in this case is given by $F_{T}(s)=0$. Recall that we assumed all edges to have non-zero capacity. This condition is equivalent to having all survey distributions be trivial, which is the same condition stated in the proposition.

In the rest of this proof, we assume $T$ has at least one growing point. Again, we prove the stated inequality by induction over the size of the tree and the number of non-trivial survey distributions. 
Consider the base case where $T$ only has one vertex. It has to be a single growing point.  Therefore, we have $F_T(s)=s$ and $\calQ_T$ being the identity function. Hence, the equality condition holds.

For the induction step, $T$ has at least two vertices. We can choose one incident edge $e$ of the root vertex, such that the corresponding subtree contains at least one growing point. We consider two possible cases. First,  
if the root has only one incident edge and $\mu_o$ is trivial, we let $\tilde{T}$ denote the subtree obtained by removing the the root and edge $e$. The BP operator can be written as $\calQ_T \nu=  B_{\delta_e}\boxconv \calQ_{\tilde{T}}\nu$   and $F_{T}(s)= F_{\delta_e}(F_{\tilde{T}}(s))$. Recall that $\tilde{T}$ has at least one growing point and all edges have non-zero capacity. We have $F_{\tilde{T}}(s)>0$ for any $s>0$. Hence, the rule of box convolution on $\beta$-curves implies
    \begin{align*}
    \alpha(F_{T}(s); \calQ_T\nu )=\alpha(F_{\delta}(F_{\tilde{T}}(s)); B_{\delta}\boxconv\calQ_{\tilde{T}}\nu )=  \left(\tanh\frac{F_{T}(s)}{2}\right) \cdot \frac{\alpha(F_{\tilde{T}}(s); \calQ_{\tilde{T}}\nu)}{ \tanh\frac{F_{\tilde{T}}(s)}{2}}.
    \end{align*}
By applying the induction assumption on $\tilde{T}$, the RHS of the above inequality is exactly upper bounded by the RHS of inequality \eqref{ineq:gent_teccot}. Furthermore, the equality holds if and only if $\tilde{T}$ has exactly one growing point and all survey distributions are trivial. This is consistent with the statement in Proposition \ref{prop:gent_ptt}.  

For the other case, let $T_1$ denote the subtree consists of $\tilde{T}$, $e$, and a root with a trivial survey distribution; let $T_2$ denote the subtree that contains all other contents of $T$. We have $\calQ_T\nu=\calQ_{T_1}\nu*\calQ_{T_2}\nu$ and $F_{T}(s)=F_{T_1}(s)+F_{T_2}(s)$.  Because $T_1$ contains at least one growing point, we have $F_{T_1}(s)>0$. Thus, we have the following expression from the second statement of Proposition \ref{prop_gent_basic_a}. 
    \begin{align*}
    \alpha(F_{T}(s); \calQ_T\nu )=\alpha(F_{T_1}(s)+F_{T_2}(s); \calQ_{T_1}\nu*\calQ_{T_2}\nu)\leq   \left(\tanh\frac{F_{T}(s)}{2}\right) \cdot \frac{\alpha(F_{T_1}(s); \calQ_{T_1}\nu)+\alpha(F_{T_2}(s);\calQ_{T_2}\nu)}{ \tanh\frac{F_{T_1}(s)}{2}+\tanh\frac{F_{T_2}(s)}{2}}.
    \end{align*}
Note that both $T_1$ and $T_2$ contain either at least one less vertex or one less non-trivial survey distribution compared to $T$. Hence, induction assumptions can be applied, which exactly gives inequality \eqref{ineq:gent_teccot}. On the other hand, recall that $\alpha(s;\nu)>0$ and $F_{T}(s)>0$. The overall equality requires $\alpha(F_{T_1}(s); \calQ_{T_1}\nu)+\alpha(F_{T_2}(s);\calQ_{T_2}\nu)>0$. Hence, to satisfy the equality condition of Proposition \ref{prop_gent_basic_a}, we need $F_{T_2}(s)=0$ and $\calQ_{T_2}\nu$ is trivial. Because $s$ is positive, these conditions imply that $T_2$ has no growing points, and all its survey distributions are trivial. Again, because $F_{T}(s)>0$, overall equality requires the equality condition for the $T_1$ induction. Therefore, $T_1$ can have at most one growing point, and all its  survey distributions are trivial as well. 
Combining the requirements on $T_1$ and $T_2$, the element tree $T$ must have have at most one growing point and all trivial survey distributions. 





\end{proof}


\subsection{Proof of Proposition \ref{prop:last}} \label{app:last}
\begin{proof}
Let $v_0$ denote the root, and $L=(v_0,v_1,v_2,...,v_\ell)$ denote the path from the root to the first growing point. When the inputs are given by $r_1,r_2,r_3,...,r_k$, 
let 
$q_j$ be the essential supremum of the summation of all messages received at $v_j$ excluding the one received from the path. Because messages are degraded when passing through each edge, the final output is upper bounded by the summation of all external messages received on the path, i.e.,
$$s\leq r_1+ q_0+...+q_{k-1} .$$

When the inputs are given by $r_0,-r_2,-r_3,...,-r_k$, 
let $p_j$ denote the essential infimum of the message sent from $r_j$. Because $p\geq 0$ and all input messages except for the first growing point are non-positive. We must have $p_0,...,p_k\geq 0$. Thus, 
the essential infimum of the message received at $r_j$ from the path is no greater than $p_{j+1}$. By symmetry, the essential infimum of the summation of all other received messages at $r_j$ is $-q_j$. Hence, we have $p_j\leq p_{j+1}-q_j$, which implies that 
$$p\leq r_0-q_0-...-q_{k-1}.$$

Combine these two bounds, we have $p+s\leq r_0+r_1$.
\end{proof}

\bibliographystyle{IEEEtran}
\bibliography{references}
\end{document}